\documentclass[12pt]{article}

\usepackage[margin=1in]{geometry}
\usepackage{setspace}
\usepackage{enumitem}
\usepackage{mathtools}
\usepackage{amsthm,amsmath}
\usepackage{dsfont}
\usepackage{amssymb}
\usepackage{dsfont}
\usepackage{bm,comment}
\usepackage{breqn}
\usepackage{color}
\usepackage{xcolor}
\usepackage{subcaption}
\usepackage{pgfplots,pgfplotstable}
\pgfplotsset{compat=1.11} 
\usepackage{tikz}
\usetikzlibrary{patterns,decorations.pathreplacing,math}
\usepackage[authoryear]{natbib}
\usepackage[normalem]{ulem}

\usepackage[hidelinks,colorlinks=true]{hyperref}%\usepackage[hidelinks]{hyperref}
\definecolor{darkblue}{rgb}{0.0, 0.0, 0.55}
\definecolor{teal}{rgb}{0.0, 0.5, 0.5}
\hypersetup{linkcolor=teal, urlcolor=blue, citecolor=darkblue}

\usepackage{algorithmic}
\usepackage[boxruled,linesnumbered]{algorithm2e}
\usepackage[flushleft]{threeparttable}
\usepackage{multibib}
\newcites{app}{References}
\usepackage{ulem}

\allowdisplaybreaks

\newtheorem{lemma}{Lemma}
\newtheorem{theorem}{Theorem}
\newtheorem{assumption}{Assumption}
\newtheorem{example}{Example}

\newtheoremstyle{named}{}{}{\itshape}{}{\bfseries}{.}{.5em}{#1\thmnote{ #3}}
\theoremstyle{named}

\newcommand{\esum}{n^{-1}\sum_{i=1}^n}
\newcommand{\C}{\mathcal{C}}
\newcommand{\E}{\mathbb{E}}

\newcommand{\M}{\mathcal{M}}
\newcommand{\N}{\mathbb{N}}
\newcommand{\mbbP}{\mathbb{P}}
\newcommand{\R}{\mathbb{R}}
\newcommand{\calS}{\mathcal{S}}

\newcommand{\argmin}{\textup{argmin}}
\newcommand{\argmax}{\textup{argmax}}
\newcommand{\poly}{\textup{poly}}
\newcommand{\ppoly}{\textup{ppoly}}
\newcommand{\rk}{\textup{rk}}

\begin{document}

\title{Testing Inequalities Linear in Nuisance Parameters}
\author{Gregory Fletcher Cox\footnote{Department of Economics, National University of Singapore (\href{mailto:ecsgfc@nus.edu.sg}{ecsgfc@nus.edu.sg})} 
\and Xiaoxia Shi\footnote{Department of Economics, University of Wisconsin-Madison (\href{mailto:xshi@ssc.wisc.edu}{xshi@ssc.wisc.edu})}
\and Yuya Shimizu\footnote{Department of Economics, University of Wisconsin-Madison (\href{mailto:yuya.shimizu@wisc.edu}{yuya.shimizu@wisc.edu})}}
\date{\today}

\maketitle

\begin{abstract}
This paper proposes a new test for inequalities that are linear in possibly partially identified nuisance parameters. This type of hypothesis arises in a broad set of problems, including subvector inference for linear unconditional moment (in)equality models, specification testing of such models, and inference for parameters bounded by linear programs. The new test uses a two-step test statistic and a chi-squared critical value with data-dependent degrees of freedom that can be calculated by an elementary formula. Its simple structure and tuning-parameter-free implementation make it attractive for practical use. We establish uniform asymptotic validity of the test, demonstrate its finite-sample size and power in simulations, and illustrate its use in an empirical application that analyzes women's labor supply in response to a welfare policy reform. 
\end{abstract}

{\bf Keywords:}  Linear Program, Moment Inequalities, Quadratic Programming, Subvector Inference, Uniform Inference

\pagebreak

\section{Introduction}

We propose a simple new test for hypotheses of the form $H_0$: there exists a $\delta$ such that $C\delta\leq b$, where elements of the Jacobian matrix $C$ and the intercept vector $b$ are reduced-form parameters that can be consistently estimated, and elements of $\delta$ are unknown parameters whose values are partially identified by the inequalities under $H_0$. Since the inequalities, rather than $\delta$, are of central interest, $\delta$ is a nuisance parameter vector. Hypotheses of this form arise in specification testing and subvector inference for linear unconditional moment (in)equality models and in inference for parameters bounded by linear programs, including discrete instrumental variable (IV) models with shape restrictions and policy relevant treatment effect models. These models have wide applications in empirical work. We explain the applications and give examples in Section \ref{sec:setup}. 

Testing this hypothesis is non-standard both because the nuisance parameter $\delta$ may not be point-identified and because the hypothesis involves inequalities. As a result, commonly used test statistics have non-standard asymptotic distributions involving parameters that cannot be consistently estimated, in particular, the local slackness of the inequalities evaluated at the true value of $\delta$. 
This complicates the design of critical values. 
A common approach is to simulate the asymptotic distribution with a conservative estimator of the local slackness plugged in. However, the conservative estimators typically involve user-chosen tuning parameters that introduce arbitrariness to the procedure. Moreover, the simulated critical values can be computationally burdensome. 

In the special case that the Jacobian is {\em known}, \cite{CoxShi2023}, hereafter CS23, propose the subvector conditional chi-squared (sCC) test that does not require simulation or user-chosen tuning parameters and yet still has uniform asymptotic size control and good power. 
The simplicity of the test is achieved by considering the conditional distribution of the quasi-likelihood ratio (QLR) statistic given the identity of the active inequalities.\footnote{An inequality is active if it holds with equality at our null-imposed estimator of $\delta$.} The conditional distribution is shown to be bounded by a chi-squared distribution with degrees of freedom (DoF) dependent on the conditioning event. CS23 recommends computing the DoF by solving a sequence of linear programming problems. 

Our first contribution is to derive an elementary formula for the DoF that replaces the recommendation in CS23. The formula makes the computation of the critical value elementary. Implementing the sCC test is now no harder than calculating the test statistic, which is a convex quadratic programming problem (CQPP). 
The formula also reveals an intuitive interpretation of the sCC test: The sCC test turns out to be the same as the classic Sargan-Hansen's J test for a moment equality model, where the equalities are determined by the active inequalities. 

While the sCC test is attractive, the known Jacobian assumption significantly restricts its applicability. 
In both linear unconditional moment inequality models and models with a parameter of interest bounded by linear programs, a known Jacobian only applies in special cases. 
The known Jacobian assumption even rules out conditional moment inequality models where $\delta$ includes coefficients on endogenous covariates. 
The second contribution of this paper is to propose a new test, called the generalized conditional chi-squared (GCC) test, that accounts for the estimation error of an {\em unknown} Jacobian while maintaining the simplicity, size, and power properties of the sCC test.

Estimating the Jacobian presents two challenges. 
The first challenge is finding the limit of the constraint set in the definition of the QLR statistic. 
In general, convergence of the matrix of coefficients in a system of inequalities is insufficient for the set defined by those inequalities to converge in a setwise sense. 
We show that a simple stable rank condition on the Jacobian is sufficient for convergence of the constraint set. 
The stable rank condition requires the rank of certain submatrices of the Jacobian to not change in the limit. 
The submatrices that need to satisfy the condition are minimal in some sense. 
They are associated with collections of inequalities that (implicitly or explicitly) define equality restrictions in the limit. 
While the stable rank condition is not innocuous, it relaxes the commonly used strong identification assumption in moment equality models, which requires the Jacobian to be full rank.%; see Section \ref{Assumption1Discussion}. 

The second challenge is finding a consistent variance estimator. 
The effect of the estimated Jacobian on the variance depends on $\delta$, but $\delta$ is only partially identified. This makes the additional variance term difficult to account for. Our solution is to use a first-stage estimator of $\delta$ that converges to a point in the identified set for $\delta$. The procedure resembles optimal weighting in two-step generalized method of moments (GMM). 
The GCC test compares this two-step statistic to a chi-squared critical value with DoF determined by the formula from our first contribution. 
We also define a refinement of the GCC test in the supplemental appendix that has slightly more power. 

We next review three strands of related literature. 
The first consists of papers that propose tests of moment inequalities that are possibly nonlinear in nuisance parameters. 
These include \cite{BugniCanayShi2015, BugniCanayShi2017}, \cite{CCT2018}, \cite{BelloniBugniChernozhukov2018}, \cite{KMS2019}, and \cite{Bei2023}. 
These methods use critical values that are nontrivial to compute and require user-chosen tuning parameters to be adaptive to the unknown slackness of the inequalities.\footnote{An exception is Procedure 3 in \cite{CCT2018} in that it does not require any tuning parameter or simulation. We include this procedure in the simulations.} In contrast, the GCC test uses an algebraic critical value and is adaptive to the unknown slackness of the inequalities without any tuning parameter, at the price of requiring linearity. 

The second strand of related literature is composed of CS23 and \cite{AndrewsRothPakes2023}, which consider conditional moment inequalities that are linear in the nuisance parameters. 
These papers rely on an idea that if the Jacobian depends only on random variables on which the inequalities hold conditionally, then the Jacobian can be treated as known. 
However, this idea does not apply if the moment inequalities are unconditional or, more generally, if the Jacobian depends on a random variable for which the moment inequalities do not hold conditionally. Our method applies regardless of whether the Jacobian can be treated as known.

When our test is applied to inference for a scalar parameter (say $\theta$) bounded by linear programs, it is related to the third strand of literature. This literature addresses the problem of inference for the value of a linear programming problem (LPP). The papers include  \cite{FreybergerHorowitz2015}, \cite{BaiSantosShaikh2022}, \cite{FSST2023}, \cite{ChoRussell2024}, \cite{Gafarov2025}, \cite{Voronin2025}, and \cite{GoffMbakop2025}, among others. There is a subtle technical difference between our setting and this literature: while our test is inverted to yield a confidence interval for $\theta$, the LPP literature aims at constructing confidence intervals for the upper or lower bound for $\theta$ defined by the value of a LPP. While the two problems are distinct, they are closely related. Our confidence interval by design can cover either bound with correct nominal coverage probability (asymptotically), and in the LPP literature, one-sided confidence intervals  of the appropriate direction for either bound are also valid confidence intervals for $\theta$.\footnote{In the LPP literature, two-sided confidence intervals for $\theta$ can be obtained by combining two one-sided confidence intervals via a Bonferroni adjustment.} Thus, the methods can be used for the same empirical problems. Notably, our method is the only one that is tuning parameter and simulation free. 

We include a simple one-sided specification in the simulations in order to compare the GCC test to representative papers in all three strands of the literature. In addition to the simple simulation, we also evaluate the GCC test in two realistic simulation examples: one of an interval outcome instrumental variables (IV) model as in \cite{GandhiLuShi2023}, hereafter GLS23, and the other of bounds on policy relevant treatment effects as in \cite{MogstadSantosTorgovitsky2018}, hereafter MST18.  The simulations show that the GCC test is computationally very fast with good size and power. 

We also implement the GCC test in an empirical analysis of female labor supply in response to a welfare policy reform. \cite{KlineTartari2016}, hereafter KT16, estimate bounds on the treatment responses by manually eliminating the nuisance parameters from revealed preference inequalities. The GCC test provides uniformly valid inference for the treatment responses based directly on the revealed preference inequalities.
Overall, we find statistically significant heterogeneous responses to the policy change, which agree with the results in KT16.

The rest of the paper is organized as follows. Section \ref{sec:setup} describes the setup and applications. Section \ref{sec:CC_Def} defines the GCC test. Section \ref{sec:momineq} describes the theoretical properties of the GCC test. Section \ref{sec:mc} presents the simulations. Section \ref{sec:empirical_KT} presents the empirical illustration. Section \ref{sec:conclusion} concludes. An appendix contains proofs of the theorems, while a supplemental appendix includes additional results, proofs, simulations, and discussion. 

\section{Setup and Examples}\label{sec:setup}

We are interested in testing the hypothesis 
\begin{equation}
H_0:C\delta\leq b \text{ for some }\delta\in\mathbb{R}^{d_\delta}, \label{first_specification}
\end{equation}
where $C$ is a $d_C\times d_\delta$ matrix of reduced-form parameters, $b$ is a $d_C$-dimensional vector of reduced-form parameters, $d_\delta$ is the dimension of $\delta$, and $d_C$ is the number of inequalities. 
To make an invertibility assumption imposed later as unrestrictive as possible, we add some structure to $C$ and $b$. 
We assume that 
\begin{equation}
C=B\Pi+D ~ \text{ and } ~ b=d-B\mu,\label{structure}
\end{equation} 
where $B$ is a known $d_C\times d_\mu$ matrix, $D$ is a known $d_C\times d_\delta$ matrix, $d$ is a known $d_C$-dimensional vector, $\Pi$ is an unknown $d_\mu\times d_\delta$ matrix of reduced-form parameters, and $\mu$ is an unknown $d_\mu$-dimensional vector of reduced-form parameters. 
This structure separates the unknown and estimated components from the known components in the Jacobian $C$ and the intercept $b$. 
It is satisfied in all the examples considered below. 
Typically, $B$ has more rows than columns, and it absorbs the linear dependence across rows for the estimation noise of the inequalities. 
This allows us to accommodate inequalities with linearly dependent estimation errors, which arise when we write an equality as a pair of opposing inequalities, when the model contains a deterministic constraint such as a shape or sign restriction, or when the law of total probability dictates that a weighted sum of the inequalities involves no unknown quantities under $H_0$. 

With the structure in (\ref{structure}), $H_0$ can be equivalently written as:
\begin{align}
H_0:B(\mu+\Pi\delta)+D\delta \leq d \text{ for some }\delta\in \R^{d_\delta}. \label{generalH0}
\end{align}
Let $\overline{\mu}_n$ and $\overline{\Pi}_n$ be estimators of $\mu$ and $\Pi$. 
Let $\overline{C}_n=D+B\overline{\Pi}_n$ and $\bar b_n=d-B\overline{\mu}_n$. 
In the next two subsections, we describe two classes of models that are covered by this framework.

\subsection{Moment (In)equality Models}\label{sub:momineq}
Moment (in)equality models are used to address data or modeling incompleteness issues, including missing data, multiple equilibria, and large intractable games.\footnote{For empirical applications and current statistical methods for such models, see the survey papers by \cite{CanayShaikh2017}, \cite{HoRosen2017}, and \cite{Molinari2020}.} 
Here, we show that both specification testing and subvector inference for moment inequality models fit the hypothesis in equation (\ref{generalH0}) when the moments are linear in the parameters. 

Consider the moment (in)equality model:
\begin{equation}
\E[\overline{m}_n^{eq}(\beta)] = \mathbf{0} \text{ and }
\E[\overline{m}_n^{ineq}(\beta)]\geq \mathbf{0} \text{ with }\beta\in {\cal B}\subseteq \R^{d_\beta},\label{mi}
\end{equation}
where $\overline{m}_n(\beta) = (\overline{m}^{eq}_n(\beta)',\overline{m}_n^{ineq}(\beta)')'$ is a $\R^{d_m}$-valued sample moment function, $\beta$ is an unknown vector of parameters, and ${\cal B}$ is its parameter space. Specifically, let $\overline{m}_n(\beta) = n^{-1}\sum_{i=1}^n m(W_i,\beta)$, where $\{W_i\}_{i=1}^n$ is a sample of observable variables and $m(\cdot, \beta)$ is a function known up to the unknown parameter $\beta$. Let the number of equalities be denoted $d_{eq}$ and the number of inequalities be denoted $d_{ineq}$, so that $d_m = d_{eq}+d_{ineq}$. 

Suppose $\overline{m}_n(\beta)$ is linear in $\beta$. That is, $\overline{m}_n(\beta) = \overline{\Gamma}_n\beta + \overline{\eta}_n$ for $\overline{\Gamma}_n = \partial\overline{m}_n(\beta)/\partial\beta'$ and $\overline{\eta}_n = \overline{m}_n(\mathbf{0})$. 
Suppose ${\cal B} = \R^{d_\beta}$.\footnote{More generally, if ${\cal B}$ is a polyhedral set, then the deterministic inequalities that define ${\cal B}$ should be included when writing the hypothesis in the form of (\ref{generalH0}). We show how deterministic constraints can be incorporated in the next subsection.}
Consider the following types of problems: 
\begin{enumerate}
\item Specification Testing. 
When specification testing, one evaluates whether there exists a $\beta\in\R^{d_\beta}$ such that the moment (in)equalities hold. If not, then the model is misspecified. 
The hypothesis is 
\[
H_0: \E[\overline{m}_n^{eq}(\beta)] = \mathbf{0}\text{ and }\E[\overline{m}_n^{ineq}(\beta)]\geq \mathbf{0} \text{ for some }\beta\in \R^{d_\beta}. 
\]
This is the type of hypothesis considered in \cite{BugniCanayShi2015}. It can be written in the form of (\ref{generalH0}) with
\[
B = \left(\begin{smallmatrix}-I_{d_{eq}}&\mathbb{O}_{d_{eq}\times d_{ineq}}\\I_{d_{eq}}&\mathbb{O}_{d_{eq}\times d_{ineq}}\\\mathbb{O}_{d_{ineq}\times d_{eq}}&-I_{d_{ineq}}\end{smallmatrix}\right),~\mu = \E[\overline{\eta}_n], ~\Pi = \E[\overline{\Gamma}_n], ~\delta = \beta, ~D=\mathbb{O}_{(d_{eq}+d_m)\times d_\beta}, ~d = \mathbf{0},
\]
where $I_a$ is an identity matrix of size $a$ and $\mathbb{O}_{a\times b}$ is a $a\times b$ zero matrix. Note that we write the equalities as pairs of opposing inequalities via the first $2d_{eq}$ rows of $B$. 
\item Subvector Inference. In subvector inference, one constructs a confidence set for a subvector of the parameters. Suppose the subvector of interest is composed of the first $\ell$ elements of $\beta$ and denote it by $\theta$. 
Let $\delta$ denote the rest of the elements in $\beta$. 
Then a confidence set for $\theta$ can be constructed by testing the following hypothesis at each value of $\theta$ and collecting the values of $\theta$ at which the hypothesis is not rejected:
\[
H_0: \E[\overline{m}^{eq}(\theta,\delta)] = \mathbf{0}\text{ and }\E[\overline{m}^{ineq}(\theta,\delta)]\geq \mathbf{0} \text{ for some }\delta \in \R^{d_\delta}.
\]
This is the type of hypotheses considered in \cite{BugniCanayShi2017}, and it can be written in the form of (\ref{generalH0}) with
\[
B = \left(\begin{smallmatrix}-I_{d_{eq}}&\mathbb{O}_{d_{eq}\times d_{ineq}}\\I_{d_{eq}}&\mathbb{O}_{d_{eq}\times d_{ineq}}\\\mathbb{O}_{d_{ineq}\times d_{eq}}&-I_{d_{ineq}}\end{smallmatrix}\right),~\mu = \E[\overline{m}_n(\theta,\mathbf{0})], ~\Pi = \E[\overline{\Gamma}_n^{\delta}], ~D=\mathbb{O}_{(d_{eq}+d_m)\times (d_\beta-\ell)}, ~d = \mathbf{0},
\]
where $\overline{\Gamma}_n^{\delta} = \partial\overline{m}_n(\theta,\delta)/\partial\delta'$, which is the last $d_\beta -\ell$ columns of $\overline{\Gamma}_n$.
\end{enumerate}
Two remarks are in order regarding subvector inference: 

\noindent{\bf Remarks:}
(1) {\it Linearity in $\theta$ is not needed for subvector inference. The discussion remains unchanged if $\overline{m}_n(\theta,\delta)$ is linear in $\delta$ and nonlinear in $\theta$. }

(2) {\it \label{rem:orthonormalize} Sometimes, the parameter of interest is not a subvector of $\beta$, but instead a linear function of $\beta$. 
That is, $\theta = \Lambda\beta$ for a known $d_\theta\times d_\beta$ full-rank matrix $\Lambda$ with $d_\theta<d_\beta$. One approach is to 
reparameterize $\beta$ so that $\theta$ becomes a subvector of the new parameter. 
Let $\Lambda^c$ be a $(d_\beta-d_\theta) \times d_\beta$ row-augmenting matrix so that $\left(\begin{smallmatrix}\Lambda\\\Lambda^c\end{smallmatrix}\right)$ is nonsingular.\footnote{In Matlab,  one can find such a matrix by applying the function \texttt{null}(~) on $\Lambda$.} 
The reparameterization is given by $\gamma = \left(\begin{smallmatrix}\Lambda\\\Lambda^c\end{smallmatrix}\right)\beta$. 
Then, plugging $\beta = \left(\begin{smallmatrix}\Lambda\\\Lambda^c\end{smallmatrix}\right)^{-1}\gamma$ into \textup{(\ref{mi})} reparameterizes the model so that $\theta$ is a subvector of $\gamma$. Equivalently, one can add $\theta= \Lambda\beta$ to the model as deterministic constraints and treat $\beta$ as the nuisance parameter. }

We end this subsection with two examples of linear moment inequality models.

\begin{example}[Interval Outcome IV Regression]\label{ex:interval}  Consider a linear model $Y^\ast = X'\beta+\varepsilon$ with $\E[\varepsilon Z]=\mathbf{0}$, 
where $Z$ is a vector of instruments. 
The dependent variable $Y^\ast$ is not observed. 
Instead, we observe $Y^L$ and $Y^U$ that satisfy: $\E[Y^LZ]\leq \E[Y^\ast Z]\leq \E[Y^U Z]$. 
Then, we have the following unconditional moment inequalities:
\begin{equation}
\E[Y^L Z -ZX'\beta]\leq \mathbf{0}\text{ and }\E[ZX'\beta-Y^UZ]\leq \mathbf{0}.\label{uncmi}
\end{equation}
This is an example of \textup{(\ref{mi})}.  The interval outcome IV regression model was proposed in \textup{\cite{ManskiTamer2002}}. A generalization of such a model to a non-standard aggregate demand estimation problem is studied in \textup{GLS23}. In their generalization, $Y^L$ and $Y^U$ can be nonlinear functions of the parameter of interest. 

{\textup{CS23} and \textup{\cite{AndrewsRothPakes2023}} cover a related model where the inequalities in \textup{(\ref{uncmi})} hold conditionally on $Z$. Their tests apply to hypotheses that fix the coefficients on all the endogenous regressors. 
Then, the Jacobian of the inequalities with respect to the nuisance parameter is known after conditioning on $Z$ since it is not a function of the endogenous regressors. Their tests do not apply when a nuisance parameter is a coefficient on an endogenous regressor.} 
\end{example}

\begin{example}[Panel Data Multinomial Choice Model]\label{ex:multi} Consider a panel data multinomial choice model where individual $i$ at time $t$ obtains utility $u_{ijt}$ from choosing option $j$.  Let $y_{ijt}=1$ if $i$ chooses $j$ at time $t$, and $y_{ijt}=0$ otherwise. The random utility model stipulates that $y_{ijt}=1$ if and only if $u_{ijt}\geq u_{ij't}\text{ for all }j'\in \{0,1,2,\dots,J\}$. 
Consider the linear index model of the random utility: $u_{ijt} = X_{ijt}'\gamma+\lambda_{ij}+\varepsilon_{ijt}$, where $X_{ijt}$ is a vector of observed covariates, $\lambda_{ij}$ is an unobserved fixed effect, and $\varepsilon_{ijt}$ is an idiosyncratic taste shock. Normalize $X_{i0t}=\mathbf{0}$. For illustration, let there be only two time periods ($t=1,2$) and let the individuals be independent and identically distributed. Under a conditional time homogeneity assumption on $\varepsilon_{ijt}$, \cite{ShiShumSong2018} show that the following moment inequality holds: 
\begin{equation}
\E[\Delta y_i' \Delta X_{i}\gamma | X_{i}]\geq 0, \label{pmn0}
\end{equation}
where $\Delta y_i$ is a $J$-dimensional vector with its $j$th element being $y_{ij2}-y_{ij1}$, $\Delta X_{i}$ is a $J\times d_x$ dimensional matrix with its $j$th row being $(X_{ij2}-X_{ij1})'$, and $X_{i}$ collects $X_{ijt}$ for all $j\in \{1,2,\dots,J\}$ and $t\in\{1,2\}$. 
Note that none of the elements of $\Delta y_i' \Delta X_{i}$ can be considered exogenous because they depend on $y_{ijt}$. Thus, the inequalities do not fit into the conditional moment inequality setup in CS23 or  \cite{AndrewsRothPakes2023}. 
Let ${\cal I}(X_i)$ be a non-negative vector-valued instrumental function. Then, 
\begin{equation}
\E[{\cal I}(X_i)\Delta y_i' \Delta X_{i}\gamma]\geq \mathbf{0}\label{pmn}
\end{equation}
rewrites the inequalities in \textup{(\ref{pmn0})} into the form of \textup{(\ref{mi})}.\footnote{In this model, a normalization is usually imposed on $\gamma$, such as the first element being one, that can be accommodated by simply setting that element to 1.} 
\end{example}

\subsection{Parameters Bounded by Linear Programming}\label{sub:bound_lp}
Recently, an important class of models have arisen in the structural estimation literature where a scalar parameter of interest is not point-identified but bounded by the values of linear programming problems (LPPs).\footnote{Some notable examples appear in KT16, MST18, \cite{KKLS2021}, and \cite{SyrgkanisTamerZiani2021}.} 
The constraint sets for the LPPs are defined by linear (in)equalities, where the constants and coefficients in the (in)equalities are unknown parameters to be estimated. 
To fix ideas, suppose the parameter of interest is 
\begin{equation}
\theta = \gamma' \delta,\label{target}
\end{equation}
where $\gamma\in\R^{d_\delta}$ is a known vector that defines a linear combination of the nuisance parameters $\delta$. 
The constraint set is defined by a collection of (in)equalities: 
\begin{equation}
\Gamma\delta = m ~\text{ and }~ A\delta\leq b,\label{deltaID}
\end{equation}
where the elements of $m\in \R^{d_{\Gamma}}$ and $\Gamma\in\R^{d_\Gamma\times d_\delta}$ are known or can be consistently estimated,  $A$ is a known $d_A\times d_\delta$ matrix, and $b\in\R^{d_A}$ is a known vector.\footnote{We focus on the case $A$ and $b$ are known because it is common in applications. Unknown and estimated $A$ and/or $b$ can also be covered.} 
The known inequalities in (\ref{deltaID}) define a parameter space for $\delta$. 
For example, if $\delta$ is a vector of weights, then $\delta$ takes values in the simplex, which can be represented by an appropriate choice of $A$ and $b$. 

Inference on $\theta$ can be based on the GCC test for the existence of a value of $\delta$ that simultaneously satisfies (\ref{target}) and (\ref{deltaID}) at a hypothesized value of $\theta$. 
A confidence interval for $\theta$ can be calculated by inverting a family of tests. 
The restrictions in (\ref{target}) and (\ref{deltaID}) can be written in the form of (\ref{generalH0}) with 
\begin{align}
B = \left(\begin{smallmatrix} \mathbf{0}'_{d_{\Gamma}}\\
	\mathbf{0}'_{d_{\Gamma}}\\
         I_{d_\Gamma}\\
       - I_{d_{\Gamma}}\\
         \mathbb{O}_{d_A\times d_{\Gamma}}\end{smallmatrix}\right), 
~\mu =-m,~\Pi =\Gamma, ~
D = \left(\begin{smallmatrix}
\gamma'\\
-\gamma'\\
\mathbb{O}_{d_{\Gamma}\times d_\delta}\\
\mathbb{O}_{d_{\Gamma}\times d_\delta}
\\A
\end{smallmatrix}\right),  \text{ and }
d = \left(\begin{smallmatrix}\theta\\-\theta\\\mathbf{0}_{d_{\Gamma}}\\\mathbf{0}_{d_{\Gamma}}\\b\end{smallmatrix}\right). \label{known_gamma_translation}
\end{align}
Note that the first two rows and the last $d_A$ rows of $B$ are zeros to accommodate the deterministic constraints $\theta = \gamma\delta$ and $A\delta\leq b$. 

Another approach to inference on $\theta$ is to use LPPs. 
From the point of view of identification, $\theta$ is bounded sharply by $\theta_\min = \min_{\delta:~\text{(\ref{deltaID}) holds}}\gamma'\delta$ and $\theta_\max = \max_{\delta:~\text{(\ref{deltaID}) holds}}\gamma'\delta$. 
Then one can construct one-sided confidence intervals that are bounded from below (above) for $\theta_{\min}$ ($\theta_{\max}$) and use them as confidence intervals for $\theta$. 
Inference for the value of a LPP is generally based on plugging the estimators of $m$ and $\Gamma$ into the LPPs and simulating or bootstrapping the asymptotic distributions of these estimators. 
However, the asymptotic distributions depend on which corner or face of the constraint set solves the LPP, which is not smooth as a function of the estimated reduced-form parameters. 
Thus, the na\"{i}ve strategy of bootstrapping the value of a LPP is generally invalid. 
In order to obtain valid inference, the LPP literature recommends various modifications that involve tuning parameters and/or simulating/bootstrapping nonstandard distributions. 
Using the GCC test to obtain a confidence interval for $\theta$ avoids these complications. 

\noindent{\bf Remarks:} (1){\it
The GCC test is valid for any hypothesized value of $\theta$ in $[\theta_\min, \theta_\max]$, including the endpoints. Thus, the GCC test is a valid way to do inference for the value of a LPP. However, the GCC test may direct power in a one-sided way. In particular, when $\theta_\max-\theta_\min$ is not small, the GCC test is effectively one-sided for the value of the LPP (though still two-sided for $\theta$). Therefore, when the parameter of interest is $\theta_\min$ or $\theta_\max$ instead of $\theta$ and the researcher desires two-sided inference, then  some of the inference recommendations from the LPP literature are preferred. }

(2) \label{rem:known_gamma}{\it
In some applications, $\gamma$ is unknown and estimated. 
There is a convenient way to write \textup{(\ref{target})} and \textup{(\ref{deltaID})} in the form of \textup{(\ref{generalH0})}. 
The idea is to add an element to $\mu$ that is always zero, while adding $\gamma'$ as a row of $\Pi$.
Specifically, \textup{(\ref{generalH0})} can be satisfied when $\gamma$ is unknown by letting 
\begin{equation}
B = \left(\begin{smallmatrix}1&\mathbf{0}'_{d_{\Gamma}}\\
                                         -1&\mathbf{0}'_{d_{\Gamma}}\\
                                          \mathbf{0}_{d_{\Gamma}}&I_{d_{\Gamma}}\\
                                         \mathbf{0}_{d_{\Gamma}}&-I_{d_{\Gamma}}\\
                                         \mathbf{0}_{d_A}&\mathbb{O}_{d_A\times d_{\Gamma}}\end{smallmatrix}\right), ~\mu = \left(\begin{smallmatrix}0\\ -m\end{smallmatrix}\right),~
\Pi = \left(\begin{smallmatrix}\gamma'\\
\Gamma\end{smallmatrix}\right), ~
D = \left(\begin{smallmatrix}
\mathbf{0}'_{d_\delta}\\
\mathbf{0}'_{d_\delta}\\
\mathbb{O}_{d_{\Gamma}\times d_\delta}\\
\mathbb{O}_{d_{\Gamma}\times d_\delta}
\\A
\end{smallmatrix}\right), \text{ and }
d = \left(\begin{smallmatrix}\theta\\-\theta\\\mathbf{0}_{d_{\Gamma}}\\\mathbf{0}_{d_{\Gamma}}\\b\end{smallmatrix}\right).\label{LP1}
\end{equation}
More generally, if one of the equations in \textup{(\ref{deltaID})} has a known intercept but the corresponding row of the Jacobian is unknown, then the intercept should be included in $d$ while the row of the Jacobian should be included in $\Pi$, possibly including a zero in $\mu$ and augmenting the columns of $B$.\footnote{This idea can be applied more generally. Consider a generalization of (2), where $C$ can be written as $C=B_1\Pi_1+B_2\Pi_2+D$ and $b=d-B_1\mu_1$. That is, $B_2\Pi_2$ is a component of $C$ that needs to be estimated but cannot be written as a linear combination of the columns of $B_1$. Then, we can satisfy the structure in (2) by taking $B=[B_1,B_2]$ and $\mu=(\mu'_1,\mathbf{0}')'$. This works because we do not require the estimator of $\mu$ to have a nonsingular variance matrix, but we only require the estimator of $\mu+\Pi\delta$ to have a nonsingular variance matrix for a value of $\delta$ in the identified set; see Assumption \ref{Assumption1}(iii), below.}  
This demonstrates the flexibility of the specification in \textup{(\ref{generalH0})}. }

We demonstrate the relevance of the setup in (\ref{target}) and (\ref{deltaID}) with examples.

\begin{example}[Discrete IV Regression with Shape Restrictions]\label{ex:npiv}
IV regressions with discrete regressors and instruments are common in  practice. Prominent examples include \cite{PermuttHebel1989}, \cite{AngristKrueger1991}, and
\cite{Angrist1998}. 
\cite{FreybergerHorowitz2015} consider an IV model with discrete $X_i$ and $Z_i$: 
\begin{equation}
Y_i = \delta(X_i)+\varepsilon_i,~~~
\E[\varepsilon_i|Z_i]=0,\label{npiv_setup}
\end{equation} 
where $Y_i$ is the dependent variable, $X_i$ a discrete endogenous regressor, $Z_i$ is a discrete instrument, $\delta(\cdot)$ an unknown function that represents the structural relationship between $X_i$ and $Y_i$, and $\varepsilon_i$ is the error term. 
While linearity of $\delta(\cdot)$ is often assumed, \cite{FreybergerHorowitz2015} emphasize that no such functional form restrictions are needed. 

Let the support of $X_i$ and $Z_i$ be $\{x_1,\dots,x_{d_x}\}$ and $\{z_1,\dots,z_{d_z}\}$, respectively. 
Let $\delta_k = \delta(x_k)$ for $k\in\{1,\dots,d_x\}$, and $\delta= (\delta_1,\dots,\delta_{d_x})'$. 
The model in \textup{(\ref{npiv_setup})} implies that $\delta$ satisfies the equalities in \textup{(\ref{deltaID})} with 
\begin{equation}
m = \left(\begin{smallmatrix}\E[Y_i1\{Z_i=z_1\}]\\ \E[Y_i1\{Z_i=z_2\}]\\ \vdots\\ \E[Y_i1\{Z_i=z_{c_z}\}]\end{smallmatrix}\right) \text{ and } 
\Gamma =\left(\begin{smallmatrix}\mbbP(X_i =x_1,Z_i=z_1)&\mbbP(X_i =x_2,Z_i=z_1)&\cdots&\mbbP(X_i =x_{d_x},Z_i=z_1)\\
\mbbP(X_i =x_1,Z_i=z_2)&\mbbP(X_i =x_2,Z_i=z_2)&\cdots&\mbbP(X_i =x_{d_x},Z_i=z_2)\\
\vdots&\vdots&\ddots&\vdots\\
\mbbP(X_i =x_1,Z_i=z_{d_z})&\mbbP(X_i =x_2,Z_i=z_{d_z})&\cdots&\mbbP(X_i =x_{d_x},Z_i=z_{d_z})
\end{smallmatrix}\right). \label{npiv_mom}
\end{equation}
Note that $m$ and $\Gamma$ are reduced-form parameters that can be estimated by sample averages. 

When $d_z<d_x$, $\delta$ is not point identified by the equalities in \textup{(\ref{deltaID})}. 
To sharpen identification, \cite{FreybergerHorowitz2015} add shape restrictions of the form $A\delta\le b$ for some known matrix $A$ and known vector $b$. 
This covers several types of shape restrictions including monotonicity and/or convexity of $\delta(\cdot)$. 
The parameter of interest is typically a linear function of $\delta$. 
For example, $\theta = [-1,1,0,...,0]\delta=\delta_2-\delta_1$ is the effect of changing $X$ from $x_1$ to $x_2$. Thus, inference for the structural function in \textup{(\ref{npiv_setup})} falls into the framework of \textup{(\ref{target})-(\ref{deltaID})}. 
\end{example}

\begin{example}[Policy Relevant Treatment Effects (PRTE)]\label{ex:PRTE} 
Treatment effects that are relevant for policy are often not equal to the local average treatment effects (LATEs) associated with any available instrument. In a standard program evaluation model, they are weighted averages of underlying marginal treatment responses (MTRs), where the weights are identified or known, but the MTRs are not. MST18 show that the MTRs can be partially identified from the LATEs, or more generally from IV-like estimands, because these estimands are weighted averages of the MTRs. Then, bounds on a PRTE can be deduced from the identified set of the MTRs. Under a parameterization of the MTRs, MST18 show that the bounds are values of LPPs.\footnote{In some cases, MST18 show that even the nonparametric bounds can be written as finite-dimensional LPPs; see their Proposition 4.} 

To be specific, consider an outcome $Y$, a binary treatment indicator $D$, and covariates $Z=(X,Z_0)$, where $X$ is a vector of control variables and $Z_0$ is the vector of excluded instruments. 
Let $Y_0$ and $Y_1$ denote the potential outcomes corresponding to the two treatment arms. 
Suppose treatment is determined by a weakly separable selection equation: $D=\mathds{1}\{p(Z)\ge U\}$ for some unobserved uniformly distributed variable $U$, where $p(Z)$ is the propensity score. 
The MTR functions are defined to be 
\begin{equation}
\kappa_0(u,x)=\mathbb{E}[Y_0|U=u, X=x] \text{ and }\kappa_1(u,x)=\mathbb{E}[Y_1|U=u, X=x]. 
\end{equation}
A wide range of PRTEs can be written as weighted averages of the MTRs. 
The MTRs can be parameterized as a linear combination of functions belonging to some basis. 
Let $\delta$ be a vector of coefficients on the basis functions. 
Then, a PRTE, say $\theta$, can be written as $\theta=\gamma'\delta$, where $\gamma$ is a vector of weighted averages applied to each basis function. 
This writes the parameter of interest in the form of \textup{(\ref{target})}. 

An IV-like estimand is a parameter of the form $m_s=\mathbb{E}[s(D,Z)Y]$ for some identified or known function $s(D,Z)$. 
MST18 show that every IV-like estimand can be written as a weighted average of the MTRs with a simple formula for the weights. 
This means that each $m_s$ can be written as a linear combination of $\delta$. 
If $m$ denotes a vector of finitely many IV-like estimands, then $m$ satisfies the equalities in \textup{(\ref{deltaID})} with each row of $\Gamma$ being the weighted average of the basis functions with weights corresponding to the IV-like estimand. 

In addition, MST18 allow the researcher to specify additional shape restrictions on the MTEs, in a similar manner as Example \ref{ex:npiv}. 
Depending on the choice of basis functions, these can sometimes be written as deterministic linear inequalities on $\delta$. 
Overall, this shows that $\theta$ is bounded by LPPs and satisfies the structure in \textup{(\ref{target})} and \textup{(\ref{deltaID})}. 
We demonstrate the GCC test in a simulation of this example in Appendix \ref{app:MTE}. 
\end{example}

\begin{example}[State Transition Probabilities]\label{ex:labor} 
Consider a model where individuals choose between finitely many states, $s\in\mathcal{S}$. 
A policy change may induce individuals to choose a different state. 
The state transition probabilities are defined by 
\begin{equation}
\delta_{s, s'}=\Pr(S_a=s'|S_b=s) \text{ for }s, s'\in\mathcal{S}, 
\end{equation}
where $S_a$ denotes an individual's choice after a policy change and $S_b$ denotes an individual's choice before a policy change. 
These transition probabilities represent the fraction of individuals who start in state $s$ and change to state $s'$ in response to the policy change. 

The transition probabilities are not identified from the data if all one has is a repeated cross-section of individuals before and after the treatment, or a cross-section of individuals randomly assigned to different policy regimes. On the other hand, the marginal probabilities of $S_a$ and $S_b$, denoted $p_a$ and $p_b$, respectively, are identified. Then, the transition probabilities are partially identified through their relationship with $p_a$ and $p_b$: 
\begin{equation}
p_a=\Delta p_b, \label{conditional_marginal_relation}
\end{equation} 
where $\Delta$ is the matrix of $\delta_{s,s'}$ values for $s,s'\in\mathcal{S}$. These equations fit the structure of the equalities in \textup{(\ref{deltaID})}, with elements of $\Delta$ forming the nuisance parameter vector $\delta$.\footnote{One of the equations in (\ref{conditional_marginal_relation}) is redundant because the sum of the elements in $p_a$ is one. This is not a problem.} 

In addition, the transition probabilities satisfy: 
\begin{equation}
\delta_{s,s'}\ge 0 \text{ for all } s, s'\in\mathcal{S} \text{ and } \sum_{s'\in\mathcal{S}}\delta_{s,s'}=1 \text{ for all } s\in\mathcal{S}. \label{s_inequalities}
\end{equation}
These restrictions can be written as the inequalities in \textup{(\ref{deltaID})} with known $A$ and $b$. 
Then, inference for a particular transition probability fits \textup{(\ref{target})} and \textup{(\ref{deltaID})}. 

KT16 use such a model to study women's labor supply in response to a welfare policy change. 
In their model, the transition probabilities represent the fraction of women who gain/lose employment and/or register for welfare. 
KT16 analyze the data from an experiment that randomized exposure to a new welfare policy. 
The policy change could have heterogeneous effects on labor supply through the extensive margin (encouraging some women to gain employment) and the intensive margin (encouraging some women to decrease their hours or wages in order to qualify for welfare). 
The experiment identifies the distribution of employment, welfare participation, and income for women under two welfare policies. 
KT16 point out that the details of the policy change, combined with weak assumptions on the utility functions of the individuals, restrict many of the transition probabilities to zero. 
They then manually solve for bounds on each of the remaining transition probabilities from \textup{(\ref{conditional_marginal_relation})} and \textup{(\ref{s_inequalities})} and find significant labor market effects along both the extensive and intensive margins. 
In Section \ref{sec:empirical_KT}, we employ the GCC test to construct confidence intervals for this example, avoiding the need to solve for the bounds manually.
\end{example}

\section{The Generalized Conditional Chi-Squared Test}
\label{sec:CC_Def}

In this section, we define the generalized conditional chi-squared (GCC) test. 
The test depends on $\overline{\mu}_n$ and $\overline{\Pi}_n$, consistent and asymptotically normal estimators of $\mu$ and $\Pi$. 
The test also depends on $\overline{\Omega}_n$, a consistent estimator of the asymptotic variance of $(\overline{\mu}'_n,\textup{vec}(\overline{\Pi}_n)')'$, where $\textup{vec}(A)$ denotes the vectorization of a matrix, $A$. 

We start with preliminary $H_0$-restricted estimators for $\mu$ and $\delta$:
\begin{equation}
(\widetilde\mu_n,\widetilde\delta_n)=\underset{\mu,\delta: B\mu+\overline{C}_n\delta\le d}{\argmin} n(\overline{\mu}_n-\mu)'\widehat\Upsilon_n(\overline{\mu}_n-\mu), \label{tildemudef}
\end{equation}
where $\widehat\Upsilon_n$ is a preliminary weight matrix that is converging in probability to a deterministic positive definite limit.\footnote{This is similar to the weight matrix used in the first step of two-step GMM; the limit theory is invariant to the choice of weight matrix used. We take $\widehat{\Upsilon}_n$ to be the identity matrix.} 
While $\widetilde\mu_n$ is always unique, $\widetilde\delta_n$ need not be. 
In that case, we can take $\widetilde\delta_n$ to be the minimizer that has the smallest norm: 
\begin{equation}
\widetilde\delta_n=\underset{\delta: B\widetilde\mu_n+\overline{C}_n\delta\le d}{\argmin} \|\delta\|. \label{tildedeltadef}
\end{equation}
Equation (\ref{tildedeltadef}) is a tie-breaking procedure that is only used if the value of $\widetilde\delta_n$ that minimizes (\ref{tildemudef}) is not unique.\footnote{Equation (\ref{tildedeltadef}) uses the Euclidean norm, although it could be replaced with any other norm.} 
Later, we add assumptions so that $\widetilde\delta_n$ is consistent for its population analogue, $\delta^\ast_F$.\footnote{We formally define $\delta^\ast_F$ in (\ref{deltanF}), below. For now, it is enough to think of $\delta^\ast_F$ as the probability limit of $\widetilde\delta_n$.} 

We next estimate the asymptotic variance of $\overline{\mu}_n+\overline{\Pi}_n\delta^\ast_F$ using 
\begin{equation}
\widetilde{\Sigma}_n=\left(\begin{matrix} I_{d_\mu}\\ \widetilde{\delta}_n\otimes I_{d_\mu}\end{matrix}\right)'\overline{\Omega}_n \left(\begin{matrix} I_{d_\mu}\\ \widetilde{\delta}_n\otimes I_{d_\mu}\end{matrix}\right), 
\end{equation}
where $\otimes$ denotes the Kronecker product. Define the QLR test statistic as 
\begin{equation}
T_n=\min_{\mu,\delta: B\mu+\overline{C}_n\delta\le d} n(\overline{\mu}_n-\mu)'\widetilde\Sigma^{-1}_n(\overline{\mu}_n-\mu). \label{QLRnew}
\end{equation}
Let $(\widehat\mu_n,\widehat\delta_n)$ solve the minimization problem in (\ref{QLRnew}). 
Similar to the initial estimators, $\widehat\mu_n$ is always unique and $\widehat\delta_n$ may not be unique. 
In that case, we can take $\widehat\delta_n$ to be any minimizer.\footnote{Note the subtlety in the definitions of $\widetilde\delta_n$ and $\widehat\delta_n$: $\widetilde\delta_n$ is required to minimize (\ref{tildedeltadef}) because it has to be consistent for $\delta^\ast_F$, while $\widehat\delta_n$ can be arbitrary because its consistency is not essential.} 

For any $K\subseteq\{1,...,d_C\}$, let  $|K|$ denote the cardinality of $K$, and let $I_K$ denote the submatrix of the $d_C\times d_C$ identity matrix formed by taking the rows corresponding to the indices in $K$. 
In this way, conformable premultiplication of $I_K$ to a matrix $B$  selects the rows of $B$ corresponding to the indices in $K$ to form the submatrix of $B$ with $|K|$ rows.

We are ready to define the DoF and the critical value. 
Let $\widehat b_n=d-B\widehat\mu_n$ and $\widehat{K}= \{j\in\{1,\dots,d_C\}: e'_j\overline{C}_n\widehat\delta_n= e'_j\widehat b_n\}$, where $e_j$ is the $j$th standard normal basis vector. Note that $\widehat{K}$ denotes the set of indices at which the inequality constraint holds as equality for the minimizers in (\ref{QLRnew}). 
Let 
\begin{equation}
\widehat{s}_n = \textup{rk}\left(I_{\widehat K}\left[B,D\right]\right)-\textup{rk}(I_{\widehat K}\overline{C}_n). \label{rhatformula}
\end{equation}
For a significance level $\alpha$, let $cv(s,\alpha)$ denote the $1-\alpha$ quantile of the $\chi^2$ distribution with DoF equal to $s$. 
The GCC test rejects if $T_n$ is greater than $cv(\widehat s_n,\alpha)$. 

We end this section with some remarks on the definition of the GCC test. 

\noindent\textbf{Remarks: }{\it 
\textup{(1)} Equation \textup{(\ref{rhatformula})} gives an algebraic formula for the DoF. 
In CS23, the DoF, $\widehat{r}_n$, is defined as the dimension of the span of a polyhedral cone. 
Theorem \textup{\ref{lem:rhat}}, below, shows that $\widehat{r}_n$ is equal to $\widehat s_n$ with probability one in the limit. 
For computational reasons, we recommend $\widehat s_n$ over $\widehat r_n$. 

\textup{(2)} The GCC test is in some sense a ``na\"ive'' test. 
It is equivalent to first selecting the inequalities that are active according to the finite sample CQPP in} (\ref{QLRnew}){\it, pretending that these inequalities are equalities, and then forming a Sargan-Hansen's $J$-test for overidentification of a model defined by these equalities. 
To see this, note that the typical case is when $\rk(I_{\widehat K}[B,D])=|\widehat K|$ and $\rk(I_{\widehat K}C)=d_\delta$. 
Then, the DoF used by the GCC test is $|\widehat K|-d_\delta$, which represents the number of active inequalities minus the number of nuisance parameters. 

\textup{(3)} The weight matrix, $\widetilde{\Sigma}_n$, used in the definition of the GCC test is an estimator of the asymptotic variance of $\sqrt{n}(\overline{\mu}_n+ \overline\Pi_{n}\delta^\ast_F)$ instead of that of $\sqrt{n}\overline{\mu}_n $. The latter is used in CS23 for the sCC test. 
The new weight matrix accounts for the estimation error in $\overline{\Pi}_n$. To gain some  intuition for this weight matrix,  note that $T_n$ can be rewritten as
\begin{align}
{T}_n &= \min_{\eta,\gamma:B\eta+\overline{C}_n\gamma\leq h_n}(X_n-\eta)'\widetilde{\Sigma}_n^{-1}(X_n-\eta), \label{Tnrep}
\end{align}
where $\gamma = \sqrt{n}(\delta-\delta^\ast_F)$, $\eta = \sqrt{n}(\mu-\mu_F+(\overline{\Pi}_n-\Pi_F)\delta_F^\ast)$,  $X_n = \sqrt{n}(\overline{\mu}_n - \mu_F +(\overline{\Pi}_n-\Pi_F)\delta_F^\ast)$, and $h_n = \sqrt{n}(d-C_F\delta_F^\ast-B\mu_F)$, where $\mu_F$, $\Pi_F$, and $C_F$ stand for the true values of $\mu$, $\Pi$, and $C$. With this change of variables, one can see that $\widetilde{\Sigma}_n$ estimates the asymptotic variance of $X_n$. 

\textup{(4)} The GCC test is very easy to compute since it only requires solving two CQPPs. Efficient interior-point algorithms for CQPPs are available in most commonly used software, and they are known to have a worst-case computational complexity of $O((d_C+d_\delta)^4)$, where $d_C$ is the number of inequalities and $d_\delta$ is the dimension of the nuisance parameter. This is only slightly slower than the computational complexity of LPPs, which is $O((d_C+d_\delta)^{3.5})$; see \cite{Karmarkar1984} and \cite{YeTse1989}. Most importantly, no simulation or bootstrap is needed to perform the test.}

\section{Theoretical Properties}\label{sec:momineq}

In this section, we present the three main theoretical results of the paper: (1) a theorem that justifies using $\widehat s_n$ and hence simplifies the rank calculation, (2) a theorem that shows the consistency of $\widetilde\delta_n$ for $\delta^\ast_F$, and (3) a theorem that shows the uniform asymptotic validity of the GCC test. 

\subsection{Rank Calculation Theorem}

We now present the result that justifies using $\widehat s_n$. 
One can also define the DoF in the GCC test using the Karush-Kuhn-Tucker (KKT) multipliers. 
Let $\widehat{L}=\{j\in\{1,...,d_C\}: e'_j\widehat\psi_n>0\}$, where $\widehat\psi_n$ is a vector of nonnegative multipliers that satisfy the KKT conditions for (\ref{QLRnew}). 
Then 
\begin{equation}
\widehat{t}_n = \textup{rk}\left(I_{\widehat L}\left[B,D\right]\right)-\textup{rk}(I_{\widehat L}\overline{C}_n) \label{rtildeformula}
\end{equation}
is another way to define the DoF. 
Also note that 
\begin{equation}
\widehat r_n=\textup{dim}(B'\{h\geq \mathbf{0}: h'\overline{C}_n=\mathbf{0}, h'(B\widehat{\mu}_n - d) = 0\})	\label{r_n}
\end{equation}
is the definition of the DoF for the sCC test in CS23, where $\textup{dim}(\cdot)$ denotes the dimension of a set, or the maximum number of linearly independent elements. 
The following theorem shows that $\widehat s_n$, $\widehat t_n$, and $\widehat r_n$ are equal with probability one in the limit. 
This theorem plays a vital role in the proof of the uniform asymptotic validity of the GCC test, below. 

\begin{theorem}\label{lem:rhat} Suppose $\widetilde{\Sigma}_n$ is positive definite. 
Then,
\textup{(a)} $\widehat t_n\le\widehat r_n\le \widehat s_n$. 

\textup{(b)} For fixed $\overline{C}_n$ and $\widetilde\Sigma_n$, there is a Lebesgue measure zero subset of $\R^{d_\mu}$, ${\cal M}_0$, such that
\begin{equation}
\widehat r_n=\widehat s_n=\widehat t_n, \text{ unless }\overline{\mu}_n\in{\cal M}_0. \label{r_equality}
\end{equation}
\end{theorem}

\noindent\textbf{Remarks:} 
(1) {\it Theorem \textup{\ref{lem:rhat}} justifies the use of $\widehat s_n$ as the DoF. 
This overrides CS23, which recommends calculating $\widehat r_n$ using an algorithm that includes a series of LPPs. 
The new recommendation applies to both the sCC test in CS23 and to the GCC test.} 

(2) {\it Theorem \textup{\ref{lem:rhat}} is not random---it does not rely on the distribution of $\overline{\mu}_n$, $\overline{C}_n$, or $\widetilde\Sigma_n$. 
The result is a general feature of CQPPs. 
Part \textup{(a)} shows that, regardless of the distribution, $\widehat s_n$ is (weakly) more conservative than $\widehat r_n$. 
Part \textup{(b)} shows that equality holds with probability one if the conditional distribution of $\overline{\mu}_n$ given $\overline{C}_n$ and $\widetilde\Sigma_n$ is absolutely continuous. 
A key case where this holds is in the limit, where $\overline{C}_n$ and $\widetilde\Sigma_n$ are deterministic and $\overline{\mu}_n$ is Gaussian. 
Thus, a simple corollary of Theorem \ref{lem:rhat} is that $\widehat r_n=\widehat s_n$ with probability one in the limit. }

(3) {\it The expression for $\mathcal{M}_0$ can be found in the Supplemental Appendix, equation \textup{(\ref{eq:M0})}. 
The value of $\M_0$ may depend on the value of $\overline{C}_n$ or $\widetilde\Sigma_n$ that is fixed. 
Note that $\widehat{\delta}_n$ (and $\widehat{K}$) may not be unique for the definition of $\widehat s_n$, and $\widehat\psi_n$ (and $\widehat{L}$) may not be unique for the definition of $\widehat t_n$. 
When they are not unique, $\mathcal{M}_0$ does not depend on the choice of $\widehat{\delta}_n$ or $\widehat \psi_n$. }

(4) {\it In general, $\M_0$ is not the empty set. 
To clarify the necessity of $\M_0$ in part \textup{(b)}, we give a simple example to show that $\widehat r_n<\widehat s_n$ is possible, albeit on a set of measure zero. 
Suppose $d_\mu=d_\delta=1$ and $d_C=2$. 
Let $B=(0,1)'$, $\overline{C}_n=D=(1,1)'$, and $d=(0,0)'$.  
If $\overline{\mu}_n=0$, then $\widehat\mu_n=\widehat\delta_n=0$ solves \textup{(\ref{QLRnew})} (for any positive scalar $\widetilde\Sigma_n$). 
From \textup{(\ref{rhatformula})}, $\widehat s_n=1$ because $\widehat K=\{1,2\}$, but from \textup{(\ref{r_n})}, $\widehat r_n=0$. 
This example requires degenerate features that make $\widehat r_n<\widehat s_n$ unlikely to occur in practice.} 

\subsection{Uniform Asymptotic Validity of the GCC Test}
\label{UniformitySection}

Before describing the assumptions and theorems, we first clarify the true values of the parameters and the underlying distribution of the estimators. 
Let $F$ denote the joint distribution of $\overline{\mu}_n$, $\overline{C}_n$, and $\overline{\Omega}_n$, and let $P_F(\cdot)$ denote probabilities taken with respect to $F$. 
Let $\mathcal{F}_n$ be a parameter space for $F$.\footnote{We subscript $\mathcal{F}_n$ with $n$ because $\overline{\mu}_n$, $\overline{C}_n$, and $\overline{\Omega}_n$ are typically functions of a sample, $\{W_i\}_{i=1}^n$, with sample size $n$. Thus, their distribution naturally depends on $n$. We allow $\mathcal{F}_n$ to depend arbitrarily on $n$.} 
Let $\mu_F$ and $\Pi_{F}$ denote the true values of $\mu$ and $\Pi$. 
Also let $C_F=B\Pi_F+D$ and $b_F=d_n-B\mu_F$. 
The $F$ in the subscript makes explicit that these quantities depend on $F$. 
We allow the values of these parameters, together with the value of $d$, to depend on $n$ to incorporate the situation where we test a sequence of null hypotheses.\footnote{Testing a sequence of null hypotheses is required to evaluate the uniform coverage probability of a confidence set for a parameter of interest, $\theta$. Then, the inequalities that define the null hypothesis may depend on the hypothesized value of $\theta$.} 
We make explicit the dependence of $d$ on $n$, denoting it by $d_n$. 
For notational simplicity, we keep the dependence of $\mu_F$, $\Pi_F$, $C_F$, and $b_F$ on $n$ implicit. 

Let $\mathcal{F}_{n0}$ be the subset of $\mathcal{F}_n$ that satisfies the null hypothesis: $\mathcal{F}_{n0}=\{F\in\mathcal{F}_n: C_F\delta\le b_F \text{ for some }\delta\in\mathbb{R}^{d_\delta}\}$. 
For $F\in\mathcal{F}_{n0}$, let $\delta^\ast_{F}$ be the value of $\delta$ that satisfies $C_F\delta\le b_F$. 
If there is more than one such value of $\delta$, we take $\delta^\ast_{F}$ to be the one that has minimum norm: 
\begin{equation}
\delta^\ast_{F}=\underset{\delta: C_F\delta\le b_F}{\argmin}\|\delta\|.\label{deltanF}
\end{equation} 
This mimics the definition of $\widetilde\delta_n$. 

The following assumption ensures asymptotic normality of the estimators of the reduced-form parameters and consistent estimation of the asymptotic variance, at least along a subsequence of true data generating processes. 
It is used to show consistency of $\widetilde\delta_n$ and asymptotic uniform validity of the GCC test. 

\begin{assumption}\label{Assumption1}
For every sequence $\{F_n\}_{n=1}^{\infty}$ with $F_n\in {\cal F}_{n0}$ %($\mathcal{F}_0$ implicitly depends on $n$ too) 
and for every subsequence, $\{n_m\}$, there exists a further subsequence, $\{n_q\}$, a vector $\mu_\infty$, a vector $d_\infty$, a matrix $\Pi_\infty$, a positive semi-definite matrix, $\Omega_\infty$, and a positive definite matrix, $\Upsilon_\infty$, such that: 

\textup{(i)} $\mu_{F_{n_q}}\rightarrow\mu_\infty$, $\Pi_{F_{n_q}}\rightarrow \Pi_\infty$, and $d_{n_q}\to d_\infty$ 

\textup{(ii)} $\sqrt{n_q}\left(\begin{array}{c}\overline{\mu}_{n_q}-\mu_{F_{n_q}}\\\textup{vec}(\overline\Pi_{n_q}-\Pi_{F_{n_q}})\end{array}\right)\rightarrow_d N(\mathbf{0},\Omega_\infty)$

\textup{(iii)} $\Sigma_\infty:=\left(\begin{smallmatrix}I\\\delta^\ast_\infty\otimes I\end{smallmatrix}\right)'\Omega_\infty\left(\begin{smallmatrix}I\\\delta^\ast_\infty\otimes I\end{smallmatrix}\right)$ is positive definite, where $\delta^\ast_\infty:=\underset{\delta: B\mu_\infty+(D+B\Pi_\infty)\delta\le d_\infty}{\argmin}\|\delta\|$, 

\textup{(iv)} $\overline{\Omega}_{n_q}\rightarrow_p \Omega_\infty$, 

\textup{(v)} $\widehat\Upsilon_{n_q}\rightarrow_p \Upsilon_\infty$, and 

\textup{(vi)} $\{\delta\in \R^{d_{\delta}}:B\mu_\infty+(D+B\Pi_{\infty})\delta\leq d_\infty\} \neq \emptyset$. 
\end{assumption}

\noindent{\textbf{Remarks:}} \textup{(1)} {\it Assumption \textup{\ref{Assumption1}}  is stated using subsequences in order to ensure uniformity over $\mathcal{F}_{n0}$. 
Part \textup{(i)} assumes that $d_n$ and the sequence of true parameter values for the reduced-form parameters converge to some limits along a subsequence. 
This is equivalent to assuming that the parameter space for these parameters is compact.

\textup{(2)} Part \textup{(ii)} assumes asymptotic normality of the estimators for the reduced-form parameters along a subsequence. 
Part \textup{(iii)} requires $\Sigma_\infty$ to be positive definite. 
While this may seem restrictive, it is mitigated by the way the inequalities are specified in equation \textup{(\ref{generalH0})}. 
Specifically, the $B$ matrix allows us to write the inequalities as a linear function of a core collection of reduced-form parameters that admit an estimator with a positive definite asymptotic variance matrix. 
The matrix $B$ absorbs any linear dependence among the estimation errors of the inequalities. 
Also note that $\delta^\ast_\infty$ is well-defined by Assumption \textup{\ref{Assumption1}(vi)}. 
Part \textup{(iv)} assumes consistency of the estimator of the asymptotic variance. 
Part \textup{(v)} assumes consistency of the first-step weight matrix in equation \textup{(\ref{tildemudef})} for a positive definite limit. 
Parts \textup{(ii)}, \textup{(iv)}, and \textup{(v)} can be verified using standard consistency and asymptotic normality arguments. 
For example, when the data are i.i.d.\ and the model is a moment (in)equality model, they can be verified by the Lindeberg-Feller central limit theorem and a law of large numbers for triangular arrays. 

\textup{(3)} Part \textup{(vi)} assumes the constraint set for the limit is nonempty. Part \textup{(vi)} is guaranteed under part \textup{(i)} if, for example, there is a fixed compact set $\Delta$ such that $\{\delta\in \mathbb{R}^{d_\delta}: C_{F}\delta\leq b_{F}\}\subseteq\Delta$ for all $F\in{\cal F}_{n0}$.\footnote{This type of assumption is common in the literature. For example, it is assumed by \cite{Voronin2025} and \cite{GoffMbakop2025}.} In that case, for any sequence $F_{n_q}\in{\cal F}_{n_q0}$, there exists a $\delta_{F_{n_q}}$ such that  $B\mu_{F_{n_q}}+(B\Pi_{F_{n_q}}+D)\delta_{F_{n_q}}\leq d_{n_q}$. This sequence has a subsequence that converges to some limit $\delta_\infty$ that satisfies $B\mu_{\infty}+(B\Pi_\infty+D)\delta_\infty\leq d_\infty$, showing part \textup{(vi)}. Appendix \ref{A1iiiDiscussion} shows another way to verify part \textup{(vi)} under a strengthened version of Assumption \textup{\ref{assu:rank:simple}}, below.}\medskip

Next, we state the stable rank condition mentioned in the introduction.  We first introduce some new notation. Let $K^=\subseteq \{1,...,d_C\}$ be a set that contains the indices for the inequalities that were originally equalities. 
This set is special because it is always included in the set of active inequalities: $K^=\subseteq\widehat K$. 
For any $d_C\times d_\delta$ matrix $C$ and for any $d_C$-dimensional vector $b$, 
let $\mathcal{A}(C, b)=\{K\subseteq\{1,...,d_C\}: Cx\le b \text{ and } I_{K}(C x-b) = \mathbf{0} \text{ for some }x\in\mathbb{R}^{d_\delta}\}$ be the collection of all subsets of inequalities that could be simultaneously active for the system of inequalities defined by $C$ and $b$.\footnote{A combination of inequalities cannot be simultaneously active if, for example, it involves an upper and a lower bound that are parallel and separated. 
Such combinations are excluded from $\mathcal{A}(C, b)$.}

\begin{assumption}[Stable Rank]\label{assu:rank:simple}
For every sequence $\{F_n\}_{n=1}^\infty$ with $F_n\in\mathcal{F}_{n0}$ and for any subsequence, $\{n_m\}$, satisfying Assumption \textup{\ref{Assumption1}(i)} with $C_\infty=B\Pi_\infty+D$ and $b_\infty=d_\infty-B\mu_\infty$, there is a further subsequence, $\{n_q\}$, along which 
\begin{equation}
P_{F_{n_q}}(\textup{rk}(I_K\overline{C}_{n_q}) =\textup{rk}(I_K{C}_{F_{n_q}})=\textup{rk}(I_KC_\infty))\to 1\textup{, as }q\to \infty, \label{rankst}
\end{equation}
for any $K\in{\cal A}(C_\infty, b_\infty)$ that satisfies $K^{=}\subseteq K$. 
\end{assumption}

\noindent\textbf{Remarks:} (1) {\it We refer to Assumption \textup{\ref{assu:rank:simple}} as a ``stable rank'' condition because perturbations of the $C_\infty$ matrix in the directions of the estimation error do not change the rank. }

(2) {\it Assumption \textup{\ref{assu:rank:simple}} is not a necessary condition for the validity of the GCC test. 
{It is possible to relax {\it Assumption \textup{\ref{assu:rank:simple}} by reducing the number of collections of indices, $K$, for which the rank equality in \textup{(\ref{rankst})} needs to be assumed. 
In Appendix \textup{\ref{relaxations}}, we show that the rank equality only needs to hold for index sets, $K$ corresponding to collections of inequalities that define a linear subspace in the limit. 
In cases where there are no equalities and the identified set for the nuisance parameters has a positive volume in the parameter space, every inequality could be slack, and no stable rank condition is needed.} 
Due to the nuances of this discussion, it is relegated to the Supplemental Appendix.}}

(3) {\it While Assumption \textup{\ref{assu:rank:simple}} is not necessary, it is used in an essential way in the proofs of consistency of $\widetilde{\delta}_n$ and asymptotic validity of the GCC test. 
Moreover, a more than superficial connection of Assumption \textup{\ref{assu:rank:simple}} with the weak IV problem in linear IV regression models suggests that relaxing Assumption \textup{\ref{assu:rank:simple}} completely may require insights from that literature. 
We discuss this connection in Section \textup{\ref{Assumption1Discussion}}, below.}

(4) {\it Assumptions playing a similar role as Assumption \textup{\ref{assu:rank:simple}} are common in the literature on subvector inference in moment inequality models and in models defined by linear systems. 
One type of such assumptions is a known and fixed $C$, as in \cite{GuggenbergerHahnKim2008}, \cite{KaidoSantos2014}, and \cite{FSST2023}. In that case $\overline{C}_n=C_F=  C$, and Assumption \textup{\ref{assu:rank:simple}} holds trivially. 
Other types of assumptions appear in \cite{PPHI2015}, \cite{BugniCanayShi2017}, \cite{ChoRussell2024}, and \cite{GoffMbakop2025}.\footnote{Assumption \ref{assu:rank:simple} differs from constraint qualification, as considered in \cite{KMS2022}. Constraint qualification restricts a fixed collection of constraints to ensure KKT conditions are necessary or sufficient or the KKT multipliers are unique. In contrast, Assumption \ref{assu:rank:simple} concerns a sequence of linear inequality constraints and restricts the way they converge to a limiting set of constraints.} 
We discuss the connection between these assumptions in a simple example in Section \textup{\ref{sec:compare}}}. \medskip

The following theorem states an important preliminary result: consistency of $\widetilde{\delta}_n$. 

\begin{theorem}\label{lem:delta} Suppose Assumption \textup{\ref{assu:rank:simple}} holds. 
Let $\{F_n\}_{n=1}^\infty$ be a sequence with $F_n\in\mathcal{F}_{0n}$ and let $\{n_q\}$ be a subsequence satisfying Assumptions \textup{\ref{Assumption1}(i), (ii)},  \textup{(v)}, and \textup{(vi)}. 
We have that 
\[
\widetilde{\delta}_{n_q}\to_p\delta_\infty^\ast~\text{ and }~\delta^\ast_{F_{n_q}}\rightarrow \delta^\ast_\infty~\text{ as }~q\to\infty, 
\]
where $\delta^\ast_\infty$ is defined in Assumption \textup{\ref{Assumption1}(iii)}. 
\end{theorem}

\noindent\textbf{Remark:} \textit{
Consistency of estimators defined by tie-breaking procedures, such as the norm minimization in the definition of $\widetilde{\delta}_n$, is especially challenging. 
It is surprising that Assumption \textup{\ref{assu:rank:simple}} is sufficient in this case.
The proof uses a novel argument that establishes setwise convergence of the constraint set.} \medskip

The following theorem is the main theoretical result of the paper. 
\begin{theorem}\label{thm:level_CC} If Assumptions \textup{\ref{Assumption1}} and \textup{\ref{assu:rank:simple}} hold, then 
\begin{equation*}
\underset{n\to\infty}{\textup{limsup}}\sup_{F\in{\cal F}_{n0}}P_F(T_n>cv(\widehat s_n,\alpha)) \leq\alpha. 
\end{equation*}
\end{theorem}

\noindent{\textbf{Remarks:}} (1) {\it Theorem \textup{\ref{thm:level_CC}} establishes the uniform asymptotic validity of the GCC test. 
This extends the result for the sCC test from CS23 to allow $C$ to be estimated, as long as the estimator is consistent and asymptotically normal and a stable rank condition is satisfied. 
The generalization is essential for handling the applications discussed in Section} \ref{sec:setup}. 

(2) {\it The asymptotic validity of the GCC test is surprising because, intuitively, the active inequalities are not necessarily binding in population and even when all inequalities are binding, the limit distribution of $T_n$ is not $\chi^2$. 
Indeed, the set of active inequalities does not converge to the set of binding-in-population inequalities but remains random in the limit. 
The key to validity of the GCC test is that the limit {\em conditional} distribution of $T_n$ given the set of active inequalities is bounded by the $\chi^2$ distribution with the associated  DoF. }

(3) {\it When $P_F(\widehat{s}_n = 0)>0$, the GCC test can be slightly conservative: Its  null rejection probability is between $\alpha (1-P_F(\widehat{s}_n = 0))$ and $\alpha$ asymptotically. The refinement in CS23 can be used to remove the conservativeness. We define the refined GCC (RGCC) test in Appendix} \ref{app:refinement}{\it. The RGCC test  differs from the GCC test only when $\widehat{s}_n=1$ and is also tuning parameter and simulation free. However, the refinement requires calculating $A$ and $g$ such that $\{\mu\in\R^{d_\mu}:B\mu+\overline{C}_n \delta\leq d \text{ for some } \delta\in\R^{d_\delta}\} = \{\mu\in\R^{d_\mu}:A \mu\leq g\}$. 
The computation is relatively easy when $d_C$ and $d_\delta$ are small but gets exponentially harder when $d_C$ and $d_\delta$ increase. In particular, it can have a high memory requirement. We investigate the performance of the RGCC test along with the GCC test in the simulations and the empirical illustration.} 

\subsection{Assumption \ref{assu:rank:simple} and Instrumental Variable Regressions}\label{Assumption1Discussion}\label{sub:discussion_rank}
To better understand Assumption \ref{assu:rank:simple}, we now relate it to the linear IV regression model. 
\begin{example}[IV Regression]\label{ex:iv_regression} Let $Y$ be a scalar dependent variable and $X$ be a $d_x$-vector of potentially endogenous regressors. Consider the IV regression model: $Y = X'\beta +\varepsilon$, where $\varepsilon$ is the error term. 
Let $Z$ be a $d_z$-vector of instruments that satisfy $\E[Z\varepsilon]=\mathbf{0}$. Suppose we are interested in the first element of $\beta$, denoted by $\theta$. Then the rest of the elements of $\beta$ are nuisance parameters, denoted  by $\delta$. Let $X_1$ denote the first element of $X$ and $X_{-1}$ denote the rest of the elements. Then, the model can be represented by the following moment conditions: $\E[ZY] - \theta \E[ZX_1] - \E[ZX'_{-1}]\delta = \mathbf{0}$. 
If we conduct inference for $\theta$ by test inversion, the hypothesis to be tested for each $\theta$ value is  \begin{equation}
H_0: \E[ZY]-\theta \E[ZX_1]  - \E[ZX_{-1}']\delta = \mathbf{0} \text{ for some }\delta\in\mathbb{R}^{d_\delta}.\label{H0iv}
\end{equation}
This hypothesis is a special case of that in equation \textup{(\ref{generalH0})} with $B= \left(\begin{smallmatrix}I_{d_z}\\-I_{d_z}\end{smallmatrix}\right), \mu = \E[ZY]-\theta \E[ZX_1]$, $\Pi = -\E[ZX_{-1}']$, $D = \mathbb{O}$, and $d = \mathbf{0}.$
\end{example}

In this model, Assumption \ref{assu:rank:simple} allows $\E[ZX_{-1}']$ to change with $n$ as long as the rank does not change in the limit. 
Equivalently, the smallest nonzero eigenvalue of $\E[ZX_{-1}']\E[X_{-1}Z']$ does not converge to zero. 
Notably, zero eigenvalues are allowed. 
This happens, for example, when the number of instruments is smaller than the number of nuisance parameters. 
Assumption \ref{assu:rank:simple} is weaker than the usual rank condition for strong identification of $\delta$ (under a hypothesis that fixes a value of $\theta$), which is that the smallest eigenvalue of $\E[ZX_{-1}']\E[X_{-1}Z']$ is bounded away from zero. 
This means that Assumption \ref{assu:rank:simple} can be thought of as a ``no weak identification'' condition, where linear combinations of $\delta$ can be strongly identified or non-identified as long as they are not weakly identified. 

Even in the weak instruments/weak identification literature, strong identification of the nuisance parameters is a useful assumption. 
For example, \cite{StockWright2000}, \cite{Kleibergen2005},   and \cite{AndrewsMikusheva2016} propose identification-robust hypothesis tests for subvectors only when the nuisance parameters are strongly identified under the null. 
Papers that cover inference with weakly identified nuisance parameters, including \cite{ChaudhuriZivot2011}, \cite{Andrews2018}, and \cite{GKM2024}, recommend some version of two-step inference requiring a tuning parameter, among other complications.\footnote{An exception is \cite{GKMC2012}. They focus on a homoskedastic linear IV model and show that the plug-in Anderson-Rubin test remains valid with weakly identified nuisance parameters.} 
\cite{Cox2022} states separate limit theory depending on whether the nuisance parameters are strongly identified under the null. 
In this literature, strong identification of the nuisance parameters under the null is used to guarantee that the null-imposed estimator of the nuisance parameters is consistent and asymptotically normal. 
In contrast, we use Assumption \ref{assu:rank:simple} to ensure convergence of the constraint set in the QLR statistic to its limit in a setwise sense. 
These different purposes reflect the compounding complications that arise when trying to relax Assumption \ref{assu:rank:simple}. 

\section{Simulations}\label{sec:mc}

This section evaluates the finite-sample performance of the GCC test for testing inequalities that are linear in nuisance parameters. 
We also evaluate the RGCC test, defined in Appendix \ref{app:refinement}. 
Section \ref{Simple_Simulations} considers a simple one-sided hypothesis testing problem with one nuisance parameter. 
Section \ref{sub:IOR} considers a more realistic model with more nuisance parameters: an interval outcome IV regression model.
Also, a simulation of Example \ref{ex:PRTE} can be found in Section \ref{app:MTE}. 
The bottom line of all the simulations is that the GCC and RGCC tests are easy to compute and have good size and power. 

\subsection{A Simple One-Sided Model}
\label{Simple_Simulations}

Consider a simple one-sided hypothesis testing problem with one nuisance parameter and normally distributed randomness. 
The model is designed to abstract from computational and asymptotic complications, so as to focus on size and power. 
The validity of any test for linear inequalities in this simple specification should be a necessary condition for implementing the test in practice. 
Also note that, because the bound on the parameter of interest is one-sided, we can compare to methods from the literature on one-sided inference for the value of a LPP. 

The simple  model has one nuisance parameter, no equalities, and $J$ inequalities:
\begin{align}
\mu_1+c_1\delta+\theta&\leq 0\nonumber\\
\mu_2+c_2\delta+\theta&\leq 0\nonumber\\
\mu_3+c_3\delta\hspace{7.2mm}&\leq 0\nonumber\\
\vdots\hspace{10mm}&\hspace{5mm}\vdots\nonumber\\
\mu_J+c_J\delta\hspace{7.2mm}&\leq 0.\label{SimpleM}
\end{align}
This model is a special case of (\ref{generalH0}) with $B=I_J$, $\mu = (\mu_1,\mu_2,...,\mu_J)'$, $\Pi = C=(c_1,c_2,...,c_J)'$, $D = \mathbf{0}_J$, and $d = -(1, 1, 0,...,0)'\theta$. 
The first two inequalities give an upper bound on $\theta$, while the remaining $J-2$ inequalities only bound $\delta$. 

Suppose $\mu$ is estimated by $\overline{\mu}_n$ and $C$ by $\overline{C}_n$. 
Suppose $\overline{\mu}_n$ and $\overline{C}_n$ are sample means of independent random samples from $N(\mu,I_J)$ and $N(C,2I_J)$, respectively. 
The covariance matrix of $\overline{\mu}_n$ and $\overline{C}_n$ are estimated by the sample variances and covariances. 
Below, we consider $\mu = (-1,1,1-qn^{-1/2},...,1-qn^{-1/2})'$ and $C = (1,-1,-1,...,-1)'$ with $n=500$, $J\in\{3,10,50\}$, and $q\in\{0,4\}$. 
Inequalities 3 through $J$ may be binding or slack depending on the value of $q$. 
The identified set for $\theta$ is $(-\infty,0]$. 
For a fixed a value of $\theta$ in $(-\infty,0]$, the identified set for $\delta$ is $[\max(1+\theta,1-qn^{-1/2}), 1-\theta]$. 

\begin{table}[t]
\begin{center}
\scalebox{1}{
\begin{threeparttable}
\caption{Simulated Null Rejection Probabilities (NRP) and Computation Times of Various Tests of Nominal Level $5\%$ in the Simple Model}\label{Simple_Times}
\begin{tabular}{lcccccccc}
\hline\hline\vspace{-0.3cm}\\
&\multicolumn{2}{c}{$J=3$} &&\multicolumn{2}{c}{$J=10$}&&  \multicolumn{2}{c}{$J=50$} \\
\cline{2-3}\cline{5-6}\cline{8-9}\vspace{-0.3cm}\\
&NRP&Time (seconds)&&NRP&Time (seconds)& & NRP &Time (seconds)\\
\hline\vspace{-0.3cm}\\
GCC&0.035&(0.003)&&0.035&(0.003)&&0.045&(0.008)\\
RGCC&0.051&(0.003)&&0.038&(0.004)&&0.045&(0.008)\\
\hline\vspace{-0.3cm}\\
sCC&0.193&(0.025)&&0.324&(0.070)&&0.560&(0.343)\\
sRCC&0.207&(0.026)&&0.327&(0.071)&&0.560&(0.343)\\
ARP&0.204&(1.159)&&0.248&(1.097)&&0.476&(1.988)\\
\hline\vspace{-0.3cm}\\
CCT&0.045&(0.124)&&0.179&(0.138)&&0.430&(0.267)\\
BCS&0.025&(0.575)&&0.015&(1.048)&&0.004&(3.525)\\
Bei&0.052&(0.238)&&0.072&(0.267)&&0.112&(0.597)\\
\hline\vspace{-0.3cm}\\
FSST&0.231&(1.725)&&0.408&(3.585)&&0.609&(3.083)\\
Gaf&0.064&(0.129)&&0.179&(0.123)&&0.420&(0.127)\\
CR1&0.002&(25.83)&&0.002&(25.84)&&0.014&(26.22)\\
CR2&0.026&(25.87)&&0.055&(25.91)&&0.130&(26.33)\\
CR3&0.036&(25.85)&&0.074&(25.83)&&0.168&(26.08)\\
\hline
\end{tabular}
{\small
\begin{tablenotes}
\item {\em Note:} The ``NRP'' columns report the null rejection probabilities for the hypothesis $\theta=0$ in the simple specification with 1000 simulations. 
The ``Time'' columns report the median time in seconds to compute the test for all the tests except Gaf and CR1-CR3. For Gaf and CR1-CR3, the ``Time'' columns report the median number of seconds to compute a confidence interval. 
\end{tablenotes}
}
\end{threeparttable}
}
\end{center}
\end{table}

We first consider testing the null hypothesis $H_0: \theta=0$.
Table \ref{Simple_Times} reports the simulated rejection probabilities for various tests when $q=0$. This means that all $J$ inequalities are binding. 
We implement four groups of tests. 
The first group consists of the GCC and RGCC tests. 
The second group consists of the sCC and sRCC tests from CS23 and the hybrid test (ARP) from \cite{AndrewsRothPakes2023}. These tests are implemented using $\overline{C}_n$ as if it were the true value. 
The third group consists of the MR test (BCS) from \cite{BugniCanayShi2017}, procedure 3 (CCT) from \cite{CCT2018}, and the recommended test (Bei) from \cite{Bei2023}, which are designed for nonlinear inequalities. 
The fourth group consists of the recommended test (FSST) from \cite{FSST2023}, the recommended test (Gaf) from \cite{Gafarov2025}, and three tests (CR1-CR3) from \cite{ChoRussell2024} implemented with three choices of the tuning parameter.\footnote{CR1, CR2, and CR3 are implemented with $\underline{\epsilon}=0.1$, $\underline{\epsilon}=0.01$, and $\underline{\epsilon}=0.001$, respectively.} 
These tests are from the literature on inference for the value of a LPP. 
They are implemented for one-sided inference on the upper bound of $\theta$.\footnote{Gaf and CR1-CR3 report confidence intervals instead of hypothesis tests. For these methods, we say that a value of $\theta$ is rejected if $\theta$ does not belong to the confidence interval.} 

As we can see in Table \ref{Simple_Times}:\medskip 

 (1) {The GCC and RGCC tests are valid for every $J$. The GCC test is somewhat conservative when $J=3$, which is expected. Unexpected is that the GCC and RGCC tests are somewhat conservative when $J=10$. This could be due to simulation noise}. 

(2) {The sCC, sRCC, and ARP tests are invalid. This demonstrates the need to account for the estimation error in $\overline{C}_n$.\footnote{Another way to implement these tests is to consider the case that the inequalities in (\ref{SimpleM}) hold conditionally on $\overline{C}_n$. Then, the sCC, sRCC, and ARP tests can be implemented with the conditional variance of $\overline{\mu}_n$ given $\overline{C}_n$. Since $\overline{\mu}_n$ is independent of $\overline{C}_n$, the conditional variance is the same as the unconditional variance and the implementation is the same. Thus, Table \ref{Simple_Times} shows that neither way to implement these tests is valid. The problem is that the inequalities in (\ref{SimpleM}) do not hold conditionally on $\overline{C}_n$.} }

(3) {Concerning the tests in the third group, CCT is invalid for $J>3$. It is only valid under a high-level condition that is not satisfied in this model; see Assumption 4.7 in \cite{CCT2018}. 
BCS is valid and becomes quite conservative when $J=50$. 
Bei is a more computationally feasible version of BCS. 
Table \ref{Simple_Times} suggests Bei has mild to moderate over-rejection. It is unclear why the null rejection probabilities for BCS and Bei are so different. }

(4) {Concerning the tests in the fourth group, FSST is understandably invalid because it requires known Jacobian. 
Surprisingly, Gaf is also invalid for $J>3$. This could be because a rank condition is not satisfied; see Assumption 2 in \cite{Gafarov2025}. 
CR1-CR3 appear to be very sensitive to the choice of the tuning parameter. CR1 is very conservative while CR2 and CR3 have moderate over-rejection when $J=50$. }

(5) {Computationally, the GCC and RGCC tests are by far the fastest among the tests considered. 
Comparing their times to the sCC and sRCC tests demonstrates the computational improvement of the new algebraic formula for the DoF. 
For CR1-CR3 and Gaf, the computational time is for the calculation of a confidence interval and thus is not comparable to the computation time for a single hypothesis. }
\medskip

\begin{figure}
	\caption{Simulated Power Curves and Computation Times of Various Tests of Nominal Level 5\% in the Simple Model ($n=500$)}
	\label{Fig:powercurve_simple_sec}
	\begin{center}
		\begin{threeparttable}
			\begin{subfigure}[b]{0.465\textwidth}
				\begin{center}
					\begin{tikzpicture}
						\tikzmath{\n = 500; \xmin = -1; \xmax = 6; \gridmin = \xmin / sqrt(\n); \gridmax = \xmax / sqrt(\n); } 
						\begin{axis}
							[width=0.825\textwidth,
							ymin=0,ymax=1,xmin=\gridmin,xmax=\gridmax,
							xticklabel style={font=\scriptsize}, 
							yticklabel style={font=\scriptsize}, 
							legend style={at={(axis cs:\gridmax,0.985)},anchor=north west}]
							\draw [pattern=north west lines, pattern color=gray!70,draw=none] (\gridmin,0) rectangle (0,0.996);
							\draw[dotted] (\gridmin,0.05) -- (\gridmax,0.05);
							\addplot[very thick] table [x=theta, y=rr_GCC, col sep=semicolon]{Simple_Power_0_6_Sims_1000_q_0_n_500_J_3.txt};
							\addplot[dashed, very thick, gray] table [x=theta, y=rr_RGCC, col sep=semicolon]{Simple_Power_0_6_Sims_1000_q_0_n_500_J_3.txt};				
							\addplot[mark=*, mark size=0.7pt, mark repeat={5}, thick, red] table [x=theta, y=rr_BCS, col sep=semicolon]{Simple_Power_0_6_Sims_1000_q_0_n_500_J_3.txt};
							\addplot[mark=x, mark repeat={5}, thick, purple] table [x=theta, y=rr_Bei, col sep=semicolon]{Simple_Power_0_6_Sims_1000_q_0_n_500_J_3.txt};
        						\addplot[mark=square*, mark size=0.7pt, mark repeat={5}, thick, teal] table [x=theta, y=rr_CR1, col sep=semicolon]{Simple_Power_0_6_Sims_1000_q_0_n_500_J_3.txt};
        						\addplot[very thick, teal,dotted] table [x=theta, y=rr_CR2, col sep=semicolon]{Simple_Power_0_6_Sims_1000_q_0_n_500_J_3.txt};
							\addlegendentry{{\scriptsize GCC [0.09s]}};
							\addlegendentry{{\scriptsize RGCC [0.11s]}};
							\addlegendentry{{\scriptsize BCS [23s]}};
							\addlegendentry{{\scriptsize Bei [8.8s]}};
							\addlegendentry{{\scriptsize CR1 [20s]}};
							\addlegendentry{{\scriptsize CR2 [20s]}};
						\end{axis}
					\end{tikzpicture}
					\caption{$J=3$ and $q=0$}
					\label{Fig:powercurve_simple_sec_J3_q0}
				\end{center}		 
			\end{subfigure}
			\hfill
			\begin{subfigure}[b]{0.465\textwidth}  
				\begin{center}
					\begin{tikzpicture}
						\tikzmath{\n = 500; \xmin = -1; \xmax = 6; \gridmin = \xmin / sqrt(\n); \gridmax = \xmax / sqrt(\n); } 
						\begin{axis}
							[width=0.825\textwidth,
							ymin=0,ymax=1,xmin=\gridmin,xmax=\gridmax,
							xticklabel style={font=\scriptsize}, 
							yticklabel style={font=\scriptsize}, 
							legend style={at={(axis cs:\gridmax,0.985)},anchor=north west}]
							\draw [pattern=north west lines, pattern color=gray!70,draw=none] (\gridmin,0) rectangle (0,0.996);
							\draw[dotted] (\gridmin,0.05) -- (\gridmax,0.05);
							\addplot[very thick] table [x=theta, y=rr_GCC, col sep=semicolon]{Simple_Power_0_6_Sims_1000_q_4_n_500_J_3.txt};
							\addplot[dashed, very thick, gray] table [x=theta, y=rr_RGCC, col sep=semicolon]{Simple_Power_0_6_Sims_1000_q_4_n_500_J_3.txt};					
							\addplot[mark=*, mark size=0.7pt, mark repeat={5}, thick, red] table [x=theta, y=rr_BCS, col sep=semicolon]{Simple_Power_0_6_Sims_1000_q_4_n_500_J_3.txt};
							\addplot[mark=x, mark repeat={5}, thick, purple] table [x=theta, y=rr_Bei, col sep=semicolon]	{Simple_Power_0_6_Sims_1000_q_4_n_500_J_3.txt};	        						\addplot[mark=square*, mark size=0.7pt, mark repeat={5}, thick, teal] table [x=theta, y=rr_CR1, col sep=semicolon]{Simple_Power_0_6_Sims_1000_q_4_n_500_J_3.txt};
        						\addplot[very thick, teal,dotted] table [x=theta, y=rr_CR2, col sep=semicolon]{Simple_Power_0_6_Sims_1000_q_4_n_500_J_3.txt};
							\addlegendentry{{\scriptsize GCC [0.09s]}};
							\addlegendentry{{\scriptsize RGCC [0.11s]}};
							\addlegendentry{{\scriptsize BCS [22s]}};
							\addlegendentry{{\scriptsize Bei [8.7s]}};										
							\addlegendentry{{\scriptsize CR1 [20s]}};
							\addlegendentry{{\scriptsize CR2 [21s]}};
						\end{axis}
					\end{tikzpicture}
					\caption{$J=3$ and $q=4$}
					\label{Fig:powercurve_simple_sec_J3_q4}	
				\end{center}
			\end{subfigure}
			\vskip\baselineskip
			\begin{subfigure}[b]{0.465\textwidth}   
				\begin{center}
					\begin{tikzpicture}
						\tikzmath{\n = 500; \xmin = -1; \xmax =6; \gridmin = \xmin / sqrt(\n); \gridmax = \xmax / sqrt(\n); } 
						\begin{axis}
							[width=0.825\textwidth,
							ymin=0,ymax=1,xmin=\gridmin,xmax=\gridmax,
							xticklabel style={font=\scriptsize}, 
							yticklabel style={font=\scriptsize}, 
							legend style={at={(axis cs:\gridmax,0.985)},anchor=north west}]
							\draw [pattern=north west lines, pattern color=gray!70,draw=none] (\gridmin,0) rectangle (0,0.996);
							\draw[dotted] (\gridmin,0.05) -- (\gridmax,0.05);
							\addplot[very thick] table [x=theta, y=rr_GCC, col sep=semicolon]{Simple_Power_0_6_Sims_1000_q_0_n_500_J_10.txt};
							\addplot[dashed, very thick, gray] table [x=theta, y=rr_RGCC, col sep=semicolon]{Simple_Power_0_6_Sims_1000_q_0_n_500_J_10.txt};					
							\addplot[mark=*, mark size=0.7pt, mark repeat={5}, thick, red] table [x=theta, y=rr_BCS, col sep=semicolon]{Simple_Power_0_6_Sims_1000_q_0_n_500_J_10.txt};
							\addplot[mark=x, mark repeat={5}, thick, purple] table [x=theta, y=rr_Bei, col sep=semicolon]{Simple_Power_0_6_Sims_1000_q_0_n_500_J_10.txt};			\addplot[mark=square*, mark size=0.7pt, mark repeat={5}, thick, teal] table [x=theta, y=rr_CR1, col sep=semicolon]{Simple_Power_0_6_Sims_1000_q_0_n_500_J_10.txt};
        						\addplot[very thick, teal,dotted] table [x=theta, y=rr_CR2, col sep=semicolon]{Simple_Power_0_6_Sims_1000_q_0_n_500_J_10.txt};
							\addlegendentry{{\scriptsize GCC [0.10s]}};
							\addlegendentry{{\scriptsize RGCC [0.10s]}};
							\addlegendentry{{\scriptsize BCS [42s]}};
							\addlegendentry{{\scriptsize Bei [10s]}};
							\addlegendentry{{\scriptsize CR1 [20s]}};
							\addlegendentry{{\scriptsize CR2 [20s]}};
						\end{axis}
					\end{tikzpicture}
					\caption{$J=10$ and $q=0$}
					\label{Fig:powercurve_simple_sec_J10_q0}
				\end{center}
			\end{subfigure}
			\hfill
			\begin{subfigure}[b]{0.465\textwidth}   
				\begin{center}
					\begin{tikzpicture}
						\tikzmath{\n = 500; \xmin = -1; \xmax = 6; \gridmin = \xmin / sqrt(\n); \gridmax = \xmax / sqrt(\n); } 
						\begin{axis}
							[width=0.825\textwidth,
							ymin=0,ymax=1,xmin=\gridmin,xmax=\gridmax,
							xticklabel style={font=\scriptsize}, 
							yticklabel style={font=\scriptsize}, 
							legend style={at={(axis cs:\gridmax,0.985)},anchor=north west}]
							\draw [pattern=north west lines, pattern color=gray!70,draw=none] (\gridmin,0) rectangle (0,0.996);
							\draw[dotted] (\gridmin,0.05) -- (\gridmax,0.05);
							\addplot[very thick] table [x=theta, y=rr_GCC, col sep=semicolon]{Simple_Power_0_6_Sims_1000_q_4_n_500_J_10.txt};
							\addplot[dashed, very thick, gray] table [x=theta, y=rr_RGCC, col sep=semicolon]{Simple_Power_0_6_Sims_1000_q_4_n_500_J_10.txt};					
							\addplot[mark=*, mark size=0.7pt, mark repeat={5}, thick, red] table [x=theta, y=rr_BCS, col sep=semicolon]{Simple_Power_0_6_Sims_1000_q_4_n_500_J_10.txt};
							\addplot[mark=x, mark repeat={5}, thick, purple] table [x=theta, y=rr_Bei, col sep=semicolon]{Simple_Power_0_6_Sims_1000_q_4_n_500_J_10.txt};        											\addplot[mark=square*, mark size=0.7pt, mark repeat={5}, thick, teal] table [x=theta, y=rr_CR1, col sep=semicolon]{Simple_Power_0_6_Sims_1000_q_4_n_500_J_10.txt};
        						\addplot[very thick, teal,dotted] table [x=theta, y=rr_CR2, col sep=semicolon]{Simple_Power_0_6_Sims_1000_q_4_n_500_J_10.txt};
							\addlegendentry{{\scriptsize GCC [0.10s]}};
							\addlegendentry{{\scriptsize RGCC [0.12s]}};
							\addlegendentry{{\scriptsize BCS [40s]}};
							\addlegendentry{{\scriptsize Bei [10s]}};
							\addlegendentry{{\scriptsize CR1 [20s]}};
							\addlegendentry{{\scriptsize CR2 [20s]}};
						\end{axis}
					\end{tikzpicture}
					\caption{$J=10$ and $q=4$}
					\label{Fig:powercurve_simple_sec_J10_q4}
				\end{center}
			\end{subfigure}
			\vskip\baselineskip
			\begin{subfigure}[b]{0.465\textwidth}   
				\begin{center}
					\begin{tikzpicture}
						\tikzmath{\n = 500; \xmin = -1; \xmax = 6; \gridmin = \xmin / sqrt(\n); \gridmax = \xmax / sqrt(\n); } 
						\begin{axis}
							[width=0.825\textwidth,
							ymin=0,ymax=1,xmin=\gridmin,xmax=\gridmax,
							xticklabel style={font=\scriptsize}, 
							yticklabel style={font=\scriptsize}, 
							legend style={at={(axis cs:\gridmax,0.985)},anchor=north west}]
							\draw [pattern=north west lines, pattern color=gray!70,draw=none] (\gridmin,0) rectangle (0,0.996);
							\draw[dotted] (\gridmin,0.05) -- (\gridmax,0.05);
							\addplot[very thick] table [x=theta, y=rr_GCC, col sep=semicolon]{Simple_Power_0_6_Sims_1000_q_0_n_500_J_50.txt};
							\addplot[dashed, very thick, gray] table [x=theta, y=rr_RGCC, col sep=semicolon]{Simple_Power_0_6_Sims_1000_q_0_n_500_J_50.txt};					
							\addplot[mark=*, mark size=0.7pt, mark repeat={5}, thick, red] table [x=theta, y=rr_BCS, col sep=semicolon]{Simple_Power_0_6_Sims_1000_q_0_n_500_J_50.txt};
							\addplot[mark=x, mark repeat={5}, thick, purple] table [x=theta, y=rr_Bei, col sep=semicolon]{Simple_Power_0_6_Sims_1000_q_0_n_500_J_50.txt};
    \addplot[mark=square*, mark size=0.7pt, mark repeat={5}, thick, teal] table [x=theta, y=rr_CR1, col sep=semicolon]{Simple_Power_0_6_Sims_1000_q_0_n_500_J_50.txt};
        						\addplot[very thick, teal,dotted] table [x=theta, y=rr_CR2, col sep=semicolon]{Simple_Power_0_6_Sims_1000_q_0_n_500_J_50.txt};
							\addlegendentry{{\scriptsize GCC [0.32s]}};
							\addlegendentry{{\scriptsize RGCC [0.32s]}};
							\addlegendentry{{\scriptsize BCS [143s]}};
							\addlegendentry{{\scriptsize Bei [24s]}};
							\addlegendentry{{\scriptsize CR1 [19s]}};
							\addlegendentry{{\scriptsize CR2 [19s]}};
						\end{axis}
					\end{tikzpicture}
					\caption{$J=50$ and $q=0$}
					\label{Fig:powercurve_simple_sec_J50_q0}
				\end{center}
			\end{subfigure}
			\hfill
			\begin{subfigure}[b]{0.465\textwidth}   
				\begin{center}
					\begin{tikzpicture}
						\tikzmath{\n = 500; \xmin = -1; \xmax = 6; \gridmin = \xmin / sqrt(\n); \gridmax = \xmax / sqrt(\n); } 
						\begin{axis}
							[width=0.825\textwidth,
							ymin=0,ymax=1,xmin=\gridmin,xmax=\gridmax,
							xticklabel style={font=\scriptsize}, 
							yticklabel style={font=\scriptsize}, 
							legend style={at={(axis cs:\gridmax,0.985)},anchor=north west}]
							\draw [pattern=north west lines, pattern color=gray!70,draw=none] (\gridmin,0) rectangle (0,0.996);
							\draw[dotted] (\gridmin,0.05) -- (\gridmax,0.05);
							\addplot[very thick] table [x=theta, y=rr_GCC, col sep=semicolon]{Simple_Power_0_6_Sims_1000_q_4_n_500_J_50.txt};
							\addplot[dashed, very thick, gray] table [x=theta, y=rr_RGCC, col sep=semicolon]{Simple_Power_0_6_Sims_1000_q_4_n_500_J_50.txt};			
							\addplot[mark=*, mark size=0.7pt, mark repeat={5}, thick, red] table [x=theta, y=rr_BCS, col sep=semicolon]{Simple_Power_0_6_Sims_1000_q_4_n_500_J_50.txt};
							\addplot[mark=x, mark repeat={5}, thick, purple] table [x=theta, y=rr_Bei, col sep=semicolon]	{Simple_Power_0_6_Sims_1000_q_4_n_500_J_50.txt};
 \addplot[mark=square*, mark size=0.7pt, mark repeat={5}, thick, teal] table [x=theta, y=rr_CR1, col sep=semicolon]{Simple_Power_0_6_Sims_1000_q_4_n_500_J_50.txt};
        						\addplot[very thick, teal,dotted] table [x=theta, y=rr_CR2, col sep=semicolon]{Simple_Power_0_6_Sims_1000_q_4_n_500_J_50.txt};
							\addlegendentry{{\scriptsize GCC [0.30s]}};
							\addlegendentry{{\scriptsize RGCC [0.81s]}};
							\addlegendentry{{\scriptsize BCS [134s]}};
							\addlegendentry{{\scriptsize Bei [24s]}};
							\addlegendentry{{\scriptsize CR1 [19s]}};
							\addlegendentry{{\scriptsize CR2 [20s]}};   
						\end{axis}
					\end{tikzpicture}
					\caption{$J=50$ and $q=4$}
					\label{Fig:powercurve_simple_sec_J50_q4}
				\end{center}
			\end{subfigure}
		\end{threeparttable}
	\end{center}
{\small {\em Note:} The shaded regions indicate the identified set for $\theta$. The number of simulations is 1000. The numbers in the square brackets in the legends are the median time needed to calculate the test at all grid points for one simulation (measured in seconds).} 
\end{figure}
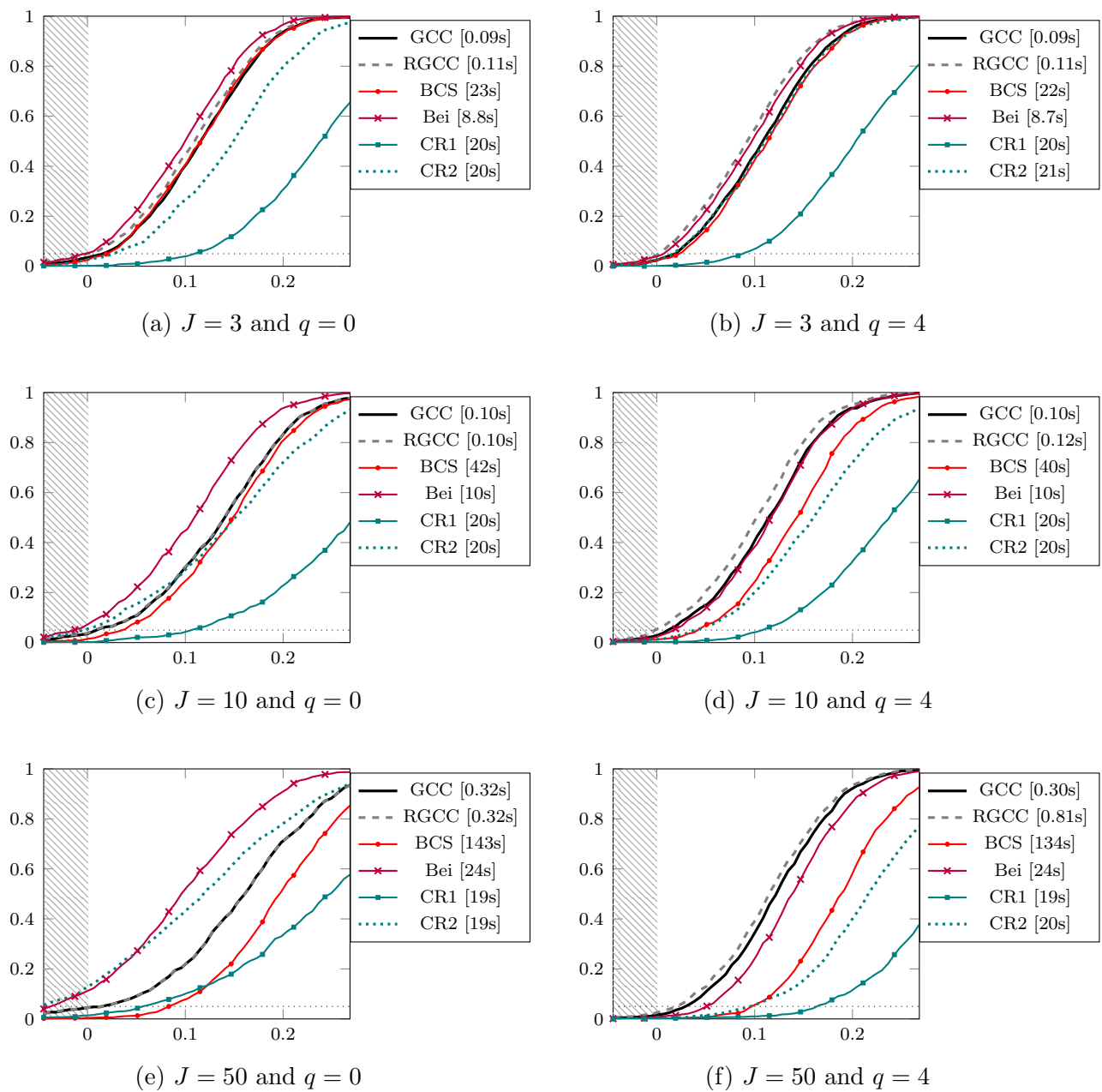

\addtocounter{figure}{-1} % I don't know why, but for some reason latex thinks the above figure is two figures. 

We also consider the power functions of some of the tests as a function of the hypothesized value of $\theta$. 
In addition to the GCC and RGCC tests, we include the BCS and Bei tests as benchmarks and the CR1 and CR2 tests to see whether the sensitivity of their NRPs to the tuning parameter carries over to the power functions. 
These tests are included because they are the ones that are valid or exhibit only moderate over-rejection in Table \ref{Simple_Times}. 

We test the inequalities in (\ref{SimpleM}) for a grid of values of $\theta$ between $-1/\sqrt{n}$ and $6/\sqrt{n}$. %in $\{-1, -0.9, -0.8, ..., 1.8, 1.9, 2\}/\sqrt{n}$. 
Dividing the grid by $\sqrt{n}$ ensures the power function approaches the asymptotic local power function as $n\rightarrow\infty$. 
For $\theta\le 0$ (the shaded region), the rejection probabilities are under the null and thus should be less than or equal to $\alpha=5\%$. 
For $\theta>0$, the rejection probabilities represent the power of the tests. 
Figure \ref{Fig:powercurve_simple_sec} reports the power curves for $q\in\{0,4\}$ and $J\in\{3,10,50\}$. 
When $q=0$, the 3rd to the $J$th inequalities are binding at $\delta=1$, and when $q=4$, these inequalities are slack at $\delta=1$. 

As we can see in Figure \ref{Fig:powercurve_simple_sec}: \medskip

(1) {The GCC and RGCC tests are very fast with good power for all $J\in\{3,10,50\}$ and $q\in\{0,4\}$. The power is especially impressive when $q=4$, so $J-2$ of the inequalities are slack. 
Also note that, when the number of inequalities increases, the difference between the GCC and RGCC power curves decreases, especially when all the inequalities bind.}

(2) {The BCS test has good power for $J=3$, but becomes more conservative, and therefore less powerful, for larger $J$.
The Bei test is more powerful than BCS, GCC, and RGCC, with especially high power when $q=0$ and $J\in\{10,50\}$. These specifications correspond to the cases of moderate over-rejection of the Bei test at $\theta=0$.}

(3) {The sensitivity that the CR tests show in Table \ref{Simple_Times} is reflected in power. CR1 has low power, and the power of CR2 is closely related to the over-rejection of the test at $\theta=0$.}

\subsection{An Interval Outcome IV Regression Model}\label{sub:IOR}

This subsection considers the interval outcome IV regression model from Example \ref{ex:interval}. 
We simulate the power curves of the GCC and RGCC tests. 
We also include the BCS and Bei tests as benchmarks. 
Since this is a two-sided problem, we do not include recommendations from the literature on inference for the value of a LPP. 

The model is  based on the aggregate demand model considered in GLS23, where the market shares are noisy measures of conditional choice probabilities and may contain zero values. The model boils down to the interval outcome IV regression in Example \ref{ex:interval}. 
Write out $X=(X_1, X_2, W')'$, where $X_1$ is a scalar endogenous regressor, $X_2$ is a scalar exogenous regressor, and $W$ is a $d_W$-vector of additional exogenous controls. 
Similarly, write out $\beta=(\theta_1, \theta_2, \gamma')'$ with $\gamma\in\R^{d_W}$. 
Also let $Z_e$ be an excluded exogenous instrument. 
We take $Z=\mathcal{I}(X_2, W, Z_e)$ to be a non-negative vector-valued instrumental function. 

In this model, there is a latent market share that satisfies a logit specification: 
\begin{equation}
s^\ast=\frac{\exp\left(\theta_{10}X_1+\theta_{20}X_2+\epsilon\right)}{1+\exp\left(\theta_{10}X_1+\theta_{20}X_2+\epsilon\right)}, \label{IV_reg_sim_1}
\end{equation}
where $\theta_{10}$ and $\theta_{20}$ denote the true values of $\theta_1$ and $\theta_2$ and $\epsilon$ is an error term.  
(The true value of $\gamma$ is zero.) 
In the framework of GLS23, $s^\ast$ is the (unobserved) conditional probability that a representative consumer buys a product, and $Y^\ast=\log(s^\ast)-\log(1-s^\ast)$ is the mean utility of the  product. 
We observe a market share $s_N\sim Binomial(N,s^\ast)/N$, where $N$ is the number of participants in the market. 
GLS23 argue that $Y^U = \log(s_N+2/N)-\log(1-s_N+0.00125)$ and $Y^L = \log(s_N+0.00125)-\log(1-s_N+2/N)$ satisfy $\mathbb{E}[Y^L|Z]\leq \mathbb{E}[Y^\ast|Z]\leq \mathbb{E}[Y^U|Z]$, which justifies the use of the model in Example \ref{ex:interval}. 

We take $X_2$, $Z_e$, and the components of $W$ to be independent $Bernoulli(0.5)$ random variables, except for the first element of $W$, which is taken to be the constant one. 
We take $Z=\mathcal{I}(X_2,W,Z_e)$ to be the vector of indicators for each point in the support of $(X_2,W,Z_e)$. 
When $d_W=1,2,3$, the dimension of $Z$ is $4,8,16$, respectively. 
Also, the fact that there are two inequalities in (\ref{uncmi}) means that there are $8$, $16$, and $32$ inequalities, respectively. 
Independently of $Z$, let $\epsilon\sim \min\left(\max\left(N(0,1),-4\right),4\right)$. 
Then, let $X_1=\mathds{1}\{Z_e+\epsilon/2>0\}$ be the endogenous regressor. 
We calculate $s^\ast$ according to (\ref{IV_reg_sim_1}) with $\theta_{10}=\theta_{20}=-1$. 
We simulate $s_N$ with $N=100$ independently from $Z$ and $\epsilon$. 
This specifies the data generating process for all the observed variables: $s_N$, $X_1$, $X_2$, $W$, and $Z_e$. 
We simulate a sample of size $n$ from this model for $n\in\{500, 1000\}$ and calculate confidence intervals for $\theta_2$ treating $\delta=(\theta_1,\gamma')'$ as nuisance parameters.\footnote{The same model is considered in Section 5.2 of CS23. However, CS23 construct confidence intervals for $\theta_1$, treating $(\theta_2,\gamma')'$ as nuisance parameters. Since $X_2$ and $W$ are both exogenous, it is valid to conduct inference conditional on $(X_2,W,Z_e)$ using the sCC test in CS23 because the Jacobian of the sample moments with respect to $(\theta_2,\gamma')'$ is known given the sample for $(X_2,W,Z_e)$. In contrast, this paper takes $(\theta_1,\gamma')'$ to be the nuisance parameters. Then, the Jacobian with respect to $\theta_1$ is not known even after conditioning on $(X_2,W,Z_e)$. Thus, the sCC test in CS23 is invalid.} 

Figure \ref{Fig:powercurve_gen} plots simulated power curves for the GCC, RGCC, BCS, and Bei tests. 
In each graph, the horizontal axis represents the value of $\theta_2$, while the vertical axis represents the rejection probability.\footnote{The rejection probabilities reported are frequencies that each $\theta_2$ value lies outside the confidence interval for $\theta_2$. The confidence intervals are computed using a bisection algorithm for the endpoints.} 
The shaded region indicates the identified set for $\theta_2$. 
Since the control variables have zero coefficients and are independent of the other random variables, they do not affect the identified set of $\theta_2$. 
In the legend, each number in the square brackets is the median computational time (in seconds) to compute one confidence interval. 
We make the following remark on Figure \ref{Fig:powercurve_gen}. 

\begin{figure}
	\caption{Power Curves and Computation Times of Various Tests of Nominal Level 5\% in the Interval Outcome IV Regression Model}
	\label{Fig:powercurve_gen}
	\begin{center}
		\begin{threeparttable}
			\begin{subfigure}[b]{0.475\textwidth}
				\begin{center}
					\begin{tikzpicture}
						\begin{axis}
							[width=0.79\textwidth,
							ymin=0,ymax=1,xmin=-1.55,xmax=-0.45,
							xticklabel style={font=\scriptsize}, 
							yticklabel style={font=\scriptsize}, 
							legend style={at={(axis cs:-0.45,0.985)},anchor=north west}]
							\draw [pattern=north west lines, pattern color=gray!70,draw=none] (-1.0986,0) rectangle (-0.8613,0.996);
							\draw[dotted] (-1.65,0.05) -- (-.35,0.05);
							\addplot[very thick] table [x=theta, y=rr_GCC, col sep=semicolon]				
							{IntIV_Rej_ALL_5000_dx_1_n_500.txt};
							\addplot[dashed, very thick, gray] table [x=theta, y=rr_RGCC, col sep=semicolon]
							{IntIV_Rej_ALL_5000_dx_1_n_500.txt};
							\addplot[mark=*, mark size=0.7pt, mark repeat={5}, thick, red] table [x=theta, y=rr_BCS, col sep=semicolon] 				
							{IntIV_Rej_ALL_5000_dx_1_n_500.txt};
							\addplot[mark=x, mark repeat={5}, thick, purple] table [x=theta, y=rr_Bei, col sep=semicolon]				
							{IntIV_Rej_ALL_5000_dx_1_n_500.txt};
							\addlegendentry{{\tiny GCC  [0.07s]}};
							\addlegendentry{{\tiny RGCC test [0.07s]}};
							\addlegendentry{{\tiny BCS  [96.54s]}};
							\addlegendentry{{\tiny Bei  [4.00s]}};
						\end{axis}
					\end{tikzpicture}
					\caption{$n=500$, $d_\delta=2$, $d_C = 8$}
					\label{Fig:powercurve_gen_n500_delta2}
				\end{center}		 
			\end{subfigure}
			\hfill \hspace{0.2in}
			\begin{subfigure}[b]{0.475\textwidth}  
				\begin{center}
					\begin{tikzpicture}
						\begin{axis}
							[width=0.79\textwidth,
							ymin=0,ymax=1,xmin=-1.55,xmax=-.45,
							xticklabel style={font=\scriptsize}, 
							yticklabel style={font=\scriptsize}, 
							legend style={at={(axis cs:-.45,0.985)},anchor=north west}]
							\draw [pattern=north west lines, pattern color=gray!70,draw=none] (-1.0986,0) rectangle (-0.8613,0.996);
							\draw[dotted] (-2,0.05) -- (0,0.05);
							\addplot[very thick] table [x=theta, y=rr_GCC, col sep=semicolon]				
							{IntIV_Rej_ALL_5000_dx_1_n_1000.txt};
							\addplot[dashed, very thick, gray] table [x=theta, y=rr_RGCC, col sep=semicolon]
							{IntIV_Rej_ALL_5000_dx_1_n_1000.txt};
							\addplot[mark=*, mark size=0.7pt, mark repeat={5}, thick, red] table [x=theta, y=rr_BCS, col sep=semicolon] 				
							{IntIV_Rej_ALL_5000_dx_1_n_1000.txt};
							\addplot[mark=x, mark repeat={5}, thick, purple] table [x=theta, y=rr_Bei, col sep=semicolon]				
							{IntIV_Rej_ALL_5000_dx_1_n_1000.txt};
							\addlegendentry{{\tiny GCC  [0.07s]}};
							\addlegendentry{{\tiny RGCC test [0.07s]}};
							\addlegendentry{{\tiny BCS  [142.69s]}};
							\addlegendentry{{\tiny Bei  [4.13s]}};
						\end{axis}
					\end{tikzpicture}
					\caption{$n=1000$, $d_\delta=2$, $d_C=8$}
					\label{Fig:powercurve_gen_n1000_delta2}	
				\end{center}
			\end{subfigure}
			\vskip\baselineskip
			\begin{subfigure}[b]{0.475\textwidth}   
				\begin{center}
					\begin{tikzpicture}
						\begin{axis}
							[width=0.79\textwidth,
							ymin=0,ymax=1,xmin=-1.55,xmax=-.45,
							xticklabel style={font=\scriptsize}, 
							yticklabel style={font=\scriptsize}, 
							legend style={at={(axis cs:-.45,0.985)},anchor=north west}]
							\draw [pattern=north west lines, pattern color=gray!70,draw=none] (-1.0986,0) rectangle (-0.8613,0.996);
							\draw[dotted] (-2,0.05) -- (0,0.05);
							\addplot[very thick] table [x=theta, y=rr_GCC, col sep=semicolon]			
							{IntIV_Rej_ALL_5000_dx_2_n_500.txt};
							\addplot[dashed, very thick, gray] table [x=theta, y=rr_RGCC, col sep=semicolon]
							{IntIV_Rej_ALL_5000_dx_2_n_500.txt};
							\addplot[mark=*, mark size=0.7pt, mark repeat={5}, thick, red] table [x=theta, y=rr_BCS, col sep=semicolon] 				
							{IntIV_Rej_ALL_5000_dx_2_n_500.txt};
							\addplot[mark=x, mark repeat={5}, thick, purple] table [x=theta, y=rr_Bei, col sep=semicolon]				
							{IntIV_Rej_ALL_5000_dx_2_n_500.txt};
							\addlegendentry{{\tiny GCC  [0.09s]}};
							\addlegendentry{{\tiny RGCC test [0.09s]}};
							\addlegendentry{{\tiny BCS  [311.16s]}};
							\addlegendentry{{\tiny Bei  [5.51s]}};
						\end{axis}
					\end{tikzpicture}
					\caption{$n=500$, $d_\delta=3$, $d_C=16$}
					\label{Fig:powercurve_gen_n500_delta3}
				\end{center}
			\end{subfigure}
			\hfill
			\begin{subfigure}[b]{0.475\textwidth}   
				\begin{center}
					\begin{tikzpicture}
						\begin{axis}
							[width=0.79\textwidth,
							ymin=0,ymax=1,xmin=-1.55,xmax=-0.45,
							xticklabel style={font=\scriptsize}, 
							yticklabel style={font=\scriptsize}, 
							legend style={at={(axis cs:-.45,0.985)},anchor=north west}]
							\draw [pattern=north west lines, pattern color=gray!70,draw=none] (-1.0986,0) rectangle (-0.8613,0.996);
							\draw[dotted] (-2,0.05) -- (0,0.05);
							\addplot[very thick] table [x=theta, y=rr_GCC, col sep=semicolon]					
							{IntIV_Rej_ALL_5000_dx_2_n_1000.txt};
							\addplot[dashed, very thick, gray] table [x=theta, y=rr_RGCC, col sep=semicolon]
							{IntIV_Rej_ALL_5000_dx_2_n_1000.txt};
							\addplot[mark=*, mark size=0.7pt, mark repeat={5}, thick, red] table [x=theta, y=rr_BCS, col sep=semicolon] 				
							{IntIV_Rej_ALL_5000_dx_2_n_1000.txt};
							\addplot[mark=x, mark repeat={5}, thick, purple] table [x=theta, y=rr_Bei, col sep=semicolon]				
							{IntIV_Rej_ALL_5000_dx_2_n_1000.txt};
							\addlegendentry{{\tiny GCC  [0.10s]}};
							\addlegendentry{{\tiny RGCC test [0.10s]}};
							\addlegendentry{{\tiny BCS  [501.00s]}};
							\addlegendentry{{\tiny Bei  [6.07s]}};
						\end{axis}
					\end{tikzpicture}
					\caption{$n=1000$, $d_\delta=3$, $d_C=16$}
					\label{Fig:powercurve_gen_n1000_delta3}
				\end{center}
			\end{subfigure}
			\vskip\baselineskip
			\begin{subfigure}[b]{0.475\textwidth}   
				\begin{center}
					\begin{tikzpicture}
						\begin{axis}
							[width=0.79\textwidth,
							ymin=0,ymax=1,xmin=-1.55,xmax=-0.45,
							xticklabel style={font=\scriptsize}, 
							yticklabel style={font=\scriptsize}, 
							legend style={at={(axis cs:-.45,0.985)},anchor=north west}]
							\draw [pattern=north west lines, pattern color=gray!70,draw=none] (-1.0986,0) rectangle (-0.8613,0.996);
							\draw[dotted] (-2,0.05) -- (0,0.05);
							\addplot[very thick] table [x=theta, y=rr_GCC, col sep=semicolon]				
							{IntIV_Rej_ALL_5000_dx_3_n_500.txt};
							\addplot[dashed, very thick, gray] table [x=theta, y=rr_RGCC, col sep=semicolon]
							{IntIV_Rej_ALL_5000_dx_3_n_500.txt};
							\addplot[mark=*, mark size=0.7pt, mark repeat={5}, thick, red] table [x=theta, y=rr_BCS, col sep=semicolon] 				
							{IntIV_Rej_ALL_5000_dx_3_n_500.txt};
							\addplot[mark=x, mark repeat={5}, thick, purple] table [x=theta, y=rr_Bei, col sep=semicolon]				
							{IntIV_Rej_ALL_5000_dx_3_n_500.txt};
							\addlegendentry{{\tiny GCC  [0.16s]}};
							\addlegendentry{{\tiny RGCC test [0.15s]}};
							\addlegendentry{{\tiny BCS  [960.48s]}};
							\addlegendentry{{\tiny Bei  [9.17s]}};
						\end{axis}
					\end{tikzpicture}
					\caption{$n=500$, $d_\delta=4$, $d_C=32$}
					\label{Fig:powercurve_gen_n500_delta4}
				\end{center}
			\end{subfigure}
			\hfill
			\begin{subfigure}[b]{0.475\textwidth}   
				\begin{center}
					\begin{tikzpicture}
						\begin{axis}
							[width=0.79\textwidth,
							ymin=0,ymax=1,xmin=-1.55,xmax=-.45,
							xticklabel style={font=\scriptsize}, 
							yticklabel style={font=\scriptsize}, 
							legend style={at={(axis cs:-.45,0.985)},anchor=north west}]
							\draw [pattern=north west lines, pattern color=gray!70,draw=none] (-1.0986,0) rectangle (-0.8613,0.996);
							\draw[dotted] (-2,0.05) -- (0,0.05);
							\addplot[very thick] table [x=theta, y=rr_GCC, col sep=semicolon]					
							{IntIV_Rej_ALL_5000_dx_3_n_1000.txt};
							\addplot[dashed, very thick, gray] table [x=theta, y=rr_RGCC, col sep=semicolon]
							{IntIV_Rej_ALL_5000_dx_3_n_1000.txt};
							\addplot[mark=*, mark size=0.7pt, mark repeat={5}, thick, red] table [x=theta, y=rr_BCS, col sep=semicolon] 				
							{IntIV_Rej_ALL_5000_dx_3_n_1000.txt};
							\addplot[mark=x, mark repeat={5}, thick, purple] table [x=theta, y=rr_Bei, col sep=semicolon]				
							{IntIV_Rej_ALL_5000_dx_3_n_1000.txt};
							\addlegendentry{{\tiny GCC  [0.18s]}};
							\addlegendentry{{\tiny RGCC test [0.17s]}};
							\addlegendentry{{\tiny BCS  [1717.74s]}};
							\addlegendentry{{\tiny Bei  [10.37s]}};
						\end{axis}
					\end{tikzpicture}
					\caption{$n=1000$, $d_\delta=4$, $d_C=32$}
					\label{Fig:powercurve_gen_n1000_delta4}
				\end{center}
			\end{subfigure}
		\end{threeparttable}
	\end{center}
{\small {\em Note:} The shaded regions indicate the identified set for $\theta_2$. The number of simulations is 5000. The numbers in the square brackets in the legends are the median time needed to calculate the confidence interval for one simulation (measured in seconds).} 
\end{figure}
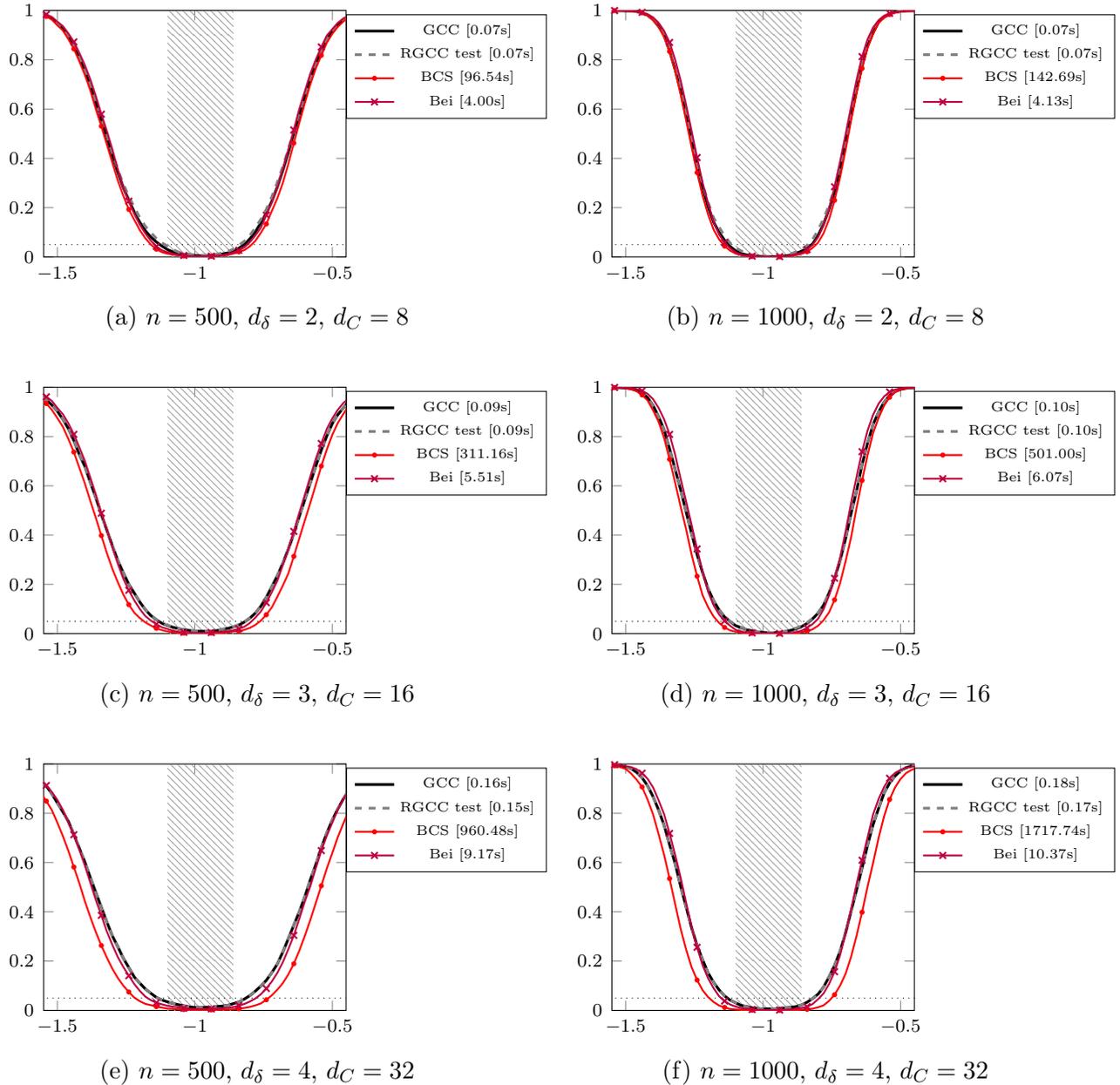

\addtocounter{figure}{-1} % I don't know why, but for some reason latex thinks the above figure is two figures. 

\noindent\textbf{Remark:} 
\textit{The power curves of the GCC, RGCC, BCS, and Bei tests are remarkably similar. 
All four have rejection probabilities below the nominal level 5\% in the shaded region. 
All four have increasing power as $\theta_2$ deviates from its identified set and as the sample size increases from 500 to 1000. 
Overall, the GCC and RGCC tests are able to match the size and power performance of the BCS and Bei tests while being much faster computationally.} 

\section{Empirical Illustration: Female Labor Supply}
\label{sec:empirical_KT}
This section demonstrates the GCC test in the model of female labor supply considered in KT16. 
KT16 use a model of state transition probabilities with states that indicate (1) the earnings of the individual (zero or not employed, positive and below the federal poverty line, or above the federal poverty line), (2) whether the individual participates in welfare, and (3) whether the individual under-reports her earnings in order to qualify for welfare. 
KT16 use data from the Manpower Development Research Corporation (MDRC) Jobs First study, which is a random experiment that assigned women with young children to one of two welfare programs: the Aid to Families with Dependent Children (AFDC) welfare program or the Jobs First Temporary Family Assistance (JF) program.\footnote{Instructions for accessing the datasets and the replication codes provided by KT16 can be found on the AEA webpage: \url{https://www.aeaweb.org/articles?id=10.1257/aer.20130824}.} 

Relative to AFDC, the JF program primarily changes how eligibility for and the amount of government transfers respond to earnings.\footnote{The JF welfare reform introduces other changes too, such as stricter work requirements and changed administration of the Food Stamps program. We focus on the changes that are relevant for bounding the state transition probabilities. For more information, see the description in KT16.} 
Under JF, for individuals earning under the federal poverty line (FPL), the government transfer does not decrease as earnings increase: everyone under the FPL receives the same government transfer. 
This is more generous than the policy under AFDC, which had government transfers decrease after earnings reached a threshold. 
This more generous policy may induce unemployed women (or women who would be unemployed under AFDC) to gain employment that pays below the FPL. This represents a labor supply response along the extensive margin. 
Another noteworthy feature of the JF program is that the government transfers abruptly drop to zero when earnings cross the FPL. 
This may incentivize some women who would be employed with earnings above the FPL (under AFDC) to decrease their earnings (or under-report their earnings) in order to be eligible for welfare. This represents a labor supply response along the intensive margin. 
The goal in KT16 is to distinguish these two responses without imposing strong assumptions on the utility functions of the individuals. 

KT16 distinguish 7 labor supply/welfare participation states: 
\begin{enumerate}
\item[0n:] zero earnings, welfare nonparticipation, 
\item[1n:] positive earnings below FPL, welfare nonparticipation, 
\item[2n:] earnings above FPL, welfare nonparticipation, 
\item[0r:] zero earnings, welfare participation, truthful reporting of earnings, 
\item[1r:] positive earnings below FPL, welfare participation, truthful reporting of earnings, 
\item[1u:] positive earnings below FPL, welfare participation, underreporting of earnings, and 
\item[2u:] earnings above FPL, welfare participation, underreporting of earnings. 
\end{enumerate}
In the label for each state, the number indicates the level of earnings: ``0'' indicates zero earnings or not employed, ``1'' indicates positive earnings below the FPL, and ``2'' indicates earnings above the FPL. 
The letter indicates welfare participation and under-reporting of earnings: ``n'' indicates nonparticipation in welfare, ``r'' indicates welfare participation with truthful reporting of earnings, and ``u'' indicates welfare participation with under-reporting of earnings. 

Each individual is associated with two states, the one she would choose under AFDC and the one she would choose under JF. 
Because AFDC is the status quo, we label the state transition probabilities with an individual's choice under AFDC first. 
For example, $\pi_{\text{0n},\text{1r}}$ is the conditional probability that an individual who chooses 0n under AFDC would choose 1r under JF. 
Because of the features of the JF reform, KT16 argue that only nine state transition probabilities need to be considered. 
The first group are flows out of 0r: $(\pi_{\text{0r},\text{0n}}, \pi_{\text{0r},\text{1n}}, \pi_{\text{0r},\text{2n}}, \pi_{\text{0r},\text{1r}}, \pi_{\text{0r},\text{2u}})$. 
Individuals with zero earnings and participating in welfare under AFDC may transition to any other state except 1u.\footnote{Under JF, no one chooses 1u because everyone earning under the FPL receives the same government transfer, so there is no reason to under-report earnings. Also note that $\pi_{\text{0r},\text{0r}}$ is not needed because it can be calculated as one minus the others. This is true in general. We do not need to include the transition probabilities from one state into itself because they are determined by the other transition probabilities.} 
The second group are flows into 1r: $(\pi_{\text{0n},\text{1r}}, \pi_{\text{1n},\text{1r}}, \pi_{\text{2n},\text{1r}}, \pi_{\text{0r},\text{1r}}, \pi_{\text{2u},\text{1r}})$. 
Individuals who choose 1r under JF may have chosen any state under AFDC.\footnote{Individuals who choose 1u under AFDC are guaranteed to choose 1r under JF, so $\pi_{\text{1u},\text{1r}}=1$ and there is no need to include it as a free parameter.} 
KT16 argue that all other transition probabilities can be set to zero either because the budget sets for the individuals is unchanged between the two policies or because the combination of choices would violate weak assumptions on individuals' utility functions. 
Let 
\[
\delta = (\pi_{\text{0n},\text{1r}}, \pi_{\text{1n},\text{1r}}, \pi_{\text{2n},\text{1r}}, \pi_{\text{0r},\text{0n}}, \pi_{\text{0r},\text{1n}}, \pi_{\text{0r},\text{2n}}, \pi_{\text{0r},\text{1r}}, \pi_{\text{0r},\text{2u}}, \pi_{\text{2u},\text{1r}})'
\]
collect the transition probabilities into a vector of nuisance parameters. (Note that one of the transition probabilities is common to both groups.) 

There is an additional problem: it is unobserved whether an individual under-reports her income. 
Thus, the marginal probabilities of only six observable states---combining 1r and 1u---is identified. 
KT16 show that after accounting for this problem, the resulting equalities are still linear in $\delta$ and they can be rewritten into five non-redundant equalities of the form $\Gamma\delta = m$, where 
\begin{equation}\label{eq:KlineTartari_eq}
\Gamma = \left(\begin{smallmatrix}-p_{\text{0n}}^{\text{A}}&0&0&p_{\text{0r}}^{\text{A}}&0&0&0&0&0\\
                                                            0&-p_{\text{1n}}^{\text{A}}&0&0&p_{\text{0r}}^{\text{A}}&0&0&0&0\\
                                                            0&0&-p_{\text{2n}}^{\text{A}}&0&0&p_{\text{0r}}^{\text{A}}&0&0&0\\
                                                            0&0&0&-p_{\text{0r}}^{\text{A}}&-p_{\text{0r}}^{\text{A}}&-p_{\text{0r}}^{\text{A}}&-p_{\text{0r}}^{\text{A}}&-p_{\text{0r}}^{\text{A}}&0\\
                                                            0&0&0&0&0&0&0&p_{\text{0r}}^{\text{A}}&-p_{\text{2u}}^{\text{A}}\end{smallmatrix}\right)
\text{ and }
m=\left(\begin{smallmatrix}p_{\text{0n}}^{\text{J}}-p_{\text{0n}}^{\text{A}}\\p_{\text{1n}}^{\text{J}} - p_{\text{1n}}^{\text{A}}\\p_{\text{2n}}^{\text{J}}-p_{\text{2n}}^{\text{A}}\\p_{\text{0r}}^{\text{J}}-p_{\text{0r}}^{\text{A}}\\p_{\text{2u}}^{\text{J}}-p_{\text{2u}}^{\text{A}}\end{smallmatrix}\right), 
\end{equation}
where $p_{\text{s}}^{\text{A}}$ denotes the marginal probability that the individual chooses state s under the AFDC policy for $s\in\{\text{0n}, \text{1n}, \text{2n}, \text{0r}, \text{2u}\}$, and similarly for $p_{\text{s}}^{\text{J}}$ for the JF policy. 
In addition, $\delta \in [0,1]^9$ and $\pi_{\text{0r},\text{0n}}+\pi_{\text{0r},\text{1n}}+\pi_{\text{0r},\text{2n}}+\pi_{\text{0r},\text{1r}}+\pi_{\text{0r},\text{2u}}\leq 1$, which can be imposed by an appropriate choice of $A$ and $b$ in (\ref{deltaID}).\footnote{In Section \ref{sec:verify}, we give a simple sufficient condition for Assumption \ref{assu:rank:simple} in this model.} 

We follow KT16 and report confidence intervals for each transition probability. 
In addition to the GCC test, we also implement the RGCC test, defined in Section \ref{app:refinement}. 
The GCC and RGCC tests are implemented by rewriting the restrictions in terms of $B$, $\mu$, $\Pi$, $D$, and $d$ using (\ref{known_gamma_translation}). We can similarly define estimators of $\mu$ and $\Pi$ from the sample averages that estimate $p_s^t$ for $s\in\{\text{0n}, \text{1n}, \text{2n}, \text{0r}, \text{2u}\}$ and $t\in\{\text{A},\text{J}\}$.\footnote{We copy KT16 and use weighted sample averages with propensity score weights to adjust for baseline differences.} 
We estimate the asymptotic variance by bootstrapping the sample averages with a cluster bootstrap, clustered at the case level and with 1000 bootstrap draws.\footnote{This is the same implementation of the bootstrap that KT16 use, except that we bootstrap the estimators of $p_s^t$ while they bootstrap the formulas for the bounds that they calculate after eliminating the nuisance parameters manually. This means that our variance estimators, while asymptotically equivalent, are numerically different.} 
The endpoints of the GCC and RGCC confidence intervals are calculated using a bisection algorithm. 

KT16 manually eliminate the nuisance parameters and work out explicit formulas for the bounds of each element of $\delta$. 
To avoid the dependence on tuning parameters that is prevalent in the literature on testing inequalities, they report two confidence intervals. 
One confidence interval, called \textit{Na\"ive}, is constructed by ignoring the uncertainty in which bounds bind. This interval is formed by the single lowest upper (highest lower) estimated bound plus (minus) its standard error. 
The asymptotic coverage probability of this interval is unknown. 
The other interval, called \textit{Conservative}, assumes all population bounds bind simultaneously, leading to asymptotically valid but often overly conservative inference.  

\begin{table}
	\begin{center}	
		\begin{threeparttable}
			\caption{Confidence Intervals for Transition Probabilities}
			\label{tab:KlineTartari2016}								
				\begin{tabular}{cccccc}
					\hline\hline
					& & \multicolumn{4}{c}{ 95\% CI } \\
					\cline { 3 - 6 }
					$\theta$ & Estimated bound & GCC &RGCC& Na\"ive & Conservative \\
					\hline
					$\pi_{\text{0n},\text{1r}}$ & $[0.055,~0.620]$ & $[0.000, ~0.764]$ &$[0.000, ~0.741]$& $[0.000, ~0.740]$ & $[0.000, ~0.782]$ \\
					$\pi_{\text{1n},\text{1r}}$ & $[0.382,~ 0.987]$ & $[0.303,~ 1.000]$ &$[0.316, ~1.000]$& $[0.318, ~1.000]$ & $[0.318, ~1.000]$ \\
					$\pi_{\text{2n},\text{1r}}$ & $[0.280,~ 1.000]$ & $[0.171,~ 1.000]$ &$[0.189, ~1.000]$& $[0.193, ~1.000]$ & $[0.193, ~1.000]$ \\
					$\pi_{\text{0r},\text{0n}}$ & $[0.000,~ 0.170]$ & $[0.000, ~0.210]$ &$[0.000, ~0.204]$& $[0.000, ~0.204]$ & $[0.000,~ 0.215]$ \\
					$\pi_{\text{0r},\text{1n}}$ & $[0.000, ~0.170]$ & $[0.000, ~0.195]$ & $[0.000, ~0.195]$& $[0.000,~0.211]$ & $[0.000, ~0.215]$ \\
					$\pi_{\text{0r},\text{2n}}$ & $[0.000, ~0.154]$ & $[0.000, ~0.179]$ &$[0.000, ~0.179]$ & $[0.000, ~0.171]$ & $[0.000, ~0.226]$ \\
					$\pi_{\text{0r},\text{1r}}$ & $[0.000, ~0.170]$ & $[0.000,~ 0.210]$ &$[0.000,~ 0.204]$& $[0.000, ~0.204]$ & $[0.000,~ 0.220]$ \\
					$\pi_{\text{0r},\text{2u}}$ & $[0.031,~ 0.051]$ & $[0.020,~ 0.060]$ &$[0.022, ~0.059]$& $[0.022,~ 0.059]$ & $[0.022, ~0.099]$ \\
					$\pi_{\text{2u},\text{1r}}$ & $[0.000, ~1.000]$ & $[0.000,~ 1.000]$ &$[0.000, ~1.000]$& $[0.000, ~1.000]$ & $[0.000, ~1.000]$ \\
					\hline
				\end{tabular}
				\begin{footnotesize}
					\begin{tablenotes}
						\item{Note:}
						``Estimated bound'' denotes an estimator of the bounds derived in KT16. 
						``GCC'' and ``RGCC'' refer to confidence intervals formed by inverting the GCC and RGCC tests using bisection. 
						``Na\"{i}ve'' and ``Conservative'' refer to the confidence intervals reported in Table 5 in KT16.
					\end{tablenotes}
				\end{footnotesize}	
		\end{threeparttable}
	\end{center}
\end{table}
Table \ref{tab:KlineTartari2016} reports the confidence intervals for each transition probability.\footnote{The values for the Na\"ive and Conservative confidence intervals are slightly different from the ones in the published version of KT16. We calculated these values using the KT16 replication code without changes. The differences are likely due to differences in random number generation across versions of Stata when implementing the bootstrap.} As Table \ref{tab:KlineTartari2016} shows: 

(1) { All the confidence intervals are qualitatively similar. 
They provide evidence for the same heterogeneous labor supply responses: statistically significant outflows from state 0r, corresponding to an increase in labor supply along the extensive margin, and statistically significant inflows into state 1r, especially from state 2n, corresponding to a decrease in labor supply along the intensive margin as women decrease their earnings to qualify for welfare. }

(2) { One would expect the endpoints of the GCC and RGCC confidence intervals to lie between the endpoints of the Na\"{i}ve and Conservative confidence intervals. 
That is mostly true, but there are a few noteworthy exceptions. }
(a) { For $\pi_{\text{0r},\text{1n}}$, the GCC and RGCC confidence intervals are narrower than the Na\"{i}ve confidence interval. 
In this case, the two smallest upper bounds are very close to each other. Together, they provide stronger statistical evidence than just the one that is active. 
The GCC and RGCC tests respond to this statistical evidence in a way that the Na\"{i}ve confidence interval does not. }
(b) {For $\pi_{\text{1n},\text{1r}}$ and $\pi_{\text{2n},\text{1r}}$, the GCC and RGCC confidence intervals are wider than the Conservative confidence intervals. 
This is not surprising for the GCC confidence interval because the GCC test is conservative when there is one binding inequality.\footnote{For all the transition probabilities in Table \ref{tab:KlineTartari2016}, there is only one nontrivial lower bound. This explains both why the Na\"{i}ve lower bounds are equal to the Conservative lower bounds and why the GCC confidence interval appears conservative for the lower bounds.} 
This is surprising for the RGCC confidence interval because with one binding inequality, the RGCC test is asymptotically equivalent to the optimal one-sided test. This is likely due to the numerical difference between the variance matrix estimators used in the RGCC and the Conservative confidence intervals.} 

(3) {To compute all nine confidence intervals, the GCC and RGCC methods took about 4 seconds and 220 seconds, respectively. 
Both approaches are quite feasible, especially considering the fact that manual elimination of the nuisance parameters is not needed to implement the GCC and RGCC tests. }

\section{Conclusion}
\label{sec:conclusion}
This paper proposes a simple, tuning-parameter-free test that is designed for inequality testing problems that are linear in nuisance parameters, including specification testing and subvector inference in moment (in)equality models and inference for parameters bounded by linear programs. 
We prove asymptotic uniform validity of the test under a stable rank condition and demonstrate its size, power, and computational performance in simulations and an empirical illustration. 

\appendix
\section{Proofs of Theorems \ref{lem:rhat}-\ref{thm:level_CC}}

The proofs of Theorems \ref{lem:rhat}-\ref{thm:level_CC} rely on lemmas that are stated and proved in Sections \ref{Section_Generalizations} and \ref{UsefulLemmas} in the Supplemental Appendix. 
The proofs also use the following notation/definitions. 
For a matrix $C$ and conformable vector $b$, let $\poly(C,b)=\{\delta\in\R^{d_\delta}: C\delta\le b\}$ be the polyhedral set defined by the inequalities with coefficients $C$ and intercepts $b$. 
For matrices $B$ and $C$ and a conformable vector $d$ let $\ppoly(B,d;C)=\{\mu\in\R^{d_\mu}: B\mu+C\delta\le d \text{ for some }\delta\in\R^{d_\delta}\}$ be the projection of a polyhedral set onto a subvector. 
A sequence of sets, $S_n$, Kuratowski converges to a limit set, $S_\infty$, denoted by $S_n\overset{K}{\rightarrow}S_\infty$, if (1) for every $x\in S_\infty$ there exists a sequence $x_n\rightarrow x$ such that $x_n\in S_n$ eventually, and (2) for every subsequence, $n_m$, and for every converging sequence $x_m\in S_{n_m}$ with limit $x_\infty$, we have $x_\infty\in S_\infty$. 
This definition of Kuratowski convergence agrees with common definitions in the setwise analysis literature as found in, for example, Definition 1.1.1 in \cite{AubinFrankowska2009}. 
The proofs also use definitions of activatable sets of indices for a collection of linear inequalities that are explained in Appendix \ref{activatability}. 

\subsection{Proof of Theorem \ref{lem:rhat}}\label{sec:proof_rhat}

By Lemma 1 in CS23, for the matrix $H$ whose rows are the vertices of the polytope $\{h\in {\mathbb R}^{d_C}|h\geq 0,h'\overline{C}_n =0, h'\mathbf{1} = 1\}$, and for $A=HB$ and $g=Hd$, 
\begin{equation}
\ppoly(B,d;\overline{C}_n)=\poly(A,g). 
\label{AbBCd_equivalence_new}
\end{equation} 
Let $d_A$ denote the number of rows in $A$. 
Let $\widehat{J}=\{j\in\{1,...,d_A\}: e'_j A\widehat\mu_n=e'_j g\}$. 
Also, for any $J\subseteq\{1,...,d_A\}$, let $A_J$ denote the submatrix of $A$ formed from the rows corresponding to the indices in $J$. 
Lemma 2 and its proof in CS23 show that $\widehat r_n=\textup{rk}(A_{\widehat J})$.

The KKT conditions associated with the CQPP in (\ref{QLRnew}) are: 
\begin{align}
2n\widetilde{\Sigma}_n^{-1}(\overline{\mu}_n-\widehat\mu_n)&=B'\widehat\psi_n\label{BCd_KKT_new}\\
0&=\overline{C}'_n\widehat\psi_n\label{BCd_KKT2_new}\\
\widehat\psi_n&\ge 0\label{BCd_nonnegativemultipliers_new}\\
d&\ge B\widehat\mu_n+\overline{C}_n\widehat\delta_n\label{BCd_KKT4_new}\\
0&=(d-B\widehat\mu_n-\overline{C}_n\widehat\delta_n)'\widehat\psi_n, \label{BCd_ComplementarySlackness_new}
\end{align}
where $\widehat\psi_n$ are the KKT multipliers on $B\mu+\overline{C}_n\delta\le d$. 
For CQPPs, the KKT conditions are necessary and sufficient; see Chapter 16 in \cite{NocedalWright2006}. 
If, instead, we write the constraints as $A\mu\le g$ with corresponding KKT multipliers $\widehat\lambda_n$, then the KKT conditions are: 
\begin{align}
2n\widetilde\Sigma^{-1}_n(\overline{\mu}_n-\widehat\mu_n)&=A'\widehat\lambda_n\label{Ab_KKT_new}\\
\widehat\lambda_n&\ge 0\label{Ab_KKT2_new}\\
g&\ge A\widehat\mu_n\\
0&=(g-A\widehat\mu_n)'\widehat\lambda_n. \label{Ab_ComplementarySlackness_new}
\end{align}
These are also necessary and sufficient for the same reason. 

Let $C_{K}$ be shortened notation for $I_{K}\overline{C}_n$ for any $K\subseteq\{1,...,d_C\}$. 

\subsubsection{Proof of Theorem \ref{lem:rhat}(a)}

We first prove that $\widehat t_n\le \widehat r_n$. 
Let $M_C=I_{d_C}-C(C'C)^+C'$ for a matrix $C$, where the superscript ``$+$'' stands for the Moore-Penrose generalized inverse.  We may also use the Moore-Penrose generalized inverse for the asymmetric matrix $C$: $C^{+} = (C'C)^+C' = C'(CC')^+$. 

The way we prove $\widehat t_n\le \widehat r_n$ is by showing that $\textup{span}(B'_{\widehat L}M_{C_{\widehat L}})\subseteq\textup{span}(A'_{\widehat J})$, where $\textup{span}(\cdot)$ denotes the span of the columns of a matrix in $\R^{d_\mu}$. % as linear subspaces of $\R^{d_\mu}$. 
This is sufficient because $\widehat r_n$ is the dimension of $\textup{span}(A'_{\widehat J})$ and, by Lemma \ref{lem:MB2}(a), $\widehat t_n$ is the dimension of $\textup{span}(B'_{\widehat L}M_{C_{\widehat L}})$. 

Let $(\widehat\mu_n,\widehat\delta_n,\widehat \psi_n)$ satisfy (\ref{BCd_KKT_new})-(\ref{BCd_ComplementarySlackness_new}) and $(\widehat \mu_n, \widehat\lambda_n)$ satisfy (\ref{Ab_KKT_new})-(\ref{Ab_ComplementarySlackness_new}). 
By the definition of $\widehat L$, $\widehat\psi_n=I'_{\widehat L}I_{\widehat L}\widehat\psi_n$. 
This implies that $B'\widehat \psi_n=B'_{\widehat L}I_{\widehat L}\widehat \psi_n$ and, using (\ref{BCd_KKT2_new}), $0=C'_{\widehat L}I_{\widehat L}\widehat \psi_n$, and therefore $I_{\widehat L}\widehat \psi_n=M_{C_{\widehat L}}I_{\widehat L}\widehat \psi_n$. 
By complementary slackness applied to $\widehat \lambda_n$, $\widehat\lambda_n=I'_{\widehat J}I_{\widehat J}\widehat\lambda_n$, which implies that $A'\widehat \lambda_n=A'_{\widehat J}I_{\widehat J}\widehat\lambda_n$. 
Putting these together, (\ref{BCd_KKT_new}) and (\ref{Ab_KKT_new}) imply that $B'_{\widehat L}M_{C_{\widehat L}}I_{\widehat L}\widehat\psi_n=A'_{\widehat J}I_{\widehat J}\widehat\lambda_n\in\textup{span}(A'_{\widehat J})$. 

Let $e_\ell$ be a standard normal basis vector in $\R^{|\widehat L|}$. 
We want to show that $B'_{\widehat L}M_{C_{\widehat L}}e_\ell\in\textup{span}(A'_{\widehat J})$. 
It is sufficient to show that $B'_{\widehat L}M_{C_{\widehat L}}e_\ell+\alpha B'_{\widehat L}M_{C_{\widehat L}}I_{\widehat L}\widehat\psi_n\in\textup{span}(A'_{\widehat J})$ for some $\alpha\in\R$ (because $\textup{span}(A'_{\widehat J})$ is a linear subspace of $\R^{d_\mu}$). 
Let $\alpha>0$ be large enough so that every element of $M_{C_{\widehat L}}e_\ell+\alpha I_{\widehat L}\widehat \psi_n$ is positive (by the definition of $\widehat L$, every element of $I_{\widehat L}\widehat \psi_n$ is positive). 
Let $\ddot\psi=I'_{\widehat L}(M_{C_{\widehat L}}e_\ell+\alpha I_{\widehat L}\widehat \psi_n)$ and $\ddot\mu=\widehat\mu_n+\widetilde\Sigma_nB'\ddot\psi/(2n)$. 
It follows that $(\widehat\mu_n, \widehat\delta_n, \ddot\psi)$ satisfy (\ref{BCd_KKT_new})-(\ref{BCd_ComplementarySlackness_new}) with $\overline{\mu}_n$ replaced by $\ddot \mu$. 
Therefore, by (\ref{AbBCd_equivalence_new}), $(\widehat\mu_n, \ddot\lambda)$ solves (\ref{Ab_KKT_new})-(\ref{Ab_ComplementarySlackness_new}) for some multipliers $\ddot\lambda$ (again, with $\overline{\mu}_n$ replaced by $\ddot\mu$). 
By complementary slackness applied to $\ddot\lambda$, $\ddot \lambda=I'_{\widehat J}I_{\widehat J}\ddot\lambda$, which implies that $A'\ddot \lambda=A'_{\widehat J}I_{\widehat J}\ddot\lambda$. 
Therefore, by (\ref{BCd_KKT_new}) and (\ref{Ab_KKT_new}) (with $\overline{\mu}_n$ replaced by $\ddot\mu$), 
\begin{equation}
B'_{\widehat L}M_{C_{\widehat L}}e_\ell+\alpha B'_{\widehat L}M_{C_{\widehat L}}I_{\widehat L}\widehat\psi_n=B'_{\widehat L}M_{C_{\widehat L}}e_\ell+\alpha B'_{\widehat L}I_{\widehat L}\widehat\psi_n=B'\ddot\psi=A'_{\widehat J}I_{\widehat J}\ddot\lambda\in\textup{span}(A'_{\widehat J}),   
\end{equation}
where the first equality follows because $M_{C_{\widehat L}}I_{\widehat L}\widehat\psi_n=I_{\widehat L}\widehat\psi_n$ (as in the previous paragraph). 
Since $e_\ell$ was arbitrary, $\textup{span}(B'_{\widehat L}M_{C_{\widehat L}})\subseteq\textup{span}(A'_{\widehat J})$, which shows that $\widehat t_n\le\widehat r_n$. 

We next prove that $\widehat r_n\le \widehat s_n$. 
The way we do this is by showing that $\textup{span}(A'_{\widehat J})\subseteq\textup{span}(B'_{\widehat K}M_{C_{\widehat K}})$. 
This is sufficient because $\widehat r_n$ is the dimension of $\textup{span}(A'_{\widehat J})$ and, by Lemma \ref{lem:MB2}(a), $\widehat s_n$ is the dimension of $\textup{span}(B'_{\widehat K}M_{C_{\widehat K}})$. 

Let $(\widehat\mu_n, \widehat\delta_n,\widehat\psi_n)$ satisfy (\ref{BCd_KKT_new})-(\ref{BCd_ComplementarySlackness_new}) and $(\widehat \mu_n, \widehat\lambda_n)$ satisfy (\ref{Ab_KKT_new})-(\ref{Ab_ComplementarySlackness_new}). 
Let $e_j$ be a standard normal basis vector in $\R^{|\widehat J|}$. 
We want to show that $A'_{\widehat J}e_j\in\textup{span}(B'_{\widehat K}M_{C_{\widehat K}})$. 
Recall from the definition of $H$ that every entry must be nonnegative and it must satisfy $H\overline{C}_n=0$. 
Note that $(d-B\widehat\mu_n-\overline{C}_n\widehat\delta_n)'H'I'_{\widehat J}e_j=(g-A\widehat\mu_n)'I_{\widehat J}e_j=0$ because $A=HB$, $g=Hd$, $H\overline{C}_n=0$, and the definition of $\widehat J$. 
This, combined with the fact that $H\ge 0$ and the definition of $\widehat K$, implies that $H'I'_{\widehat J}e_j=I'_{\widehat K}I_{\widehat K}H'I'_{\widehat J}e_j$. 
(For any $k\in\{1,...,d_C\}$, $e'_kH'I'_{\widehat J}e_j\ge 0$ and $e'_k(d-B\widehat \mu_n-\overline{C}_n\widehat\delta_n)\ge 0$, and the inequalities cannot both be strict for the same value of $k$.) 
This is the key step in showing that 
\begin{equation}
A'_{\widehat J}e_j=B'H'I_{\widehat J}e_j=B'_{\widehat K}I_{\widehat K}H'I_{\widehat J}e_j=B'_{\widehat K}M_{C_{\widehat K}}I_{\widehat K}H'I_{\widehat J}e_j, \label{abcd2}
\end{equation}
where the final equality follows from the fact that $0=\overline{C}'_nH'I'_{\widehat J}e_j=C'_{\widehat K}I_{\widehat K}H'I'_{\widehat J}e_j$, which shows that $I_{\widehat K}H'I'_{\widehat J}e_j=M_{C_{\widehat K}}I_{\widehat K}H'I'_{\widehat J}e_j$. 
Overall, (\ref{abcd2}) shows that $A'_{\widehat J}e_j\in\textup{span}(B'_{\widehat K}M_{C_{\widehat K}})$. 
Since $e_j$ was arbitrary, $\textup{span}(A'_{\widehat J})\subseteq\textup{span}(B'_{\widehat K}M_{C_{\widehat K}})$, which shows that $\widehat r_n\le \widehat s_n$. \qed

\subsubsection{Proof of Theorem \ref{lem:rhat}(b)}

Because of part (a), it suffices to prove that $\widehat s_n=\widehat t_n$. 
We prove part (b) under the additional assumption that $\widetilde\Sigma_n=I_{d_\mu}$. 
The theorem then follows for $\widetilde\Sigma_n\neq I_{d_\mu}$ by considering $B^\dagger=B\widetilde\Sigma^{-1/2}_n$, $\overline{\mu}^\dagger_n=\widetilde\Sigma^{1/2}_n\overline{\mu}_n$, and $\widehat{\mu}^\dagger_n=\widetilde\Sigma^{1/2}_n\widehat{\mu}_n$, where $\widetilde\Sigma^{1/2}_n$ denotes the symmetric matrix square root. 
Note that the definitions of $\widehat s_n$, and $\widehat t_n$ are unchanged. 

For any $K\subseteq\{1,...,d_C\}$, let $s_K=\textup{rk}(M_{C_K}B_K)$. 
With this definition, $\widehat s_n=s_{\widehat K}$ and $\widehat t_n=s_{\widehat L}$ by Lemma \ref{lem:MB2}(a). 
Let $\calS_K=\textup{span}(B'_KM_{C_K})$, a linear subspace of $\R^{d_\mu}$. 
Also let $\xi_K=(M_{C_K}B_{K})^{+} M_{C_K} I_K d$. 
For any $K, L\subseteq\{1,...,d_C\}$, let 
\begin{equation}
\M(K,L)=\left(\calS_K\cap\calS_L\right)\oplus\left(\left(\calS_K^\perp+\xi_K\right)\cap\left(\calS_L^\perp+\xi_L\right)\right), 
\end{equation}
where ``$\oplus$'' denotes the Minkowski sum and ``$\perp$'' denotes the orthogonal complement. 
For any $K, L$ with $s_L<s_K$, we have that $\calS_K\cap\calS_L$ is a linear subspace of $\R^{d_\mu}$ with dimension at most $s_L$, and $\left(\calS_K^\perp+\xi_K\right)\cap\left(\calS_L^\perp+\xi_L\right)$ is an affine subspace of $\R^{d_\mu}$ with dimension at most $d_\mu-s_K$. 
Therefore, $\M(K,L)$ is an affine subspace of $\R^{d_\mu}$ with dimension at most $d_\mu-1$ and therefore has Lebesgue measure zero. 
Let 
\begin{equation}
\M_0=\cup_{K, L: s_L<s_K}\M(K,L), 
\label{eq:M0}
\end{equation}
where $K$ and $L$ are arbitrary subsets of $\{1,...,d_C\}$. 
$\M_0$ has Lebesgue measure zero because it is the union of finitely many measure zero sets. 

To finish the proof, we show that if $\widehat t_n<\widehat s_n$, then $\overline{\mu}_n\in\M_0$. 
Fix $\overline{\mu}_n$ and suppose $\widehat K$ and $\widehat s_n$ are defined using some $(\widehat\mu_n, \widehat\delta_n,\widehat\psi_n)$ that satisfy (\ref{BCd_KKT_new})-(\ref{BCd_ComplementarySlackness_new}). 
Also suppose $\widehat L$ and $\widehat t_n$ are defined using some $(\widehat\mu_n, \ddot\delta_n, \ddot\psi_n)$ that satisfy (\ref{BCd_KKT_new})-(\ref{BCd_ComplementarySlackness_new}). 
(Since $\widehat\delta_n$ and $\widehat\psi_n$ are not unique, we want to allow $\widehat s_n$ and $\widehat t_n$ to be defined using different values of the delta and psi.) 
By Lemma \ref{lem:KKTB2_new} applied to $(\widehat\mu_n, \widehat\delta_n,\widehat\psi_n)$, $\overline{\mu}_n-\widehat\mu_n\in\calS_{\widehat K}$ and $\widehat\mu_n\in\calS^\perp_{\widehat K}+\xi_{\widehat K}$. 
By Lemma \ref{lem:KKTB2_new} applied to $(\widehat\mu_n, \ddot\delta_n,\ddot\psi_n)$, $\overline{\mu}_n-\widehat\mu_n\in\calS_{\widehat L}$ and $\widehat\mu_n\in\calS^\perp_{\widehat L}+\xi_{\widehat L}$. 
Therefore, $\overline{\mu}_n-\widehat\mu_n\in\calS_{\widehat K}\cap\calS_{\widehat L}$ and $\widehat\mu_n\in\left(\calS^\perp_{\widehat K}+\xi_{\widehat K}\right)\cap\left(\calS^\perp_{\widehat L}+\xi_{\widehat L}\right)$. 
This shows that $\overline{\mu}_n\in \M(\widehat K, \widehat L)$ with $s_{\widehat L}<s_{\widehat K}$. \qed

\subsection{Proof of Theorem \ref{lem:delta}}

We first show that 
\begin{equation}
\textup{ppoly}(B,d_{n_q};C_{F_{n_q}})\neq\emptyset \text{ and }\textup{ppoly}(B,d_{n_q};\overline{C}_{n_q})\neq\emptyset\textup{ for all }q.\label{BdcnNonEmp}
\end{equation} 
The fact that $F_{n_q}\in {\cal F}_{{n_q}0}$ implies $\mu_{F_{n_q}}\in~\ppoly(B,d_{n_q};C_{F_{n_q}})$. 
Then there exists a $\delta_q^+\in\R^{d_\delta}$ such that $B\mu_{F_{n_q}}+C_{F_{n_q}}\delta_q^+\leq d_{n_q}$. 
It follows that $\mu_q^0:= \mu_{F_{n_q}}-(\overline{\Pi}_{n_q} - \Pi_{F_{n_q}})\delta_q^+\in \textup{ppoly}(B,d_{n_q};\overline{C}_{n_q})$ because $B\mu_q^0+ \overline{C}_{n_q}\delta_q^+=B\mu_{F_{n_q}}+C_{F_{n_q}}\delta_q^+\leq d_{n_q}$. 
This verifies (\ref{BdcnNonEmp}).

Consider an arbitrary subsequence of $\{n_q\}$. 
There exists a further subsequence $\{n_a\}$ such that $(\overline{\mu}_{n_a}',\text{vec}(\overline{\Pi}_{n_a})', \text{vec}(\widehat\Upsilon_{n_a})')'\to (\mu_\infty',\text{vec}(\Pi_\infty)',\text{vec}(\Upsilon_\infty)')'$ almost surely and $\textup{rk}(I_K\overline{C}_{n_a})\to\textup{rk}(I_KC_\infty)$ and $\textup{rk}(I_KC_{F_{n_a}})\to\textup{rk}(I_KC_\infty)$ almost surely for $K=K^\dagger(\mu_\infty;B, C_\infty, d_\infty)$, where $K^\dagger(\mu_\infty;B, C_\infty, d_\infty)$ is the minimal activatable set defined in Section \ref{activatability}. 
Such a further subsequence exists under Assumption \ref{Assumption1}(i, ii, v) and Assumption \ref{assu:rank:relaxation1} (or Assumption \ref{assu:rank:simple}, using Lemma \ref{rank_sufficiency}) because every sequence that converges in probability has a subsequence that converges almost surely. 
We fix one such further subsequence and one realization from the sample space and show that $\widetilde{\delta}_{n_a}\rightarrow \delta^\ast_\infty$ and $\delta^\ast_{F_{n_a}}\rightarrow \delta^\ast_\infty$ deterministically. 
This is sufficient to show that $\widetilde{\delta}_{n_q}\rightarrow_p \delta^\ast_\infty$ and $\delta^\ast_{F_{n_q}}\rightarrow_p \delta^\ast_\infty$ along the original subsequence.  

We next show that, as $a\rightarrow\infty$, 
\begin{equation}
\textup{poly}([B,\overline{C}_{n_a}],d_{n_a}) \overset{K}\to \textup{poly}([B,C_\infty],d_\infty). \label{polyconvergence}
\end{equation}
We verify the conditions of Lemma \ref{K-convergence_nonempty2_new}. 
Condition (i) follows from $\overline{C}_{n_a}= B\overline\Pi_{n_a}+D\to B\Pi_\infty+D= C_\infty$ and $d_{n_a}\rightarrow d_\infty$. 
Condition (ii) follows from (\ref{BdcnNonEmp}). 
Condition (iii) follows  because  $\textup{rk}(I_K[B,\overline{C}_n])  = \textup{rk}([B,B\overline{\Pi}_n+D])=\textup{rk}(I_K[B,D]) = \textup{rk}(I_K[B,B\Pi_\infty+D]) = \textup{rk}(I_K[B,C_\infty])$ for all $K\subseteq \{1,\dots,d_C\}$, where the second and third equalities hold by Lemma \ref{lem:MB2}(b). Then
(\ref{polyconvergence}) follows from Lemma \ref{K-convergence_nonempty2_new}. 

We next show that 
\begin{equation}
\widetilde{\mu}_{n_a}\to \mu_\infty  \label{tildemuconvergence}
\end{equation} 
as $a\rightarrow\infty$. 
Recall $\delta^\ast_\infty = \arg\min_{\delta\in\textup{poly}(C_\infty, d_\infty-B\mu_\infty)}\|\delta\|$, where $\textup{poly}(C_\infty, d_\infty-B\mu_\infty)$ is not empty by Assumption \ref{Assumption1}(vi). 
Note that $(\mu_\infty,\delta^\ast_\infty)\in \textup{poly}([B,C_\infty],d_\infty)$. 
Then, by (\ref{polyconvergence}), there exists a sequence $(\mu^\dagger_{n_a},\delta^\dagger_{n_a})\in \textup{poly}([B,\overline{C}_{n_a}],d_{n_a})$ such that $\mu^\dagger_{n_a}\to \mu_\infty$ and $\delta_{n_a}^\dagger\to \delta^\ast_\infty$. 
Moreover, 
\begin{equation}
\|\widetilde{\mu}_{n_a} - \overline{\mu}_{n_a}\|^2_{\widehat\Upsilon_{n_a}}\leq \|\mu^\dagger_{n_a}-\overline{\mu}_{n_a}\|^2_{\widehat\Upsilon_{n_a}} \label{upsilon_bound}
\end{equation} 
since $\widetilde{\mu}_{n} = \arg\min_{\mu\in \textup{poly}(B,d_n;\overline{C}_n)}\|\overline{\mu}_n-\mu\|^2_{\widehat\Upsilon_n}$. 
The right-hand side in (\ref{upsilon_bound}) converges to 0 because $\mu^\dagger_{n_a}\to\mu_\infty$, with $\overline{\mu}_{n_a}\to\mu_\infty$ and $\widehat\Upsilon_n\rightarrow \Upsilon_\infty$ by Assumption \ref{Assumption1}(i, ii, v). 
This verifies (\ref{tildemuconvergence}). 

We next show that, as $a\rightarrow\infty$, 
\begin{align}
&\textup{poly}(C_{F_{n_a}},d_{n_a}-B\mu_{F_{n_a}})\overset{K}{\to} \textup{poly}({C}_{\infty},d_{\infty}-B{\mu}_\infty)
\textup{ and }\nonumber\\
&\textup{poly}(\overline{C}_{n_a},d_{n_a}-B\widetilde{\mu}_{n_a})\overset{K}{\to} \textup{poly}({C}_{\infty},d_{\infty}-B{\mu}_\infty).  \label{tildepolyconvergence} 
\end{align} 
We verify the conditions of Lemma \ref{K-convergence_nonempty2_new}. 
Condition (i) follows because $\overline{C}_{n_a}\rightarrow C_\infty$, $C_{F_{n_a}}\rightarrow C_\infty$, $d_{n_a}-B\mu_{F_{n_a}}\rightarrow d_\infty-B\mu_\infty$, and, using (\ref{tildemuconvergence}), $d_{n_a}-B\widetilde{\mu}_{n_a}\rightarrow d_{\infty}-B{\mu}_\infty$. 
Condition (ii) follows because $\delta^\ast_{F_{n_a}}\in \textup{poly}(C_{F_{n_a}},d_{n_a}-B\mu_{F_{n_a}})$ and $\widetilde\delta_{n_a}\in\textup{poly}(\overline{C}_{n_a},d_{n_a}-B\widetilde{\mu}_{n_a})$. 
Condition (iii) follows from the fact that $\textup{rk}(I_K\overline{C}_{n_a})\to\textup{rk}(I_KC_\infty)$ and $\textup{rk}(I_KC_{F_{n_a}})\to\textup{rk}(I_KC_\infty)$ for $K=K^\dagger(\mu_\infty; B, C_\infty, d_\infty)$. 
Note that, since the sequence of ranks is discrete, equality must hold eventually. 
Also note that, by Lemma \ref{minimal_activatable}, the only activatable $K\subseteq\{1,...,d_C\}$ such that $\poly(I_KC_\infty, I_K(d_\infty-B\mu_\infty))$ is an affine subspace of $\R^{d_\delta}$ is $K=K^\dagger(\mu_\infty; B, C_\infty, d_\infty)$. 
Therefore, (\ref{tildepolyconvergence}) follows from Lemma \ref{K-convergence_nonempty2_new}. 
Then, $\widetilde\delta_{n_a}\rightarrow\delta^\ast_\infty$ and $\delta^\ast_{F_{n_a}}\rightarrow\delta^\ast_\infty$  follow from Lemma \ref{Projection-convergence} with $x_n=x_\infty=0$ and $\Sigma_n=\Sigma_\infty=I_{d_\delta}$. 
\qed

\subsection{Proof of Theorem \ref{thm:level_CC}}

Let $n_q$ be a subsequence that achieves the limsup. 
Let $F_q\in\mathcal{F}_{n_q0}$ approximately achieve the sup so that $\lim_{q\rightarrow\infty}P_{F_q}\left(T_{n_q}>cv(\widehat{s}_{n_q},\alpha)\right)=\limsup_{n\rightarrow\infty}\sup_{F\in\mathcal{F}_{n0}} P_F\left(T_n>cv(\widehat{s}_n,\alpha)\right)$. 
It is sufficient to find a further subsequence along which this quantity is less than or equal to $\alpha$. 
For simplicity, we denote all further subsequences by $n_q$. 

First, we find a further subsequence along which the following convergence results hold: 
\begin{align}
\left(\begin{smallmatrix}X_{q,\mu}\\\text{vec}(X_{q,C})\end{smallmatrix}\right)&:=\left(\begin{smallmatrix}\sqrt{n_q}\left(\overline{\mu}_{n_q}-\mu_{F_q}\right)\\\sqrt{n_q}\text{vec}\left(\overline{\Pi}_{n_q}-\Pi_{F_q}\right)\end{smallmatrix}\right)\rightarrow_d \left(\begin{smallmatrix}X_\mu\\\text{vec}(X_C)\end{smallmatrix}\right)\sim N(\mathbf{0},\Omega_\infty) \label{distribution_convergence}\\
\widetilde{\Sigma}_{n_q}&= \left(\begin{smallmatrix} I\\ \widetilde{\delta}_{n_q}\otimes I\end{smallmatrix}\right)'\overline{\Omega}_{n_q} \left(\begin{smallmatrix} I\\ \widetilde{\delta}_{n_q}\otimes I\end{smallmatrix}\right) \to_p   \left(\begin{smallmatrix} I\\ \delta^\ast_\infty\otimes I\end{smallmatrix}\right)'\Omega\left(\begin{smallmatrix} I\\ \delta^\ast_\infty\otimes I\end{smallmatrix}\right) =\Sigma_\infty  \label{sigma_convergence}\\
\overline{C}_{n_q}&= B\overline{\Pi}_{n_q}+D\to_p C_\infty= B\Pi_\infty +D \label{C_convergence}\\
h_q&:=\sqrt{n_q}(d_{n_q}-B\mu_{F_q}-C_{F_q}\delta^\ast_{F_q})\to h_\infty \label{h_convergence}\\
\textup{rk}(I_K\overline{C}_{n_q}) &= \textup{rk}(I_KC_\infty), \text{ for every } K\in {\cal K}(B,C_\infty,h_\infty)\cup{\cal L}(B,C_\infty,h_\infty) \textup{ w.p.a.1}, \label{rank_convergence}
\end{align}
where $\mathcal{K}(B, C_\infty, h_\infty)$ and $\mathcal{L}(B, C_\infty, h_\infty)$ are defined in Section \ref{activatability} and w.p.a.1 stands for ``with probability approaching 1''. 
Equation (\ref{distribution_convergence}) follows from Assumption \ref{Assumption1}(ii) for some further subsequence and some $\Omega_\infty$. Equation (\ref{sigma_convergence}) follows from Assumption \ref{Assumption1}(iv) and Theorem \ref{lem:delta}. 
Equation (\ref{C_convergence}) follows from Assumption \ref{Assumption1}(i, ii). 
Equation (\ref{h_convergence}) holds elementwise along a further subsequence for some $h_\infty\in[0,\infty]^{d_C}$ because $F_q\in\mathcal{F}_{n_q,0}$ implies that $h_q\ge 0$. 
Equation (\ref{rank_convergence}) holds by Lemma \ref{rank_sufficiency} under Assumption \ref{assu:rank:simple}. 

We next define the limiting problem. 
Let $X = X_\mu+X_C\delta^\ast_\infty$. 
Consider the following CQPP: 
\begin{equation}
T_\infty = \min_{\eta,\gamma:B\eta+C_\infty\gamma\leq h_\infty}\|X - \eta\|_{\Sigma_\infty^{-1}}^2.\label{limitexperiment}
\end{equation}
Let $\eta^\ast_\infty$ be the unique value of $\eta$ that solves (\ref{limitexperiment}). 
(Recall $h_\infty\ge 0$, so the constraint set is not empty.) 
Let $\gamma^\ast_\infty\in\poly(C_\infty, h_\infty-B\eta^\ast_\infty)$ and achieve the minimal activatable set as defined in Section \ref{activatability}. 
Note that, trivially, 
\begin{equation}
\gamma^\ast _\infty= \underset{\gamma:B\eta^\ast_\infty+C_\infty\gamma\leq h_\infty}{\argmin}\|\gamma-\gamma^\ast_\infty\|.\label{gamma_infty^ast_def}
\end{equation}
Let $K^\ast=\{j\in\{1,...,d_C\}: e'_jB\eta^\ast_\infty+e'_jC_\infty\gamma^\ast_\infty=e'_jh_\infty\}$ and $s^\ast=\textup{rk}\left(I_{K^\ast}[B, D]\right)-\textup{rk}\left(I_{K^\ast}C_\infty\right)$. 
For any value of the multipliers, $\psi^\ast$, solving the KKT conditions for (\ref{limitexperiment}), we can define $L^\ast=\{j\in\{1,...,d_C\}: e'_j\psi^\ast>0\}$ and $t^\ast=\textup{rk}\left(I_{L^\ast}[B, D]\right)-\textup{rk}\left(I_{L^\ast}C_\infty\right)$. 
By Theorem \ref{lem:rhat}(b), there exists a Lebesgue measure zero set, $\mathcal{M}_0$, such that for all $X\notin\mathcal{M}_0$, $r^\ast=s^\ast=t^\ast$, where $r^\ast=\text{dim}\left(B'\{h\ge \mathbf{0}: h'C_\infty=\mathbf{0}, h'(B\eta^\ast_\infty-h_\infty)=0\}\right)$. 

By the almost sure representation theorem, there is a copy of $X_{q,\mu}, X_{q,C}, \widetilde{\Sigma}_{n_q}$ and $\overline{C}_{n_q}$ with identical joint distribution such that the convergence in (\ref{distribution_convergence})-(\ref{C_convergence}) holds almost surely and the equality in (\ref{rank_convergence}) holds eventually, almost surely.\footnote{Note that the eventuality occurs almost surely but not necessarily uniformly. That is, for almost every realization, $\omega$, there exists an $N(\omega)<\infty$ such that (\ref{rank_convergence}) holds for all $q\ge N(\omega)$. It need not be the case that there exists an $N$ such that for all $q\ge N$, (\ref{rank_convergence}) holds almost surely.} 
Abusing notation, we refer to the copies that converge almost surely using the original notation.  
Now fix a sample sequence along which the convergence in (\ref{distribution_convergence})-(\ref{C_convergence}) holds and the equality in (\ref{rank_convergence}) holds eventually. 
We can also take the sample sequence to satisfy $X\notin{\cal M}_0$ and 
\begin{align}
&\|(M_{C_{K}}I_KB\Sigma_\infty^{1/2})^+M_{C_{K}}I_{K}(BX - h_\infty)\|^2_{\Sigma^{-1}_\infty} \neq cv(\textup{rk}(M_{C_{K}}I_KB),\alpha) \nonumber\\
&~~~~~~~~~~~~~~~~~\text{ for all }K\subseteq\{1,\dots,d_C\}\text{ such that }M_{C_{K}}I_KB\neq\mathbf{0} \text{ and } I_K h_\infty<\infty,\label{samplepath}
\end{align}
where $M_{C_{K}} = I_{|K|}-I_{K}C_\infty (C'_\infty I'_{K}I_{K}C_\infty)^{+} C'_\infty I'_{K}$. 
Such sample sequences occur with probability one by the almost sure representation, Theorem \ref{lem:rhat}(b), and Lemma \ref{lem:MixedChi2}. 
Note that Lemma \ref{lem:MixedChi2} applies because $\Sigma_\infty^{-1/2}(M_{C_K}I_KB\Sigma_\infty^{1/2})^+M_{C_K}I_KB\Sigma^{1/2}_\infty $ is not a zero matrix when $M_{C_{K}}I_KB\neq\mathbf{0}$. 

In the rest of the proof, we show that, for the fixed sample sequence, we have
\begin{equation}
\limsup_{q\rightarrow\infty}1\{T_{n_q}>cv(\widehat s_{n_q},\alpha)\}\le 1\{T_\infty>cv(r^\ast,\alpha)\}. \label{goal0}
\end{equation}
If this is true, then (\ref{goal0}) is satisfied with probability 1. We can apply Fatou's lemma 
and conclude that 
\begin{equation}
\limsup_{q\rightarrow\infty}P_{F_q}(T_{n_q}>cv(\widehat s_{n_q},\alpha))\le P(T_\infty>cv(r^\ast,\alpha)). \label{goal2}
\end{equation}
(The probability on the left-hand side does not depend on which copy of the random variable is used because they have identical joint distribution.) 
Also, by Theorem 3(a) of CS23, we have $P(T_\infty>cv(r^\ast,\alpha))\leq \alpha$. 
This is sufficient to prove Theorem \ref{thm:level_CC}. 

Now we show (\ref{goal0}) for the fixed sample sequence. 
Take a further subsequence along which the limsup on the left-hand side is achieved. 
Then, it is sufficient to show that there exists a further subsequence along which (\ref{goal0}) holds. 
We proceed in four steps. 

\textbf{Step 1.} 
We rewrite the test statistic using a change of variables: 
\begin{equation}
T_{n_q}=\min_{\eta,\gamma:B\eta+\overline{C}_{n_q} \gamma\leq h_q}\|X_q - \eta\|_{\widetilde{\Sigma}_q^{-1}}^2, \label{Finitesample}
\end{equation}
where $X_q = X_{q,\mu}+X_{q,C}\delta_{F_q}^\ast$, $\eta= \sqrt{n_q}(\mu-\mu_{F_q})+X_{q,C}\delta_{F_q}^\ast$, and $\gamma = \sqrt{n_q}(\delta-\delta_{F_q}^\ast)$. 
The expression in (\ref{Finitesample}) should clarify that (\ref{limitexperiment}) represents the limit experiment. 
Let $\widehat{\eta}_q =\sqrt{n_q}( \widehat{\mu}_{n_q} - \mu_{F_q}) +X_{q,C}\delta_{F_q}^\ast$ be the unique value of $\eta$ that solves (\ref{Finitesample}). 
Also let 
\begin{equation}
\widehat{\gamma}_q = \underset{\gamma: B\widehat{\eta}_q +\overline{C}_{n_q}\gamma\leq h_q}{\argmin}\|\gamma - \gamma^\ast_\infty\|^2. \label{gamma_hat_def}
\end{equation}
Note that $\widehat\gamma_q$ does not have to be equal to $\sqrt{n_q}(\widehat \delta_{n_q}-\delta^\ast_{F_q})$ if the solution to (\ref{Finitesample}) is not unique. 
Also note that $\widehat\gamma_q$ need not achieve the minimal activatable set, as defined in Section \ref{activatability}. 
By Lemma \ref{delta_perturbation}, we can take $\ddot\gamma_q$ to be within $n_q^{-1}$ of $\widehat\gamma_q$, solve (\ref{Finitesample}), and achieve the minimal activatable set of $\poly(\overline{C}_{n_q}, h_q-B\widehat\eta_q)$. The $\ddot\gamma_q$ defined here is used in Step 3 below. 
Moreover, let $\widehat{\psi}_q$ be the minimum norm Lagrange multiplier that satisfies the KKT conditions: 
\begin{align}
\widehat{\psi}_q& = \underset{\psi\ge 0}{\argmin}\|\psi\|\text{ s.t. }\label{KKT_finite_q1}\\
&2\widetilde{\Sigma}_{n_q}^{-1}(X_q-\widehat{\eta}_q) - B'\psi = 0 \label{KKT_finite_q2}\\
&\overline{C}_{n_q}'\psi=0\label{KKT_finite_q3}\\
&(h_q-B\widehat{\eta}_q -\overline{C}_{n_q}\ddot\gamma_q)'\psi=0.\label{KKT_finite_q4}
\end{align}
Note that (\ref{KKT_finite_q2})-(\ref{KKT_finite_q4}) are the KKT conditions for (\ref{Finitesample}). 

\textbf{Step 2.} 
We next show that there exists a further subsequence such that: 
\begin{equation}
T_{n_q}\to T_\infty, ~\widehat{\eta}_q\to \eta^\ast_\infty, ~\widehat{\gamma}_q\to \gamma^\ast_\infty,~ \ddot\gamma_q\rightarrow \gamma^\ast_\infty, \text{ and }
\widehat{\psi}_q\to\psi^\ast_\infty, \label{95}
\end{equation}
for some $\psi^\ast_\infty$ that, along with $\eta^\ast_\infty$ and $\gamma^\ast_\infty$, satisfy the KKT conditions for (\ref{limitexperiment}). 

First note that 
\begin{equation}
\ppoly(B,h_q;\overline{C}_{n_q})\overset{K}{\to}\ppoly(B,h_\infty;C_\infty)\label{ppolyconv}
\end{equation}
by Lemma \ref{projected_poly_convergence}, which applies because $\overline{C}_{n_q}\to C_\infty$, $h_q\to h_\infty$, $\rk(I_J[B,\overline{C}_{n_q}]) = \rk(I_J[B,B\overline{\Pi}_{n_q}+D]) = \rk(I_J[B,D]) = \rk(I_J[B,B\Pi_\infty +D]) = \rk(I_J[B,C_\infty])$ for any $J\subseteq \{1,\dots,d_C\}$, and for any $\eta\in \R^{d_\mu}$, the minimal activatable set $K^\dagger(\eta)$ of $\poly(C_\infty,h_\infty-B\eta)$ satisfies $\rk(I_{K^\dagger(\eta)} \overline{C}_{n_q}) = \rk(I_{K^\dagger(\eta)}C_\infty)$ eventually. 
That $\rk(I_{K^\dagger(\eta)} \overline{C}_{n_q}) = \rk(I_{K^\dagger(\eta)}C_\infty)$ eventually follows from (\ref{rank_convergence}) because $K^\dagger(\eta)\in{\cal K}(B,C_\infty,h_\infty)$. 
Using (\ref{ppolyconv}) and (\ref{sigma_convergence}), we can apply Lemma \ref{Projection-convergence} to conclude that 
\begin{equation}
\widehat{\eta}_q\to \eta_\infty^\ast. \label{etaconv}
\end{equation}
Next, observe that $T_{n_q} = \|X_q-\widehat{\eta}_q\|^2_{\widetilde{\Sigma}_{n_q}^{-1}}$ and $T_\infty  = \|X - \eta_\infty^\ast\|^2_{\Sigma_\infty^{-1}}$. 
We have that $X_q\to X$ because of (\ref{distribution_convergence}) and the fact that $\delta_{F_q}^\ast \to \delta^\ast_\infty$ from Theorem \ref{lem:delta}. 
Therefore, by (\ref{etaconv}), the invertibility of $\Sigma_\infty$ from Assumption \ref{Assumption1}(iii), and $X_q\to X$, we have $T_{n_q}\to T_\infty$. 

Next, we show that $\widehat\gamma_q\rightarrow\gamma_\infty^\ast$. 
We verify the conditions of Lemma \ref{K-convergence_nonempty2_new} to get 
\begin{equation}
\poly(\overline{C}_{n_q}, h_q-B\widehat\eta_q)\overset{K}{\to} \poly(C_\infty, h_\infty-B\eta^\ast_\infty). \label{K-convergence_poly}
\end{equation}
Condition (i) is satisfied because $\overline{C}_{n_q}\rightarrow C_\infty$ and $h_q-B\widehat\eta_q\rightarrow h_\infty-B\eta^\ast_\infty$, using the previous paragraph. 
Condition (ii) is satisfied because $\widehat\gamma_q\in\poly(\overline{C}_{n_q}, h_q-B\widehat\eta_q)$. 
Condition (iii) is satisfied because any $K\subseteq\{1,...,d_C\}$ that is activatable for $\poly(C_\infty, h_\infty-B\eta^\ast_\infty)$ for which $\poly(I_K C_\infty, I_K h_\infty-I_KB\eta_\infty^\ast)$ defines an affine subspace of $\R^{d_\delta}$ also belongs to $\mathcal{K}(B, C_\infty, h_\infty)$ (with $x=\eta^\ast_\infty$) by Lemma \ref{minimal_activatable}(c). 
Condition (iii) then follows from (\ref{rank_convergence}). 
Therefore, by Lemma \ref{K-convergence_nonempty2_new}, (\ref{K-convergence_poly}) holds. 
We then verify the conditions of Lemma \ref{Projection-convergence} to get $\widehat\gamma_q\rightarrow\gamma_\infty^\ast$. 
Note that $\widehat\gamma_q$ is the projection of $\gamma^\ast_\infty$ onto $\poly(\overline{C}_{n_q}, h_q-B\widehat\eta_q)$ and $\gamma_\infty^\ast$ is trivially the projection of $\gamma^\ast_\infty$ onto $\poly(C_\infty, h_\infty-B\eta^\ast_\infty)$ by (\ref{gamma_hat_def}) and (\ref{gamma_infty^ast_def}), respectively. 
The K-convergence condition is satisfied by (\ref{K-convergence_poly}) and $\poly(C_\infty, h_\infty-B\eta^\ast_\infty)$ is nonempty because it includes $\gamma_\infty^\ast$. 
Therefore, $\widehat\gamma_q\rightarrow\gamma_\infty^\ast$. 
Since $\ddot\gamma_q$ is within $n_q^{-1}$ distance from $\widehat{\gamma}_q$, we also have $\ddot\gamma_q\rightarrow\gamma_\infty^\ast$. 

To show convergence of $\widehat{\psi}_q$, we use Lemma \ref{multiplier_convergence}. 
Since $\widehat\psi_q$ is defined to satisfy (\ref{KKT_finite_q1})-(\ref{KKT_finite_q4}), it is the minimum norm multiplier that satisfies (\ref{multiplier_lemma_KKT1})-(\ref{multiplier_lemma_KKT4}). 
Note that, by complementary slackness, $I_{\ddot{K}_q^c}\widehat\psi_q=0$, where $\ddot{K}_q=\{j: e'_j(h_q-B\widehat\eta_q-\overline{C}_{n_q}\ddot\gamma_q)=0\}$. 
Also note that $(\widehat\eta_q, \ddot\gamma_q)\rightarrow(\eta^\ast_\infty, \gamma^\ast_\infty)$ by the previous two paragraphs. 
Consider a further subsequence along which $\ddot{K}_q$ does not depend on $q$. 
Along this further subsequence, $\widehat\psi_q\rightarrow\psi_\infty^\ast$, where $\psi^\ast_\infty\ge \mathbf{0}$ satisfies: 
$\Sigma_\infty^{-1}{(X-\eta^\ast_\infty)}=B'\psi^\ast_\infty$, $C'_\infty\psi^\ast_\infty=\mathbf{0}$, and $I_{(K^\ast)^c}\psi^\ast_\infty=\mathbf{0}$, where $K^\ast=\{j: e'_j(h_\infty-B\eta^\ast_\infty-C_\infty\gamma^\ast_\infty)=0\}$. 
The last condition implies that $(h_\infty-B\eta^\ast_\infty-C_\infty\gamma^\ast_\infty)'\psi^\ast_\infty=0$, which shows that $(\eta^\ast_\infty,\gamma^\ast_\infty,\psi^\ast_\infty)$ satisfies all the KKT conditions for (\ref{limitexperiment}). 

\textbf{Step 3.} 
We next bound the limit of $\widehat s_{n_q}$. 
Recall $\widehat{K}_q = \{j: e_j'(B\widehat{\mu}_{n_q}+\overline{C}_{n_q}\widehat{\delta}_{n_q}-d_{n_q})=0\}$ and $\widehat{L}_q =\{j: e_j'\widehat{\psi}_q>0\}$. 
Let $L^\ast =\{j: e_j'\psi^\ast_\infty>0\}$ and $t^\ast=\textup{rk}(I_{L^\ast}[B, C_\infty])-\textup{rk}(I_{L^\ast}C_\infty)$. 
Note that 
\begin{equation}
L^\ast\subseteq \widehat{L}_q\subseteq \ddot{K}_q\subseteq \widehat K_q\label{rankordering}
\end{equation}
eventually, where the first set inclusion follows from the convergence of $\widehat{\psi}_q$ to $\psi^\ast_\infty$, the middle set inclusion follows from complementary slackness, and the final set inclusion follows from the fact that $\ddot\gamma_q$ was chosen to achieve the minimal activatable set. 
We next show that 
\begin{equation}
s^\ast=r^\ast=t^\ast=\textup{rk}(I_{L^\ast}[B, \overline{C}_q])-\textup{rk}(I_{L^\ast}\overline{C}_q)\le\widehat t_{n_q}\le\widehat s_{n_q}\label{r1conv}
\end{equation}
eventually as $q\rightarrow\infty$. 
The first two equalities follow because $X\not\in{\cal M}_0$, the third equality follows from (\ref{rank_convergence}) and $L^\ast\in\mathcal{L}(B, C_\infty, h_\infty)$ so that $\textup{rk}(I_{L^\ast}\overline{C}_q)=\textup{rk}(I_{L^\ast}C_\infty)$ eventually, and the two inequalities follow from applying Lemma \ref{lem:MB2}(c) to (\ref{rankordering}).\footnote{Note that $L^\ast\in\mathcal{L}(B, C_\infty, h_\infty)$ because $(\eta^\ast_\infty, \gamma^\ast_\infty,\psi^\ast_\infty)$ satisfy the KKT conditions for (\ref{limitexperiment}) and $\gamma^\ast_\infty$ is defined to achieve the minimal activatable set. 
The third equality also uses $\textup{rk}(I_{L^\ast}[B, \overline{C}_q])=\textup{rk}(I_{L^\ast}[B,D])=\textup{rk}(I_{L^\ast}[B, C_\infty])$.}

\textbf{Step 4.} 
To finish the proof, we verify (\ref{goal0}) in three cases. 
(A) Suppose $\widehat s_{n_q}=0$ along a subsequence. 
Fix that subsequence. 
By Lemma \ref{s=0_implies_T=0}, $T_{n_q}=0$. 
Therefore, $1\{T_{n_q}>cv(\widehat s_{n_q}, \alpha)\}=0$ along the subsequence, which satisfies (\ref{goal0}). 

(B) Suppose $\widehat s_{n_q}\ge 1$ eventually and $r^\ast=0$. 
By Lemma \ref{s=0_implies_T=0}, $T_\infty=0$. 
Therefore, $T_{n_q}\rightarrow 0$ must satisfy $1\{T_{n_q}>cv(\widehat s_{n_q}, \alpha)\}=0$ eventually. 
This satisfies (\ref{goal0}). 

(C) Suppose $r^\ast>0$. 
This implies that $s^\ast\neq 0$ by Theorem \ref{lem:rhat}(a). 
Therefore, $M_{C_{K^\ast}}I_{K^\ast}B\neq \mathbf{0}$ because $s^\ast=\textup{rk}\left(M_{C_{K^\ast}}I_{K^\ast}B\right)$
by Lemma \ref{lem:MB2}(a,b). 
By (\ref{samplepath}), 
\begin{equation}
\|(M_{C_{K^\ast}}I_{K^\ast}B\Sigma_\infty^{1/2})^+M_{C_{K^\ast}}I_{K^\ast}(BX - h_\infty)\|^2_{\Sigma^{-1}_\infty} \neq cv(s^\ast,\alpha), 
\end{equation}
using the fact that $I_{K^\ast}h_\infty<\infty$ by the definition of $K^\ast$. 
Also note that, using Lemma \ref{lem:KKTB2_new}, 
\begin{equation}
T_\infty = \|X-\eta^\ast_\infty\|_{\Sigma^{-1}_\infty}^2 = \|(M_{C_{K^\ast}}I_{K^\ast}B\Sigma_\infty^{1/2})^+M_{C_{K^\ast}}I_{K^\ast}(BX-h_\infty)\|^2_{\Sigma_\infty^{-1}}. 
\end{equation}
Therefore, $T_{n_q}\rightarrow T_\infty\neq cv(s^\ast,\alpha)$. 
This implies that $1\{T_{n_q}>cv(\widehat s_{n_q},\alpha)\}\le 1\{T_{n_q}>cv(s^\ast,\alpha)\} = 1\{T_\infty>cv(r^\ast,\alpha)\}$ eventually, where the inequality uses (\ref{r1conv}). 
This verifies (\ref{goal0}) in this case. 
\qed

\bibliographystyle{apalike}
\bibliography{references}

\newpage
\appendix

\begin{center}
{\LARGE{Supplemental Appendix for\\\bigskip ``Testing Inequalities Linear in Nuisance Parameters''}}

\bigskip

{\large{Gregory Fletcher Cox~~~~~~Xiaoxia Shi~~~~~~Yuya Shimizu}}

\bigskip

{\large{\today}}
\end{center}
This supplemental appendix is organized as follows:
\begin{itemize}
\item Section \ref{Section_Generalizations} extends and generalizes Theorems \ref{lem:delta} and \ref{thm:level_CC}, including a discussion of Assumption \ref{assu:rank:simple}. 

\item Section \ref{UsefulLemmas} states and proves lemmas used in the proofs. 

\item Section \ref{A1iiiDiscussion} discusses Assumption \ref{Assumption1}(vi). 

\item Section \ref{app:MTE} presents simulations for the policy relevant treatment effect example. 
\end{itemize}

\setcounter{section}{1}
\section{Extensions/Generalizations of Theorems \ref{lem:delta}-\ref{thm:level_CC}}\label{Section_Generalizations}

This section extends and generalizes the GCC test and Theorems \ref{lem:delta}-\ref{thm:level_CC}. 
Section \ref{activatability} gives precise definitions for activatability of a set of indices for a system of linear inequalities. 
Section \ref{relaxations} states theorems that relax Assumption \ref{assu:rank:simple}. 
Section \ref{DiscussAssu2} discusses Assumption \ref{assu:rank:simple}. 
Section \ref{app:refinement} defines the Refined GCC (RGCC) test and proves its validity. 
Section \ref{LMGCC} defines the GCC test with KKT multipliers and proves its validity. 

\subsection{Definitions of Activatability}\label{activatability}
Assumption \ref{assu:rank:simple} requires the rank stability equation to hold for all $K\in\mathcal{A}(C_\infty, b_\infty)$ with $K^=\subseteq K$. 
In this section, we define two subcollections of $\{K\in {\cal A}(C_\infty, b_\infty): K^{=}\subseteq K\}$ and state weaker conditions where the rank stability equation is only required to hold on these subcollections.  
We are especially careful to cover potential non-uniqueness of the solution to the CQPP as well as potential non-uniqueness of the KKT multipliers. 

The first concept of activatability is defined with respect to the polyhedral set $\poly(C,b)$. 
We say that a collection of inequalities, $J\subseteq\{1,...,d_C\}$ is \underline{activatable for $\poly(C,b)$} if there exists a $\delta\in \poly(C,b)$ such that $J=\{j\in\{1,...,d_C\}: e'_j C\delta=e'_j b\}$. 
Let $\mathcal{J}$ denote the set of all $J\subseteq \{1,...,d_C\}$ that are activatable for $\poly(C,b)$. 
Note that ${\cal J}$ depends on the initial $C$ and $b$, which we keep implicit. 
We also define 
\begin{equation}
J^\dagger=\cap_{J\in{\cal J}}J.
\end{equation}
We call $J^\dagger$ the minimal activatable set of $\poly(C,b)$. 
Lemma \ref{minimal_activatable}, below, explains the sense in which it is ``minimal'' and ``activatable''. 
The lemma is proved at the end of this subsection. 
\begin{lemma}\label{minimal_activatable}
\begin{enumerate}
\item[\textup{(a)}] $J^\dagger\in {\cal J}$. 
\item[\textup{(b)}] The inequalities specified by $J^\dagger$ define an affine subspace of $\R^{d_\delta}$. That is, we have $\textup{poly}(I_{J^\dagger}C, I_{J^\dagger}b) = \{\delta\in \R^{d_\delta}: I_{J^\dagger}C\delta=I_{J^\dagger}b\}$.
\item[\textup{(c)}] For any $\ddot{J}\in {\cal J}$ such that $\textup{poly}(I_{\ddot{J}}C,I_{\ddot{J}}b) = \{\delta\in \R^{d_\delta}:I_{\ddot{J}}C\delta = I_{\ddot{J}}b\}$, we have $\ddot{J} = J^{\dagger}$.  
\end{enumerate}
\end{lemma}
Intuitively, the minimal activatable set corresponds to a $\delta\in\poly(C,b)$ in the interior or relative interior for which any inequality that can be slack is slack. Any inequality that cannot be slack is implicitly an equality that defines an affine subspace of $\R^{d_\delta}$. 

The next concept of activatability is with respect to the CQPP: 
\begin{equation}
\min_{(\mu,\delta)\in \poly([B,C],d)} (x-\mu)'\Sigma^{-1}(x-\mu), \label{QP}
\end{equation}
for some positive definite $\Sigma$. 
For each $x\in \R^{d_\mu}$, let $\mu(x)$ denote the unique minimizing value of $\mu$. 
The minimizing value of $\delta$ is not unique. Instead, the set of possible minimizing values for $\delta$ are $\Delta(x):=\poly(C,d-B\mu(x))$. For each $K\subseteq\{1,\dots,d_C\}$, we call $K$ a \underline{K-activatable set} for the CQPP (\ref{QP}) if $K$ is activatable for $\Delta(x)$ in the sense defined above. By definition, any K-activatable set for (\ref{QP}) is of the form  $K(x,\delta)=\{k\in\{1,...,d_C\}: e'_kB\mu(x)+e'_kC\delta=e'_kd\}$ for some $\delta \in \Delta(x)$.  
Let the  collection of K-activatable sets for (\ref{QP}) be denoted $\mathcal{K}(x)=\{K(x,\delta): \delta\in\Delta(x)\}$. 

By Lemma \ref{minimal_activatable}, for each $x\in\R^{d_\mu}$, there exists a unique minimal K-activatable set for (\ref{QP}), and we denote it by $K^\dagger(x)$. 
Furthermore,  let the collection of minimal K-activatable sets be  $\mathcal{K}=\{K^\dagger(x): x\in\R^{d_\mu}\}$.  
Note that $K^\dagger(x)$ and $\mathcal{K}$ depend on the original specification of $B$, $C$, $d$, and $\Sigma$. 
We can make this dependence explicit with $K^\dagger(x; B, C, d, \Sigma)$ and $\mathcal{K}(B, C, d, \Sigma)$. 

The following lemma shows that for any $\delta\in \Delta(x)$, there is a point in $\Delta(x)$ arbitrarily close to it that achieves the minimal K-activatable set. 
It is used in the proof of Theorem \ref{thm:level_CC}, and its proof is at the end of this subsection. 
\begin{lemma}\label{delta_perturbation}
For any $\delta\in\Delta(x)$ and for any $\epsilon>0$, there exists a $\ddot\delta\in\Delta(x)$ such that $\|\delta-\ddot\delta\|<\epsilon$ and $K(x,\ddot\delta)=K^\dagger(x)$. 
\end{lemma}

There is another way to define activatability with respect to (\ref{QP}), based on the KKT multipliers. 
Write out the KKT conditions associated with (\ref{QP}): 
\begin{align}
0&=2\Sigma^{-1}(\mu(x)-x)+B'\psi \label{QP_KKT1}\\
0&=C'\psi\\
0&\le \psi\label{QP_KKT3}\\
0&=\psi'(B\mu(x)+C\delta-d) \label{QP_KKT4}\\
0&\ge B\mu(x)+C\delta-d,\label{QP_KKT5}
\end{align}
where $\psi$ denotes the KKT multipliers. 
For any $x\in\R^{d_\mu}$ and for any $\delta\in\Delta(x)$, (\ref{QP_KKT1})-(\ref{QP_KKT4}) specify a system of linear equalities/inequalities (a polyhedral set) for values of $\psi$ that solve the KKT problem. 
Denote that set by $\Psi(x,\delta)$. 
Notice that this set does not depend on $\delta$ except through $K(x,\delta)$. 
That is, for any $K$, let $\Psi(x,K)$ (abusing notation) denote the set of values of $\psi$ that satisfy (\ref{QP_KKT1})-(\ref{QP_KKT3}), together with 
\begin{equation}
0=I_{K^c}\psi, \label{QP_KKT6}
\end{equation}
where $K^c=\{1,...,d_C\}/K$. 
We then have $\Psi(x,\delta)=\Psi(x,K(x,\delta))$. 

We only concern ourselves with KKT multipliers associated with the minimal K-activatable set for (\ref{QP}) in the proofs. Let 
\begin{equation}
\Psi^\dagger(x)=\Psi(x,K^\dagger(x)). \label{mini_L}
\end{equation}
For any $\psi\in\Psi^\dagger(x)$, let $L(x,\psi)=\{\ell\in\{1,...,d_C\}: e'_\ell \psi>0\}$ be an \underline{L-activatable set} for (\ref{QP}). 
Let the collection of L-activatable sets be given by $\mathcal{L}(x)=\{L(x,\psi):\psi\in\Psi^\dagger(x)\}$ and $\mathcal{L}=\cup_{x\in\R^{d_\mu}}\mathcal{L}(x)$.\footnote{For the purpose of proving Theorem \ref{thm:level_CC}, we could further restrict $\mathcal{L}(x)$ to be $\{L(x,\psi):\psi\in\Psi^\dagger(x) \text{ and } \rk\left(I_{L(x,\psi)}[B,C]\right)-\rk\left(I_{L(x,\psi)}C\right)=\rk\left(I_{K^\dagger(x)}[B,C]\right)-\rk\left(I_{K^\dagger(x)}C\right)\}$. Then, all the L-activatable sets give the same DoF as $K^\dagger(x)$.} 

Note that $\mathcal{L}$ depends on the specification of $B$, $C$, $d$, and $\Sigma$.  We can write it as $\mathcal{L}(B, C, d, \Sigma)$. The following lemma shows that $\Sigma$ is redundant in the notation. The same is true for the collection of K-activatable sets: ${\cal K}(B,C,d,\Sigma)$, and for some $x$, for the minimal K-activatable set $K^\dagger(x;B,C,d,\Sigma)$. The lemma is proved at the end of this subsection. 
\begin{lemma}\label{L-invariance} The sets $\mathcal{K}(B, C, d, \Sigma)$ and $\mathcal{L}(B, C, d, \Sigma)$ do not depend on $\Sigma$. 
Furthermore, if $x\in \ppoly(B, d; C)$, then $K^\dagger(x; B, C, d, \Sigma)$ does not depend on $\Sigma$. 
\end{lemma}

In light of Lemma \ref{L-invariance}, we can denote $\mathcal{K}(B, C, d, \Sigma)$ by $\mathcal{K}(B, C, d)$, $\mathcal{L}(B, C, d, \Sigma)$ by $\mathcal{L}(B, C, d)$, and $K^\dagger(x; B, C, d, \Sigma)$ by $K^\dagger(x; B, C, d)$ for $x\in \ppoly(B,d;C)$. 
\subsubsection{Proofs of Lemmas \ref{minimal_activatable}-\ref{L-invariance}}

\begin{proof}[Proof of Lemma \textup{\ref{minimal_activatable}}]
(a) For each $J\in\mathcal{J}$, let $\delta_J\in \textup{poly}(C,b)$ be such that $J=\{j\in\{1,...,d_C\}: e'_jC\delta_J=e'_jb\}$. 
Take $\delta^\dagger=\frac{1}{|\mathcal{J}|}\sum_{J\in\mathcal{J}}\delta_J$. 
This satisfies $J^\dagger=\{j\in\{1,...,d_C\}: e'_jC\delta^\dagger = e'_j b\}$. Thus $J^\dagger \in\mathcal{J}$.

For part (b), suppose $\poly(I_{J^\dagger}C, I_{J^\dagger}b)$ is not an affine subspace of $\R^{d_\delta}$. Then there exists a $\widetilde{\delta}\in \poly(I_{J^\dagger}C, I_{J^\dagger}b)$ and a nonempty $J^+\subseteq J^\dagger$ such that every element of $I_{J^+}C\widetilde\delta$ is strictly less than the corresponding one in $I_{J^{+}}b$. Define $\delta^\dagger$ as in the proof of part (a). Then $\delta(\eta):=\eta \widetilde{\delta}+(1-\eta)\delta^{\dagger}\in \textup{poly}(C,b)$ for  small enough $\eta>0$. Also, $\{j\in\{1,\dots,d_C\}: e_j'C\delta(\eta)=e_j'b \}$ does not include elements in $J^+$ and thus is a strict subset of $J^\dagger$. This contradicts the definition of $J^\dagger$. Thus, $\poly(I_{J^\dagger}C, I_{J^\dagger}b)$ is an affine subspace of $\R^{d_\delta}$.

For part (c), $\textup{poly}(I_{\ddot J}C,I_{\ddot J}b) = \{\delta\in\R^{d_\delta}:I_{\ddot J} C\delta = I_{\ddot J}b\}$ implies that $\ddot{J}\subseteq J$ for any $J\in {\cal J}$. Thus, $J^\dagger = \ddot J$.
\end{proof}
\begin{proof}[Proof of Lemma \textup{\ref{delta_perturbation}}]
Let $\delta^\dagger$ be a point in $\Delta(x)$ such that $I_{K^\dagger(x)}C\delta^\dagger = I_{K^\dagger(x)}(d-B\mu(x))$ and for any $j\notin K^\dagger(x)$, $e_j'C\delta^\dagger<e_j'(d-B\mu(x))$. Then $\ddot\delta:=(1-\eta)\delta+\eta\delta^\dagger \in \Delta(x)$ and for any $\eta\in(0,1)$, $K(x,\ddot\delta)\subseteq K^\dagger(x)$. Since $K^\dagger(x)$ is the minimal K-activatable set of $\Delta(x)$, we have that $K(x,\ddot\delta)=K^\dagger(x)$. Choosing a $\eta$ small enough so that $\|\delta-\ddot\delta\|<\epsilon$ proves the lemma. 
\end{proof}
\begin{proof}[Proof of Lemma \textup{\ref{L-invariance}}]
If, for some $x$ and $\psi$, $L(x,\psi)\in\mathcal{L}(B, C, d, \Sigma)$, then $L(x,\psi)=L(y,\psi)\in\mathcal{L}(B, C, d, \widetilde\Sigma)$ for any positive definite $\widetilde\Sigma$ by taking $y=\mu(x)-\widetilde\Sigma\Sigma^{-1}(\mu(x)-x)$. 
This follows because the same $\psi$ (as well as the same $\delta$ and $\mu = \mu(x)$) satisfies the KKT conditions in (\ref{QP_KKT1}) - (\ref{QP_KKT5}) with $y$ in place of $x$ and $\widetilde\Sigma$ in place of $\Sigma$. 
Similarly, if for some $x$, $K^\dagger(x; B, C, d, \Sigma)\in\mathcal{K}(B, C, d, \Sigma)$, then $K^\dagger(y; B, C, d, \widetilde\Sigma)\in \mathcal{K}(B, C, d, \widetilde\Sigma)$. 
If $x\in\ppoly(B, d; C)$, then $\mu(x)=x$ and $\Delta(x)$ does not depend on $\Sigma$. 
\end{proof}

\subsection{Relaxing Assumption \ref{assu:rank:simple}}\label{relaxations}

Theorems \ref{lem:delta}-\ref{thm:level_CC} use Assumption \ref{assu:rank:simple} to ensure Kuratowski convergence of the constraint set that defines $T_n$. 
Also, Theorem \ref{thm:level_CC} uses Assumption \ref{assu:rank:simple} to bound the limit of $\widehat s_n$. 
We can relax Assumption \ref{assu:rank:simple} by clarifying precisely for which sets of indices, $K$, the rank condition in Assumption \ref{assu:rank:simple} needs to hold. 

Assumptions \ref{assu:rank:relaxation1} and \ref{assu:rank:relaxation3} state relaxed versions of Assumption \ref{assu:rank:simple} for certain collections of activatable sets. 
Lemma \ref{rank_sufficiency}, below, shows that Assumption \ref{assu:rank:simple} is sufficient for Assumptions \ref{assu:rank:relaxation1} and \ref{assu:rank:relaxation3}. 
These assumptions use the definitions of activatable sets from Section \ref{activatability}. 
\begin{assumption}\label{assu:rank:relaxation1}
For every sequence $\{F_n\}_{n=1}^{\infty}$ with $F_n\in\mathcal{F}_{n0}$ and for every subsequence, $\{n_m\}$, satisfying Assumption \textup{\ref{Assumption1}(i)} with $C_\infty=B\Pi_\infty+D$ and $b_\infty=d_\infty-B\mu_\infty$, there exists a further subsequence, $\{n_q\}$, such that 
\[
P_{F_{n_q}}(\textup{rk}(I_K\overline{C}_{n_q}) =\rk(I_KC_{F_{n_q}})=\textup{rk}(I_KC_\infty))\to 1 
\]
as $q\rightarrow\infty$ for $K=K^\dagger(\mu_\infty; B, C_\infty, d_\infty)$. 
\end{assumption}

\begin{assumption}\label{assu:rank:relaxation3}
For every sequence $\{F_n\}_{n=1}^{\infty}$ with $F_n\in\mathcal{F}_{n0}$ and for every subsequence, $\{n_m\}$, satisfying Assumption \textup{\ref{Assumption1}(i)} with $C_\infty=B\Pi_\infty+D$ and also satisfying $h_m=\sqrt{n_m}(d_{n_m}-B\mu_{F_{n_m}}-C_{F_{n_m}}\delta^\ast_{F_{n_m}})\rightarrow h_\infty$ elementwise for some vector $h_\infty\in[0,\infty]^{d_C}$, there exists a further subsequence, $\{n_q\}$, such that 
\[
P_{F_{n_q}}(\textup{rk}(I_K\overline{C}_{n_q}) =\textup{rk}(I_KC_\infty))\to 1 
\]
as $q\rightarrow\infty$ for all $K\in\mathcal{K}(B, C_\infty, h_\infty)\cup\mathcal{L}(B, C_\infty,h_\infty)$. 
\end{assumption}

\noindent\textbf{Remarks:} (1) \textit{Assumptions \ref{assu:rank:relaxation1} and \ref{assu:rank:relaxation3} only require the stable rank condition to hold for sets of inequalities that are activatable for the limit polyhedral set. In Assumption \ref{assu:rank:relaxation1}, the limit polyhedral set is defined without scaling the slackness, while for Assumption \ref{assu:rank:relaxation3}, the limit polyhedral set is defined by scaling the slackness at the $\sqrt{n}$ rate. (The difference between K-activatable and L-activatable sets of inequalities is practically irrelevant.) Furthermore, Assumptions \ref{assu:rank:relaxation1} and \ref{assu:rank:relaxation3} only require the stable rank condition to hold for minimal activatable sets. By Lemma \ref{minimal_activatable}, minimal activatable sets are the inequalities associated with equalities, either explicitly or implicitly. Any inequality that can possibly be slack is excluded. Of course, several different limits can arise depending on the sequence $\{F_n\}_{n=1}^{\infty}$ and uniformity requires the stable rank condition to hold for minimal activatable sets of inequalities in all of these limits. See Section \ref{DiscussAssu2}, below, for more discussion in some simple examples.} 

(2) \textit{\citeapp{Voronin2025} considers the problem of uniformly consistent estimation of the value of a LPP. They point out that no uniformly consistent estimator exists if the optimal vertex is allowed to become arbitrarily ``sharp'' as the sample size increases. {This is related to the rank stability condition because the problematic set of inequalities defines a linear subspace in the limit, and the rank stablility equation does not hold for the $K$ corresponding to this set.}}\medskip 

The following lemma shows that these assumptions are weaker than Assumption \ref{assu:rank:simple}. 
\begin{lemma}\label{rank_sufficiency}
Assumption \textup{\ref{assu:rank:simple}} implies Assumptions \textup{\ref{assu:rank:relaxation1}} and \textup{\ref{assu:rank:relaxation3}}. 
\end{lemma}
\begin{proof}[Proof of Lemma \textup{\ref{rank_sufficiency}}]
Fix a sequence $\{F_n\}_{n=1}^\infty$ with $F_n\in\mathcal{F}_{n0}$ and fix a subsequence $n_m$ satisfying Assumption \ref{Assumption1}(i) with $C_\infty=B\Pi_\infty+D$ and $b_\infty=d_\infty-B\mu_\infty$. 
If $\poly(C_\infty, b_\infty)$ is empty, then Assumptions \ref{assu:rank:relaxation1} and \ref{assu:rank:relaxation3} hold trivially. 
Therefore, suppose $\poly(C_\infty, b_\infty)$ is not empty. 
For verifying Assumption \ref{assu:rank:relaxation3}, suppose further that $h_m\rightarrow h_\infty$ (elementwise). 
By Assumption \ref{assu:rank:simple}, there exists a further subsequence, $n_q$, such that 
\begin{equation}
P_{F_{n_q}}(\textup{rk}(I_K\overline{C}_{n_q}) =\textup{rk}(I_KC_{F_{n_q}})=\textup{rk}(I_KC_\infty))\to 1 \label{lemma_5_goal} 
\end{equation}
for all $K\in \mathcal{A}(C_\infty, b_\infty)$ with $K^=\subseteq K$. 

To verify Assumption \ref{assu:rank:relaxation1}, note that for $K = K^\dagger(\mu_\infty;B,C_\infty,d_\infty)$, there exists a $\delta$ such that $K = \{j\in\{1,\dots,d_C\}:e_j'(B\mu_\infty+C_\infty\delta -d_\infty) = 0\}$. 
This implies that $K\in{\cal A}(C_\infty, b_\infty)$. 
Also note that $K^=\subseteq K$ because any activatable set contains $K^=$ and $K$ is an activatable set. 
Therefore, (\ref{lemma_5_goal}) holds. 
This verifies Assumption \ref{assu:rank:relaxation1}. 

To verify Assumption \ref{assu:rank:relaxation3}, let $K\in {\cal K}(B,C_\infty,h_\infty)\cup {\cal L}(B,C_\infty,h_\infty)$. 
By the definition of ${\cal K}(B,C_\infty,h_\infty)$ and ${\cal L}(B,C_\infty,h_\infty)$, there exists an $x\in \R^{d_\mu}$ and a $\delta\in \textup{poly}(C,h_\infty-B\mu(x))$ such that $e_j'(B\mu(x) + C_\infty\delta -h_\infty)=0$ for all $j\in K$. 
Here the function $\mu(x)$ is the solution to the CQPP $\min_{\mu,\delta: B\mu+C_\infty\delta \leq h_\infty}(x-\mu)'\Sigma^{-1}(x-\mu)$ for some positive definite matrix $\Sigma$. 
The solution $\mu(x)$ is finite because $(\mu',\delta')' = \mathbf{0}$ satisfies $B\mu+C_\infty\delta \leq h_\infty$ and hence $(x-\mu(x))'\Sigma^{-1}(x-\mu(x))\leq x'\Sigma^{-1}x<\infty$. 
Therefore, $\textup{ for any }j\in K\textup{, we have }e_j'h_\infty = e_j'(B\mu(x)+C_\infty\delta)<\infty$. 
This and the definition of $h_\infty$ together imply that 
\begin{equation}
e_j'(d_{n_q}-B\mu_{F_{n_q}}-C_{F_{n_q}}\delta^\ast_{F_{n_q}})\to 0 ~ \textup{ for any }j\in K. \label{finiteh}
\end{equation} 
Next, we invoke Lemma \ref{K-convergence_nonempty2_new} to get $\poly(C_{F_{n_q}}, d_{n_q}-B\mu_{F_{n_q}})\rightarrow \poly(C_\infty, b_\infty)$, where the conditions of the lemma are guaranteed by Assumptions \ref{Assumption1}(i), \ref{assu:rank:simple} and the fact that $\poly(C_\infty, b_\infty)$ is nonempty. 
It then follows from Lemma \ref{Projection-convergence} that $\delta^\ast_{F_{n_q}}\rightarrow \delta^\ast_\infty$, which, together with Assumption \ref{Assumption1}(i), implies that $d_{n_q}-B\mu_{F_{n_q}}-C_{F_{n_q}}\delta^\ast_{F_{n_q}}\rightarrow b_\infty - C_\infty \delta_\infty^\ast$. 
This combined with (\ref{finiteh}) implies that for $x=\delta^\ast_\infty$, $I_K(C_\infty x-b_\infty)=0$, and hence $K\in\mathcal{A}(C_\infty, b_\infty)$. 
Also note that $K^=\subseteq K$ because any activatable set contains $K^=$. 
Therefore, (\ref{lemma_5_goal}) holds verifying Assumption \ref{assu:rank:relaxation3}. 
\end{proof}

We next state the generalizations of Theorems \ref{lem:delta} and \ref{thm:level_CC}. 
\begin{theorem}\label{Relaxed1}
Theorem \textup{\ref{lem:delta}} continues to hold if Assumption \textup{\ref{assu:rank:simple}} is replaced by Assumption  \textup{\ref{assu:rank:relaxation1}}. 
\end{theorem}

\begin{proof}[Proof of Theorem \textup{\ref{Relaxed1}}] 
The proof of Theorem \textup{\ref{Relaxed1}} follows the same argument as the proof of Theorem \ref{lem:delta}.  Assumption \ref{assu:rank:simple} is only used in the proof of Theorem \ref{lem:delta} to ensure $\textup{rk}(I_K\overline{C}_{n_a})=\textup{rk}(I_KC_\infty)$ eventually for $K=K^\dagger(\mu_\infty; B, C_\infty, d_\infty)$, which follows from Assumption \textup{\ref{assu:rank:relaxation1}}. 
\end{proof}

\begin{theorem}\label{Relaxed2}
Theorem \textup{\ref{thm:level_CC}} continues to hold if Assumption \textup{\ref{assu:rank:simple}} is replaced by Assumptions  \textup{\ref{assu:rank:relaxation1} and \ref{assu:rank:relaxation3}}. 
\end{theorem}

\begin{proof}[Proof of Theorem \textup{\ref{Relaxed2}}]
The proof of Theorem \textup{\ref{Relaxed2}} follows the same argument as the proof of Theorem \ref{thm:level_CC} with two minor modifications. 
First, Assumption \textup{\ref{assu:rank:relaxation1}} is sufficient for Theorem \ref{lem:delta}, by Theorem \ref{Relaxed1}. 
Second, when (\ref{rank_convergence}) is required to hold for every $K\in \mathcal{K}(B, C_\infty, h_\infty)\cup \mathcal{L}(B, C_\infty, h_\infty)$ along a subsequence, this follows from Assumption \textup{\ref{assu:rank:relaxation3}}. 
\end{proof}

\noindent\textbf{Remark:} \textit{The relaxation in Theorem \ref{Relaxed2} is a substantial improvement over requiring the stable rank condition to hold for all activatable sets of inequalities in Assumption \ref{assu:rank:simple}. In a moment equality model, weak identification is determined by the Jacobian of the moments. The relaxation to Assumptions \ref{assu:rank:relaxation1} and \ref{assu:rank:relaxation3} means that in a model with both moment equalities and inequalities, the same rank condition that would be required for strong identification of the model defined only by the equalities is sufficient for validity of the GCC (and RGCC) tests. The addition of possibly slack inequalities to a model does not change the rank condition. 
Section \ref{DiscussAssu2} gives further discussion on the relaxation of Assumption \ref{assu:rank:simple}.} 

\subsection{Assumption \ref{assu:rank:simple} Discussion}\label{DiscussAssu2}
This section discusses Assumption \ref{assu:rank:simple} and its relaxations and relates it to existing assumptions in the literature. 
Section \ref{sec:verify} gives a sufficient condition for Assumption \ref{assu:rank:simple} in the empirical illustration of Example \ref{ex:labor}. 
Section \ref{sec:compare} demonstrates the content of Assumptions \ref{assu:rank:relaxation1} and \ref{assu:rank:relaxation3} in simple moment inequality examples. 
Section \ref{nesting_cho_russell} shows that the regularity conditions in \citeapp{ChoRussell2024} are stronger than Assumption \ref{assu:rank:simple}. 

\subsubsection{A Sufficient Condition for the Empirical Illustration of Example \ref{ex:labor}}\label{sec:verify}

Suppose the parameter of interest is $\theta = e'_1\delta$. 
Then, $C_F=\left(\begin{matrix}e_1&-e_1&\Gamma'_F&-\Gamma'_F&A'\end{matrix}\right)'$,  
where $\Gamma_{F}$ is defined in (\ref{eq:KlineTartari_eq}). A subscript $F$ is included to indicate dependence on the true distribution of the sample. Similarly, let $p^{t}_{s,F}$ denote the true marginal probabilities for $t\in\{\text{A}, \text{J}\}$ and $s\in\{\text{0n}, \text{1n}, \text{2n}, \text{0r}, \text{2u}\}$. 
Let $\overline{\Gamma}_n$ be the estimated version, so that 
$\overline{C}_n=\left(\begin{matrix}e_1&-e_1&\overline{\Gamma}'_n&-\overline{\Gamma}'_n&A'\end{matrix}\right)'$. 
In this case, a sufficient condition for Assumption \ref{assu:rank:simple} is that
\begin{itemize}
\item $\exists c>0$ such that $p^{\text{A}}_{s,F}>c$ for all $s\in\{\text{0n}, \text{1n}, \text{2n}, \text{0r}, \text{2u}\}$, $F\in {\cal F}_{n0}$, and $n\ge 1$. 
\end{itemize}
With this condition imposed, under the subsequence in Assumption \ref{assu:rank:simple}, the exact sparse structure of $\overline{\Gamma}_n$ is preserved in the probability limit. As a result, the linear dependence of the rows of $\overline{\Gamma}_n$ with the rows of $A$ is also preserved.\footnote{Note that the row of the Jacobian associated with the inequality ``$\pi_{\text{0r},\text{0n}}+\pi_{\text{0r},\text{1n}}+\pi_{\text{0r},\text{2n}}+\pi_{\text{0r},\text{1r}}+\pi_{\text{0r},\text{2u}}\leq 1$'' is perfectly collinear with the rows of the Jacobian associated with the fourth row in $\Gamma$. This demonstrates the importance of allowing the rank to be singular as long as it is stable.} This ensures that the rank of any collection of the rows of $\overline{C}_n$ is unchanged in the limit.

\subsubsection{Assumption \ref{assu:rank:simple} in Simple Moment Inequality Examples}\label{sec:compare}
This section demonstrates the content of Assumptions \ref{assu:rank:simple}-\textup{\ref{assu:rank:relaxation3}} in simple moment inequality examples. 

\begin{example}\label{ex:strongID0}
Consider the moment inequality model in \textup{(\ref{mi})} with no equalities, $\mathcal{B}=\mathbb{R}^2$, and inequalities specified by 
\begin{align}
\beta_1+\beta_2+\eta_{1F}&\geq 0\nonumber\\
\beta_1-c_F\beta_2+\eta_{2F}&\geq 0\nonumber\\
-\beta_1+\eta_{3F}&\geq 0,\label{eq:strongID0}
\end{align}
where $\beta=(\beta_1,\beta_2)'$ are structural parameters and $\eta_{jF}$ and $c_F$ are reduced-form parameters for $j\in\{1,2,3\}$. 
If the reduced-form parameters are expectations of observed variables, then these inequalities represent moment inequalities; for example, $\eta_{jF} = \E_F[Y_j]$ for observed random variables $Y_j$. 
If $\eta_{jF}$ are all estimated at the $\sqrt{n}$ rate with positive definite asymptotic variance-covariance matrix, then \textup{(\ref{eq:strongID0})} can be written in the form of \textup{(\ref{generalH0})} with $B=-I_3$, $D=\mathbb{O}_{3\times 2}$, $d=\mathbf{0}$, $\Pi=\left[\begin{smallmatrix}1&1\\1&-c_F\\-1&0\end{smallmatrix}\right]$, $\mu=(\eta_{1F},\eta_{2F},\eta_{3F})'$, and $\delta=(\beta_1,\beta_2)'$.\footnote{Alternatively, if only some of the $\eta_{jF}$ are estimated, then \textup{(\ref{eq:strongID0})} can be written in the form of \textup{(\ref{generalH0})} in other ways. For example, if $\eta_{1F}$ and $\eta_{3F}$ are known to be zero, then \textup{(\ref{eq:strongID0})} can be written in the form of \textup{(\ref{generalH0})} with $B=-(0,1,0)'$, $D=-\left[\begin{smallmatrix}1&1\\0&0\\-1&0\end{smallmatrix}\right]$, $d=\mathbf{0}$, $\Pi=(1,-c_F)$, $\mu=\eta_{2F}$, and $\delta=(\beta_1,\beta_2)'$. If $\eta_{2F}$ is known to be zero, then the strategy in Remark (2) above Example \textup{\ref{ex:npiv}} can be used.} 
We maintain that $\eta_{1F}=\eta_{2F}=0$, $\eta_{3F}=1$, and $c_F\ge 0$. 
For fixed $c_F>0$, the inequalities in (\ref{eq:strongID0}) define a triangular region in $\mathbb{R}^2$ of possible values of $\delta$; see Figure \textup{\ref{fig:strongID2}(a)}.  

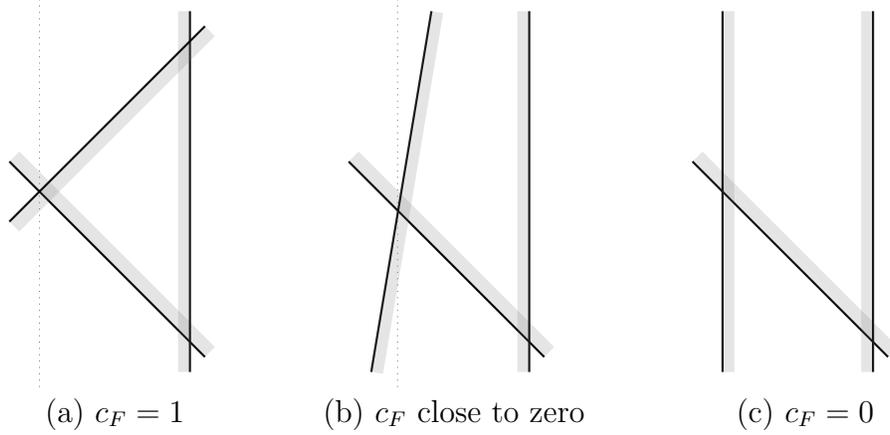
\begin{figure}
\begin{center}
\caption{Illustration of Example \ref{ex:strongID0}}
\label{fig:strongID2}
\centering
\begin{tikzpicture}
\draw[gray, dotted] (0,-2.6) -- (0,2.6);
\draw[thick] (-.4,.4) -- (2.2,-2.2);
\draw[gray,opacity=0.2,line width = 5pt](-.33,.47) -- (2.27, -2.13);
\draw[thick] (2,2.4)--(2,-2.4);
\draw[gray,opacity=0.2,line width = 5pt](1.93,2.4)--(1.93,-2.4);
\draw[thick] (-.4,-.4) -- (2.2,2.2);
\draw[gray, opacity = 0.2, line width = 5pt](-.33,-.47) -- (2.27,2.13);
\node[align=center, below] at (1,-2.6) {(a) $c_F=1$};
\end{tikzpicture}
\hspace{1cm}
\begin{tikzpicture}
\draw[gray, dotted] (.25,-2.6) -- (0.25,2.6);
\draw[thick] (-.4,.4) -- (2.2,-2.2);
\draw[gray,opacity=0.2,line width = 5pt](-.33,.47) -- (2.27, -2.13);
\draw[thick] (2,2.4)--(2,-2.4);
\draw[gray,opacity=0.2,line width = 5pt](1.93,2.4)--(1.93,-2.4);
\draw[thick] (0.7,2.4)--(-.1,-2.4);
\draw[gray,opacity=0.2,line width = 5pt](.77,2.39)--(-.03,-2.41);
\node[align=center, below] at (1,-2.6) {(b) $c_F$ close to zero};
\end{tikzpicture}
\hspace{1cm}
\begin{tikzpicture}
\draw[thick] (-.4,.4) -- (2.2,-2.2);
\draw[gray,opacity=0.2,line width = 5pt](-.33,.47) -- (2.27, -2.13);
\draw[thick] (2,2.4)--(2,-2.4);
\draw[gray,opacity=0.2,line width = 5pt](1.93,2.4)--(1.93,-2.4);
\draw[thick] (0,2.4)--(0,-2.4);
\draw[gray,opacity=0.2,line width = 5pt](.07,2.4)--(.07,-2.4);
\node[align=center, below] at (1.1,-2.6) {(c) $c_F=0$};
\end{tikzpicture}
\end{center}
{\small{\em Note:} Sets of $(\beta_1,\beta_2)'$ defined by (\ref{eq:strongID0}) for various values of $c_F$. The coordinate axes are omitted to avoid clutter. The horizontal axis is $\beta_1$ and the vertical axis is $\beta_2$. Shade indicates the direction of the inequalities. The vertical dotted line in (a) and (b) represents imposing the null hypothesis that $\beta_1=0$.} 
\end{figure}

We walk through the content of Assumptions \textup{\ref{assu:rank:simple}-\ref{assu:rank:relaxation3}} for specification testing the inequalities in \textup{(\ref{eq:strongID0})} in two cases: \textup{(a)} $c_F$ converges to a  $c_\infty>0$, and \textup{(b)} $c_F$ converges to $0$. 
%In addition, we consider two exercises: \textup{(a)} specification testing the inequalities in (\ref{eq:strongID0}) and \textup{(b)} subvector inference for $\beta_1$. 

\textup{(a)} $c_\infty>0$.  
Start with Assumption \textup{\ref{assu:rank:simple}} and the definition of $\mathcal{A}(C_\infty,b_\infty)$. 
Note that $C_\infty=\left[\begin{smallmatrix}-1&-1\\-1&c_\infty\\1&0\end{smallmatrix}\right]$ and $b_\infty=(0,0,1)'$. 
We seek the collections of inequalities that are simultaneously activatable. 
Consider the vertices of $\poly(C_\infty,b_\infty)$, depicted in Figure \textup{\ref{fig:strongID2}(a)}. 
At the leftmost point, the first two inequalities are active. 
At the upper-right point, the second and the third inequalities are active. 
At the lower-right point, the first and the third inequalities  are active. 
Thus, $\mathcal{A}(C_\infty,b_\infty)$ includes $\{1,2\}$, $\{2,3\}$, and $\{1,3\}$. 
Also, $\mathcal{A}(C_\infty,b_\infty)$ includes all subsets of activatable sets, so $\mathcal{A}(C_\infty,b_\infty)=\{\emptyset, \{1\}, \{2\}, \{3\}, \{1,2\}, \{1,3\}$, $\{2,3\}\}$. 
Next, note that there are no equalities, so $K^==\emptyset$. 
Assumption \textup{\ref{assu:rank:simple}} considers submatrices of $C_\infty$ formed by the rows associated with indices in $K$ for all $K\in\mathcal{A}(C_\infty,b_\infty)$. 
For any two rows of $C_\infty$, the rank of the submatrix is $2$. 
Also, any one row of $C_\infty$ is nonzero and thus has rank $1$. 
In both cases, the rank of the limit is full, and therefore the rank stability equation holds automatically for any sequence of matrices converging to $C_\infty$. 

Next, consider Assumption \textup{\ref{assu:rank:relaxation1}}. 
For Assumption \textup{\ref{assu:rank:relaxation1}}, we view the inequalities in \textup{(\ref{eq:strongID0})} as functions of both $\delta=(\beta_1,\beta_2)'$ and $\mu=(\eta_{1F},\eta_{2F},\eta_{3F})'$. 
Plug in $\mu=\mu_\infty$, and the resulting collection of inequalities is still represented by Figure \textup{\ref{fig:strongID2}(a)}. 
In this figure, all the inequalities are possibly inactive, and therefore the minimal collection of inequalities is $K^\dagger(\mu_\infty; B,C_\infty,d_\infty)=\emptyset$. 
Therefore, the rank stability equation for Assumption \textup{\ref{assu:rank:relaxation1}} holds trivially. 

Next, consider Assumption \textup{\ref{assu:rank:relaxation3}}. 
We need to rescale the inequalities by $\sqrt{n}$. 
Notice that $\delta^\ast_F=(0,0)'$ because that is the minimum norm value of $\delta$ that satisfies the inequalities. 
Also note that $\mu_F=(0,0,1)'$. 
Therefore, after rescaling, the limit is $h_\infty=(0,0,\infty)'$. 
Because the third inequality is infinitely slack, it cannot be K- or L- activatable for $(B,C_\infty,h_\infty)$. 
For the first and second inequalities, we consider the CQPP \textup{(\ref{QP})} with some ``shift'' $x\in\mathbb{R}^3$. For any $x$, we can select a $\mu(x)=x$ that solves the CQPP while satisfying the inequalities strictly with some $\delta$. Hence, both inequalities can be slack. 
Therefore, the minimal activatable set for any shift is $K^\dagger(x; B,C_\infty,h_\infty)=\emptyset$. 
Also note that $L(x,\psi)\subseteq K^\dagger(x)$ for any $\psi\in\Psi^\dagger(x)$. 
Therefore, $\mathcal{K}(B,C_\infty,h_\infty)=\mathcal{L}(B,C_\infty,h_\infty)=\{\emptyset\}$, and the rank stability equation in Assumption \textup{\ref{assu:rank:relaxation3}} is satisfied trivially.\footnote{The fact that the K- and L- activatable sets are empty is a general property when the identified set for $\delta$ has positive volume or measure. All inequalities are possibly slack, so the minimal K- and L- activatable sets of inequalities are empty. Thus, the rank stability equations (and Assumptions \textup{\ref{assu:rank:relaxation1}} and \ref{assu:rank:relaxation3}) hold trivially.} 

\textup{(b)} $c_\infty=0$. The limit for this case is depicted in Figure \textup{\ref{fig:strongID2}(c)}, while Figure \ref{fig:strongID2}(b) depicts the case that $c_F$ is small but positive. 
Concerning Assumption \textup{\ref{assu:rank:simple}}, the definition of $\mathcal{A}(C_\infty,b_\infty)$ is changed because there is no longer a vertex associated with the second and the third inequalities. Therefore, $\mathcal{A}(C_\infty,b_\infty)=\{\emptyset,\{1\},\{2\}, \{3\}, \{1,2\}, \{1,3\}\}$. The difference is crucial. The rank stability equation holds for $K\in\mathcal{A}(C_\infty,b_\infty)$, but it would not hold for $K=\{2,3\}$. That is, the rank of $\left[\begin{smallmatrix}1&-c_F\\-1&0\end{smallmatrix}\right]$ is $2$ for $c_F$ positive while the rank of the limit is $1$. 
Still, Assumption \textup{\ref{assu:rank:simple}} holds because $K=\{2,3\}$ is excluded from $\mathcal{A}(C_\infty,b_\infty)$. 

Next, consider Assumptions \textup{\ref{assu:rank:relaxation1}} and \textup{\ref{assu:rank:relaxation3}}. 
For Assumption \textup{\ref{assu:rank:relaxation1}}, all the inequalities are still possibly inactive, and therefore $K^\dagger(\mu_\infty; B, C_\infty, d_\infty)=\emptyset$. 
Similarly, for Assumption \ref{assu:rank:relaxation3}, the limit of the rescaled inequalities is the same. Also, for any shift of the constants, the resulting inequalities are still feasible. Therefore, the minimal activatable set is always the empty set, and $\mathcal{K}(B, C_\infty,h_\infty)=\mathcal{L}(B,C_\infty,h_\infty)=\{\emptyset\}$. 
\end{example}

The next example uses the same inequalities as Example \textup{\ref{ex:strongID0}}. The difference is that it focuses on subvector inference for $\beta_1$. 

\begin{example}\label{ex:strongID0_2}
Consider the inequalities in \textup{(\ref{eq:strongID0})}, but suppose we are testing the lower bound on the identified set for $\beta_1$. 
That is, $H_0: \beta_1=0$. 
This hypothesis is depicted by the vertical dotted line in Figure \textup{\ref{fig:strongID2}(a)}. 
This problem can be written in the form of \textup{(\ref{generalH0})} with $B=-I_3$, $D=\mathbf{0}$, $d=\mathbf{0}$, $\Pi=(1,-c_F,0)'$, $\mu=(\eta_{1F},\eta_{2F},\eta_{3F})'$, and $\delta=\beta_2$. 
Suppose $\eta_{1F}=\eta_{2F}=0$ and $\eta_{3F}=1$. 
We walk through the content of Assumptions \textup{\ref{assu:rank:simple}-\ref{assu:rank:relaxation3}} in two cases: \textup{(a)} $c_F$ converges to a positive value, say $c_\infty$, and \textup{(b)} $c_F$ converges to zero. 

\textup{(a)} $c_\infty>0$. 
In this example, it is obvious that the rank of every submatrix formed by the rows of $C_F=(-1,c_F,0)'$ is unchanged in the limit, and Assumptions \textup{\ref{assu:rank:simple}-\ref{assu:rank:relaxation3}} hold. 
Still, it is instructive to go through the definitions of $\mathcal{A}(C_\infty,b_\infty)$, $K^\dagger(\mu_\infty;B,C_\infty,d)$, $\mathcal{K}(B,C_\infty,h_\infty)$, and $\mathcal{L}(B,C_\infty,h_\infty)$. 

Start with Assumption \textup{\ref{assu:rank:simple}}. 
The set of values for $\beta_2$ defined by the inequalities is just $\{0\}$. The first and second inequalities are active while the third is inactive. 
Therefore, $\mathcal{A}(C_\infty,b_\infty)=\{\emptyset,\{1\},\{2\},\{1,2\}\}$. 

Second, consider Assumption \textup{\ref{assu:rank:relaxation1}}. 
When we plug in $\mu=\mu_\infty$, we get a set of values for $\beta_2$ that is just $\{0\}$, same as above. 
Both the first and second inequalities are implicitly equalities because they cannot be slack. Therefore, the minimal activatable set is $K^\dagger(\mu_\infty; B, C_\infty,d_\infty)=\{1,2\}$.\footnote{We point out a general result that, if $\poly(C_\infty,b_\infty-B\mu_\infty)$ is a singleton, then the rank of $I_{K^\dagger}C_\infty$ must be full. Then, the rank stability equation holds for any sequence of matrices converging to $I_{K^\dagger}C_\infty$.} 

Next, consider Assumption \textup{\ref{assu:rank:relaxation3}}.
When we rescale the inequalities, we get $h_\infty=(0,0,\infty)'$, as before. 
Then, we adjust the constants in the first two inequalities from $\mu_\infty$ to an arbitrary $x$. 
After adjustment, there are three possible cases for the set of $\beta_2$ values that satisfy the inequalities: \textup{(i)} the set is an interval, \textup{(ii)} the set is empty, and \textup{(iii)} the set is a point. 
In case \textup{(i)} both inequalities are possibly slack, and the minimal activatable set is $K^\dagger(x)=\emptyset$. Also, in this case, $L(x,\psi)=\emptyset$ for all $\psi\in\Psi^\dagger(x)$. 
In case \textup{(ii)}, a value of $\mu(x)$ is found solving the CQPP in \textup{(\ref{QP})}. The resulting set of $\beta_2$ values will be a point. Then, both inequalities are implicitly equalities and the minimal activatable set is $K^\dagger(x)=\{1,2\}$. 
Also, since both inequalities must have positive KKT multipliers for the problem in \textup{(\ref{QP})}, $L(x,\psi)=\{1,2\}$ for any $\psi\in\Psi^\dagger(x)$. 
Finally, in case \textup{(iii)}, as in case \textup{(ii)}, both inequalities are implicitly equalities and the minimal activatable set is $K^\dagger(x)=\{1,2\}$. 
The difference here is that the inequalities both have zero KKT multipliers for the problem in \textup{(\ref{QP})}, and therefore $L(x,\psi)=\emptyset$ for any $\psi\in\Psi^\dagger(x)$. 
Overall, $\mathcal{K}(B,C_\infty,h_\infty)=\mathcal{L}(B,C_\infty,h_\infty)=\{\emptyset,\{1,2\}\}$. 

\textup{(b)} $c_\infty=0$. 
%Recall $C_F=(-1,c_F,0)'$. 
When $c_\infty=0$, then the second inequality does not restrict $\beta_2$. 
At the same time, the second inequality is active in that it holds with equality. 
Thus, $\mathcal{A}(C_\infty,b_\infty)=\{\emptyset,\{1\}, \{2\}, \{1,2\}\}$. 
One can see that when $K=\{2\}$, the rank stability equation does not hold, and therefore Assumption \textup{\ref{assu:rank:simple}} does not hold. 

Turning now to Assumption \textup{\ref{assu:rank:relaxation1}}, note that the set of possible values for $\beta_2$ is $[0,\infty)$. 
This means that the first inequality is possibly inactive, while the second inequality must be active and is implicitly an equality (even though it is always zero). 
Thus, the minimal activatable set is $K^\dagger(\mu_\infty; B,C_\infty,d_\infty)=\{2\}$. 
For this set of indices, the rank stability equation does not hold, and therefore Assumption \textup{\ref{assu:rank:relaxation1}} does not hold. 
Similarly for Assumption \textup{\ref{assu:rank:relaxation3}}: there are shifts of the constants that make the second inequality active. Also, the first inequality is possibly inactive for any shift. Therefore, $\mathcal{K}(B,C_\infty,h_\infty)=\mathcal{L}(B,C_\infty,h_\infty)=\{\emptyset, \{2\}\}$. (The empty set arises when a shift makes the second inequality slack.) 
As before, for $K=\{2\}$, the rank stability equation does not hold, and Assumption \textup{\ref{assu:rank:relaxation3}} does not hold. 

To summarize, Assumptions \textup{\ref{assu:rank:simple}-\ref{assu:rank:relaxation3}} are not satisfied when $c_F$ converges to zero. 
Essentially, Assumptions \textup{\ref{assu:rank:simple}-\ref{assu:rank:relaxation3}} restrict the possible values of $c_F$ to belong to $\{0\}\cup[\epsilon,\infty)$ for some $\epsilon>0$, where the case $c_F=0$ requires $c_F$ to be known. In this case, Assumptions \textup{\ref{assu:rank:relaxation1} and \ref{assu:rank:relaxation3}} help to pinpoint the collection of inequalities that violate the rank stability equation. 

This restriction can be compared to assumptions imposed by other approaches to testing inequalities. 
For this example, the regularity conditions in \citeapp{PPHI2015} as well as the regularity conditions in \citeapp{ChoRussell2024} both imply that $c_F$ is bounded away from zero.\footnote{This comes from evaluating Assumptions A1-A4 in \citeapp{PPHI2015} in Example \ref{ex:strongID0}. Similarly, for \citeapp{ChoRussell2024}, it follows from Conditions CQ and US in their Definition 3.1.} 
Similarly, Assumption A.3 in \citeapp{BugniCanayShi2017} and Assumption $4.27$ in \citeapp{GoffMbakop2025} can be shown to imply that $c_F$ belongs to $\{0\}\cup[\underline{c},\infty)$ for some $\underline{c}>0$. 
Hence, it is common in the literature on subvector inference for moment inequality models to rule out the case that $c_F\rightarrow 0$ when doing subvector inference on $\beta_1$. 
\end{example}

Examples \ref{ex:strongID0} and \ref{ex:strongID0_2} demonstrate the content of Assumptions \ref{assu:rank:simple}-\ref{assu:rank:relaxation3} in simple moment inequality models. 
Next, we give an example where Assumption \ref{assu:rank:simple} is not satisfied, but Assumptions \ref{assu:rank:relaxation1} and \ref{assu:rank:relaxation3} are. 
This motivates the relaxation of Assumption \ref{assu:rank:simple} in Assumptions \ref{assu:rank:relaxation1} and \ref{assu:rank:relaxation3}. 

\begin{example}\label{Counter_Example}
Consider the moment inequality model in (\ref{mi}) with no equalities, $\mathcal{B}=\mathbb{R}^2$, and inequalities specified by
\begin{align}
\beta_1+\beta_2+\eta_{1F}&\geq 0\nonumber\\
\beta_1-c_F\beta_2+\eta_{2F}&\geq 0\nonumber\\
\beta_1-\beta_2+\eta_{3F}&\geq 0,\label{eq:Counter_Example}
\end{align}
where $\beta=(\beta_1,\beta_2)'$ are structural parameters and $\eta_{jF}$ and $c_F$ are reduced-form parameters for $j\in\{1,2,3\}$. 
The inequalities in (\ref{eq:Counter_Example}) are similar to the inequalities in (\ref{eq:strongID0}) except there is no upper bound and there is an extra lower bound for $\beta_1$. 

Suppose $\eta_{1F}=\eta_{3F}=0$ and $c_F\rightarrow 0$. 
Also suppose $\eta_{2F}$ is positive and converges to zero but at a rate that is slower than $n^{-1/2}$. 
In this case, we evaluate Assumptions \ref{assu:rank:simple}-\ref{assu:rank:relaxation3} in the context of subvector inference for $\beta_1$. 
Suppose we test the hypothesis that $H_0: \beta_1=0$, which is the lower bound on the identified set for $\beta_1$. 
If the $\eta_{jF}$ for $j\in\{1,2,3\}$ are estimated at the $\sqrt{n}$ rate with positive definite asymptotic variance-covariance matrix, then (\ref{eq:Counter_Example}) under $H_0$ can be written in the form of (\ref{generalH0}) with $B=-I_3$, $D=\mathbf{0}$, $d=\mathbf{0}$, $\Pi=(1,-c_F,-1)'$, $\mu=(\eta_{1F},\eta_{2F},\eta_{3F})'$, and $\delta=\beta_2$.
Note that $C_F=(-1, c_F, 1)'$, and the second row is converging to zero. 

We first evaluate Assumption \ref{assu:rank:simple}. 
The inequalities define the set of possible values of $\beta_2$ to be $\{0\}$. 
All three inequalities are active. %The active inequalities are $K=\{1, 2, 4\}$. 
Thus, $\mathcal{A}(C_\infty,b_\infty)=\{\emptyset,\{1\}, \{2\}, \{3\}, \{1,2\}, \{1,3\}$, $\{2,3\}, \{1,2,3\}\}$. 
It follows that Assumption \ref{assu:rank:simple} is not satisfied because $K=\{2\}\in\mathcal{A}(C_\infty,b_\infty)$.  

Turning to Assumption \ref{assu:rank:relaxation1}, the three active inequalities delineate a point for $\beta_2$, which is an affine subspace. 
Therefore, those active inequalities are implicitly equalities. 
The minimal activatable set is given by $K^\dagger(\mu_\infty; B, C_\infty,d_\infty)=\{1,2,3\}$. 
We see that the rank stability equation is satisfied, and therefore Assumption \ref{assu:rank:relaxation1} holds. 

Finally, consider Assumption \ref{assu:rank:relaxation3}. 
Note that the slackness of the second inequality is converging to zero at a rate slower than $n^{-1/2}$. 
Therefore, after rescaling the slackness of the inequalities by $\sqrt{n}$, the second inequality has an infinite limiting slackness: $h_\infty=(0,\infty,0)'$. 
Then, as in Example \ref{ex:strongID0}, for any arbitrary shift of the inequalities, the K- and L- activatable sets are either empty or $\{1,3\}$. Therefore, $\mathcal{K}(B,C_\infty,h_\infty)=\mathcal{L}(B,C_\infty,h_\infty)=\{\emptyset,\{1,3\}\}$ and Assumption \ref{assu:rank:relaxation3} holds. 
\end{example}

\subsubsection{Nesting with \cite{ChoRussell2024} Assumptions}
\label{nesting_cho_russell}

Consider a moment inequality model as in (\ref{mi}) with no equalities. 
Suppose $m(W,\beta)$ is linear in $\beta$ and thus $\E_F[m(W,\beta)]$ can be written as $\Gamma_F \beta + \eta_F$, where $\Gamma_F = \E_F[\partial m(W,\beta)/\partial\beta']$ and $\eta_F = \E_F[m(W,\mathbf{0})]$. 
Let the parameter of interest be $\theta = \lambda'\beta$, where $\lambda\neq 0$ is a fixed known vector. 
Without loss of generality, we take $\|\lambda\|$ $=1$. 
In this setting, we can apply the confidence interval in \citeapp{ChoRussell2024} for $\theta_{\max}= \max_{\beta:\E_F[m(W,\beta)]\leq 0}\lambda'\beta$. 

Following Remark (2) above Example \ref{ex:interval} in Section \ref{sub:momineq}, let $\Lambda^c$ be a column-augmenting matrix so that $[\lambda,\Lambda^c]$ is full rank. It is convenient to take $\Lambda^c$ to be orthonormal and orthogonal to $\lambda$ because then $\beta = \lambda\theta+\Lambda^c\delta$ with $\delta = (\Lambda^c)'\beta$. 
In our framework, inference for $\theta$ amounts to testing $H_0:\exists \delta \text{ s.t. } B\mu_F+C_F\delta\leq d$ with $B = I$, $C_F = \Gamma_F \Lambda^c$, $\mu_F = \eta_F+\Gamma_F \lambda{\theta}_0$, and $d = \mathbf{0}$, where $\theta_0$ is a hypothesized value of $\theta$. 
The following assumption adapts assumptions from \citeapp{ChoRussell2024} for this model. 

\begin{assumption}\label{assu:suppfun}
Given a matrix $\Gamma_\infty$ and vector $\eta_\infty$, we have that \textup{(a)} $\arg\max_{\beta:\Gamma_\infty\beta+\eta_\infty\leq 0}\lambda'\beta$ is unique (denoted by $\beta^\ast_\infty$), and \textup{(b)} $I_{K}\Gamma_\infty$ has full row rank for $K =\left \{k:e_k'\left(\Gamma_\infty\beta^\ast_\infty +\eta_\infty\right) =0\right\}$. 
\end{assumption}
\noindent\textbf{Remark:} \textit{Assumption \ref{assu:suppfun} adapts the CQ and the US conditions from Definition 3.1 in \citeapp{ChoRussell2024} for the moment inequality model. Part (a) requires the set of $\beta$ values that achieve $\theta_{\max}$ to be a singleton. Part (b) requires the active inequalities at the maximum to have full row rank. In particular, the number of active inequalities cannot be larger than $d_\beta$ (the dimension of $\beta$).}\medskip

Let $\overline{\Gamma}_n$ and $\overline{\eta}_n$ be estimators of $\Gamma_{F}$ and $\eta_{F}$, respectively. 
The next lemma shows that Assumption \ref{assu:suppfun}, applied to the limit along a sequence of $F$ values, implies Assumption \ref{assu:rank:simple}. 
\begin{lemma}\label{lem:suppfun} 
Suppose for any sequence $\{F_n\in {\cal F}_{n0}\}$ and any subsequence $\{n_m\} $ of $\{n\}$, there exists a further subsequence, $n_q$, along which $\Gamma_{F_{n_q}}\to \Gamma_\infty$, $\overline{\Gamma}_{n_q}\to_p\Gamma_\infty$, $\eta_{F_{n_q}}\to \eta_\infty$, and $\overline{\eta}_{n_q}\to_p \eta_\infty$ as $q\to\infty$, and Assumption \textup{\ref{assu:suppfun}} holds for $\Gamma_\infty$ and $\eta_\infty$. Then, Assumption \textup{\ref{assu:rank:simple}} holds. 
\end{lemma}
\noindent\textbf{Remark:} \textit{\citeapp{ChoRussell2024} show that their assumptions can be satisfied by perturbing the inequalities by a random amount. That, combined with Lemma \ref{lem:suppfun}, implies that Assumption \ref{assu:rank:simple} can be satisfied by perturbing the inequalities by a random amount.} 

\begin{proof}[Proof of Lemma \textup{\ref{lem:suppfun}}] 
Fix a sequence, $F_n\in\mathcal{F}_{n0}$, and a subsequence, $n_m$, satisfying Assumption \ref{Assumption1}(i), so $\mu_{F_{n_m}}=\eta_{F_{n_m}}+\Gamma_{F_{n_m}}\lambda\theta_{\max}\rightarrow\mu_\infty$ and $\Gamma_{F_{n_m}}\rightarrow\Gamma_\infty$. 
Let $C_\infty=\Gamma_\infty\Lambda^c$. 
By the condition in the statement of Lemma \ref{lem:suppfun}, there exists a further subsequence along which $\Gamma_{F_{n_q}}\rightarrow \Gamma_\infty$, $\overline{\Gamma}_{n_q}\rightarrow_p \Gamma_\infty$, $\eta_{F_{n_q}}\rightarrow\eta_\infty$, and $\overline{\eta}_{n_q}\rightarrow_p \eta_\infty$. 
First, we show that, under Assumption \ref{assu:suppfun},
\begin{equation}
\textup{poly}(C_\infty, -\mu_\infty) = \{{\Lambda^c}'\beta_\infty^\ast\}.\label{uniqueness}
\end{equation}
In particular, this implies that $\textup{poly}(C_\infty, -\mu_\infty)$ is nonempty. 

To show (\ref{uniqueness}), observe that 
\begin{align}
\textup{poly}(C_\infty, -\mu_\infty)&  = \{\delta: C_\infty\delta+\mu_\infty\leq 0\}\nonumber\\
&=\{\delta: \Gamma_\infty \Lambda^c \delta +\eta_\infty+\Gamma_\infty \lambda\theta_{\max}\leq 0\}\nonumber\\
& = \{\delta:\Gamma_\infty(\Lambda^c\delta+\lambda\theta_{\max})+\eta_\infty\leq 0\}.
\end{align}
We know that $\delta = {\Lambda^c}'\beta_\infty^\ast\in \textup{poly}(C_\infty,-\mu_\infty)$ because $\Gamma_\infty(\Lambda^c{\Lambda^c}'\beta_\infty^\ast+\lambda\theta_{\max})+\eta_\infty = \Gamma_\infty(\Lambda^c{\Lambda^c}' + \lambda \lambda')\beta_\infty^\ast+\eta_\infty = \Gamma_\infty\beta_\infty^\ast +\eta_\infty \leq 0$. 
To verify uniqueness, consider a $\delta_0\in \textup{poly}(C_\infty,-\mu_\infty)$; then $\Gamma_\infty \Lambda^c\delta_0+\eta_\infty+\Gamma_\infty \lambda\theta_{\max}\leq 0$. 
Let $\beta^{\dagger} =\Lambda^c\delta_0+\lambda\theta_{\max}$. 
Then, $\Gamma_\infty \beta^\dagger + \eta_\infty\leq 0$. 
Since the columns of $\Lambda^c$ are orthogonal to $\lambda$, we have $\lambda'\Lambda^c=\mathbf{0}$, and hence $\lambda'\beta^\dagger = \lambda'(\Lambda^c\delta_0+\lambda\theta_{\max}) = \theta_{\max}$. 
Therefore, $\beta^\dagger$ is an $\argmax$ of $\lambda'\beta$ subject to the constraint $\Gamma_\infty\beta+\eta_\infty\leq 0$. 
This, together with Assumption \ref{assu:suppfun}(a), implies that $\beta^\dagger = \beta^\ast_\infty$. 
Multiplying ${\Lambda^c}'$ on both sides, we get: $\delta_0 = {\Lambda^c}'\beta^\ast_\infty$. 
Thus, (\ref{uniqueness}) holds. 

Next, consider
\begin{align}
K&=\{k:e'_k(\Gamma_\infty \beta_\infty^\ast+\eta_\infty) =0\}\nonumber\\
&=\{k:e'_k(\Gamma_\infty \Lambda^c {\Lambda^c}'\beta_\infty^\ast +\Gamma_\infty \lambda\lambda'\beta_\infty^\ast+\eta_\infty) =0\}\nonumber\\
&=\{k:e'_k(\Gamma_\infty \Lambda^c {\Lambda^c}'\beta_\infty^\ast +\mu_\infty) = 0\}\nonumber\\
&=\{k:e'_k(C_\infty {\Lambda^c}'\beta_\infty^\ast +\mu_\infty) = 0\},\label{activatableK}
\end{align}
where the second equality follows because $I=\Lambda^c {\Lambda^c}'+\lambda\lambda'$. 
This shows that $K$ is the set of active inequalities for $\max_{\beta:\Gamma_\infty\beta+\eta_\infty\leq 0}g'\beta$. 
By Theorem 4.5 in \citeapp{Szilagyi2006}, Assumption \ref{assu:suppfun}(a) implies that $I_K\Gamma_\infty$ has full column rank. 
This, combined with Assumption \ref{assu:suppfun}(b), implies that $I_K\Gamma_\infty$ is an invertible square matrix. 
Then, for any $\Gamma_q\rightarrow \Gamma_\infty$, $I_K\Gamma_q$ is also an invertible square matrix eventually as $q\rightarrow\infty$. 
This also implies that 
\begin{equation}
\textup{rk}(I_KC_q) = \textup{rk}(I_K\Gamma_q\Lambda^c) = \textup{rk}(\Lambda^c) = \textup{rk}(I_K\Gamma_\infty \Lambda^c) =\textup{rk}(I_KC_\infty). \label{StableK}
\end{equation}
where $C_q=\Gamma_q\Lambda^c$. 
Combined with the fact that $\Gamma_{F_{n_q}}\rightarrow \Gamma_\infty$ and $\overline{\Gamma}_{n_q}\rightarrow_p\Gamma_\infty$, (\ref{StableK}) implies that the rank stability equation holds with probability approaching one for $K$, where $\overline{C}_n=\overline{\Gamma}_n\Lambda^c$ and $C_{F_n}=\Gamma_{F_n}\Lambda^c$. 

To verify Assumption \textup{\ref{assu:rank:simple}}, let $\widetilde{K}\in{\cal A}(C_\infty, -\mu_\infty)$. 
It follows from the definition of ${\cal A}(C_\infty, -\mu_\infty)$ that $\widetilde{K}\subseteq K$. 
We have already seen that if $\widetilde{K}=K$, then the rank stability equation in Assumption \ref{assu:rank:simple} holds with probability approaching one. 
Below, we show that for any $\widetilde{K}\subsetneq K$,
\begin{equation}
\textup{rk}(I_{\widetilde{K}}C_\infty) = |\widetilde{K}|.\label{K0rank}
\end{equation}
Then, since $\overline{C}_{n_q} = \overline\Gamma_{n_q}\Lambda^c\to_p\Gamma_{\infty}\Lambda^c = C_\infty$ and $C_{F_{n_q}}=\Gamma_{F_{n_q}}\Lambda^c\rightarrow \Gamma_\infty\Lambda^c=C_\infty$, we have $\textup{rk}(I_{\widetilde{K}}\overline{C}_{n_q}) = \textup{rk}(I_{\widetilde{K}}C_{F_{n_q}}) = |\widetilde{K}|$ with probability approaching 1, verifying Assumption \ref{assu:rank:simple}. 

To show (\ref{K0rank}), consider the KKT conditions for $\max_{\beta:\Gamma_\infty\beta+\eta_\infty\leq 0}\lambda'\beta$:
\begin{equation}
\lambda = -\Gamma'_\infty\psi_\infty, ~ \psi_\infty\geq 0,~\Gamma_\infty\beta_\infty^\ast+\eta_\infty\leq 0\text{ and }\psi_\infty'( \Gamma_\infty\beta_\infty^\ast+\eta_\infty) = 0,
\end{equation}
where $\psi_\infty$ is the vector of KKT multipliers. 
The KKT conditions imply that for $k\notin K$, we have $e_k' \psi_\infty=0$. 
This implies that, $\lambda = -\Gamma'_\infty I'_KI_K\psi_\infty$. 
Moreover, Theorems 4.8 and 4.9 of \citeapp{Szilagyi2006} imply that, under Assumption \ref{assu:suppfun}(b), $e_k'\psi_\infty>0$ for all $k\in K$. 
Thus, $\lambda$ is a linear combination of the rows of $I_K\Gamma_\infty$, where all rows receive positive weights. 
Since the rows of $I_K\Gamma_\infty$ are linearly independent, that means $\lambda$ is linearly independent with any strict subset of rows of $I_K\Gamma_\infty$. 
In other words, 
\begin{equation}
\lambda'\neq c'I_{\widetilde{K}}\Gamma_\infty\text{ for any }\widetilde{K}\subsetneq K \textup{ and any }c\in \R^{|\widetilde{K}|}.\label{gindep}
\end{equation} 
To reach a contradiction, suppose for some $\widetilde{K}\subsetneq K$, we have $\textup{rk}(I_{\widetilde{K}}\Gamma_\infty \Lambda^c)<|\widetilde{K}|$. 
Then, there exists a vector $c\neq \mathbf{0}$ such that $c'I_{\widetilde{K}}\Gamma_\infty \Lambda^c = \mathbf{0}$. 
Since the columns of $\Lambda^c$ span the orthogonal complement of $\lambda$, any $\gamma$ such that $\gamma'\Lambda^c=\mathbf{0}$ must be a scalar multiple of $\lambda$. 
Also, since $c\neq \mathbf{0}$, we must have that $c'I_{\widetilde{K}}\Gamma_\infty = \alpha\lambda'$ for some $\alpha\neq 0$. 
This implies $\lambda'=\alpha^{-1}c'I_{\widetilde{K}}\Gamma_\infty$, which violates (\ref{gindep}). This contradiction shows that (\ref{K0rank}) holds. 
\end{proof}

\subsection{The Refined GCC Test}\label{app:refinement}
This subsection defines the refined GCC (RGCC) test. The RGCC test uses the same refinement procedure as the RCC test in CS23. 
To implement the RGCC test, one first calculates the GCC test statistic and critical value. 
If $T_n\notin[cv(1,2\alpha),cv(1,\alpha)]$ or $\widehat{s}_n\neq 1$, then the outcome of the RGCC test is the same as the GCC test. 
However, when $T_n\in[cv(1,2\alpha),cv(1,\alpha)]$ and $\widehat{s}_n=1$, then the critical value is reduced from the original critical value $cv(1,\alpha)$ to $cv(1,\widehat{\beta}_n)$, where $\widehat{\beta}_n$ is defined below. 
This refines the GCC test to reject slightly more often. 

The new ``level'' of the test $\widehat{\beta}_n$ is constructed  based on the inactivity of the inequalities after vertex enumeration, which finds a $d_A \times d_C$ matrix $H$ such that $H\overline{C}_n=0$ and $\{\mu\in\R^{d_\mu}: B\mu+\overline{C}_n\delta\leq d \text{ for some }\delta\in\R^{d_\delta}\}=\{\mu\in\R^{d_\mu}: A\mu \leq g\}$, where $A=HB$ and $g=Hd$. 
The matrix $H$ may be random in that it depends on $\overline{C}_n$. 
Lemma 1 in CS23 ensures that such an $H$ exists and the vertex enumeration algorithm in CS23 calculates it. 
When $\widehat r_n=1$, there exists a $\bar j\in\widehat J$ such that $e'_{\bar j}A\neq 0$.\footnote{Note that $\widehat r_n$ and $\widehat\beta_n$ are only calculated in the case that $\widehat s_n=1$. By Theorem \ref{lem:rhat}(a), this implies that $\widehat r_n\in\{0,1\}$. When $\widehat r_n=0$, then $T_n=0$, so the definition of $\widehat\beta_n$ does not matter (the test does not reject).} 
For simplicity, we take $\bar j=1$. 
We define a measure of inactivity of the $j$-th inequality at any $x\in\R^{d_\mu}$ for $j\in\{2,\dots, d_A\}$: 
\begin{align}
	&\tau_j(x)= \begin{cases}
		\frac{\sqrt{n}\left\|a_1\right\|_{\widetilde{\Sigma}_n}\left(g_j-a_j^{\prime} x\right)}
		{\left\|a_1\right\|_{\widetilde{\Sigma}_n} \| a_j \|_{\widetilde{\Sigma}_n}-a_1^{\prime} \widetilde{\Sigma}_n a_j} 
		& \text { if }\left\|a_1\right\|_{\widetilde{\Sigma}_n}\left\|a_j\right\|_{\widetilde{\Sigma}_n} - a_1^{\prime} \widetilde{\Sigma}_n a_j > 0 \\ 
		\infty 
		& \text { otherwise }
	\end{cases},\label{92}
\end{align}
where $a_j$ and $g_j$ are the $j$-th row of $A$ and $g$, respectively, and $\|a\|_{\Sigma}=\left(a^{\prime} \Sigma a\right)^{1 / 2}$. 
Also define 
\begin{equation}
	\beta(x) = \begin{cases} 
			2\alpha \Phi(\tau(x)) & \text { if } \widehat{r}=1 \\
			\alpha & \text { otherwise }
		\end{cases},\label{93}
\end{equation}
where $\Phi(\cdot)$ is the cumulative distribution function of the standard normal distribution and
\begin{equation}
	\tau(x)=\inf _{j \in\left\{2, \ldots, d_A\right\}} \tau_j(x) \label{94}
\end{equation}
is the smallest inactivity of the inequalities. 
The refined level is $\widehat{\beta}_n=\beta(\widehat\mu)$. 
Since $\tau(\widehat\mu)$ is nonnegative, $\widehat{\beta}_n\in[\alpha,2\alpha]$. 
The refined level is close to its maximum possible value, $2\alpha$, when all inactive inequalities are far from binding. 
For further details of the refinement procedure, including geometric intuition, see CS23. 

The following theorem gives theoretical justification for the refinement. 
\begin{theorem}\label{thm:level_RCC} If Assumptions \textup{\ref{Assumption1}} and \textup{\ref{assu:rank:simple}} hold, or if Assumptions \textup{\ref{Assumption1}} and \textup{\ref{assu:rank:relaxation1}-\ref{assu:rank:relaxation3}} hold, then 
\begin{equation*}
\underset{n\to\infty}{\textup{limsup}}\sup_{F\in{\cal F}_{n0}}P_F(T_n>cv(\widehat s_n,\widehat\beta_n)) \leq\alpha. 
\end{equation*}
\end{theorem}

\begin{proof}[Proof of Theorem \textup{\ref{thm:level_RCC}}] 
The proof of Theorem \ref{thm:level_RCC} follows the same argument as the proof of Theorem \ref{thm:level_CC}. 
We describe the differences here. 
Note that the differences are compatible with the changes in the proof of Theorem \ref{Relaxed2}, so either Assumptions \textup{\ref{Assumption1}} and \textup{\ref{assu:rank:simple}} or Assumptions \textup{\ref{Assumption1}} and \textup{\ref{assu:rank:relaxation1}-\ref{assu:rank:relaxation3}} are sufficient. 

(1) Following (\ref{limitexperiment}), we insert a definition of $\beta^\ast$. 
By Lemma 1 in CS23, there exist $A_\infty$ and $g_\infty$ such that $\poly(A_\infty, g_\infty)=\ppoly(B,h_\infty; C_\infty)$. 
Then, we can define $\beta^\ast$ using (\ref{92})-(\ref{94}) applied to the problem in (\ref{limitexperiment}) with the constraint set given by $\poly(A_\infty, g_\infty)$, so $\beta^\ast=\beta(\eta^\ast_\infty)$. 
Note that by Lemma \ref{equivalent_representation}, any $A_\infty$ and $g_\infty$ such that $\poly(A_\infty, g_\infty)=\ppoly(B,h_\infty; C_\infty)$ yield the same definition of $\beta^\ast$. 

Since $\eta^\ast$ depends on the random variable $X$,  $\beta^\ast$ is  random. Thus, when we fix a sample sequence satisfying (\ref{samplepath}), we have to account for the randomness in $\beta^\ast$. 
For any $K\subseteq\{1,...,d_C\}$, let $\widetilde{P}_K=\Sigma_\infty^{1/2}B'_KM_{C_K}\left(M_{C_K}B_K\Sigma_\infty B'_KM_{C_K}\right)^{+}M_{C_K}B_K\Sigma_\infty^{1/2}$ be the projection matrix onto the span of $\Sigma_\infty^{1/2}B'_KM_{C_K}$, and let $\widetilde{M}_K=I_{d_\mu}-\widetilde{P}_K$ be residual matrix. 
Replace (\ref{samplepath}) with 
\begin{align}
&\|(M_{C_{K}}I_KB\Sigma_\infty^{1/2})^+M_{C_{K}}I_{K}(BX - h_\infty)\|^2_{\Sigma^{-1}_\infty} \neq cv(\textup{rk}(M_{C_{K}}I_KB),\beta({\Sigma_\infty^{1/2}}(\widetilde{M}_K\Sigma^{-1/2}_\infty X+\widetilde{P}_K\xi_K))) \nonumber\\
&~~~~~~~~~~~~~~~~~\text{ for all }K\subseteq\{1,\dots,d_C\}\text{ such that }M_{C_{K}}I_KB\neq\mathbf{0} \text{ and } I_K h_\infty<\infty,\label{samplepath_new}
\end{align}
where $\xi_K$ is defined in Lemma \ref{lem:KKTB2_new}. 
It follows from Lemma \ref{lem:KKTB2_new} that $\eta^\ast_\infty={\Sigma_\infty^{1/2}}(\widetilde{M}_K\Sigma_\infty^{-1/2}X+\widetilde{P}_K\xi_K)$ with $K=K^\ast$ where $K^\ast$ is defined below (\ref{gamma_infty^ast_def}). 
If $s^\ast=1$, then (\ref{samplepath_new}) implies that $T_\infty\neq cv(1,\beta^\ast)$ by plugging in $K=K^\ast$. 

To show that (\ref{samplepath_new}) holds with probability 1, fix a $K\subseteq\{1,\dots,d_C\}$ such that $M_{C_{K}}I_KB\neq\mathbf{0}$ and $I_K h_\infty<\infty$. 
We condition on $\widetilde{M}_K\Sigma^{-1/2}_\infty X$, so the right-hand side is a constant. 
For the left-hand side, note that the conditional distribution of $(M_{C_{K}}I_KB\Sigma_\infty^{1/2})^+M_{C_{K}}I_{K}BX=(M_{C_{K}}I_KB\Sigma_\infty^{1/2})^+M_{C_{K}}I_{K}B\Sigma^{1/2}_\infty \widetilde{P}_K\Sigma^{-1/2}_\infty X$ is equal to the unconditional distribution because $\widetilde{P}_K\Sigma^{-1/2}_\infty X$ and $\widetilde{M}_K\Sigma^{-1/2}_\infty X$ are independent. 
By Lemma \ref{lem:MixedChi2}, that unconditional distribution is continuous because $M_{C_{K}}I_KB\neq\mathbf{0}$. 
This implies that the conditional probability of (\ref{samplepath_new}) is one, and therefore the unconditional probability is also one. 

(2) The second difference is replacing $\alpha$ with $\widehat\beta_n$ and $\beta^\ast$ in various places in the proof of Theorem \ref{thm:level_CC}. 
In the second sentence of the proof, replace $\alpha$ with $\widehat\beta_{n_q}$ twice. 
In (\ref{samplepath}), replace $\alpha$ with $\beta^\ast$. 
In (\ref{goal0}) and (\ref{goal2}), replace the $\alpha$ on the left-hand side of both expressions with $\widehat\beta_{n_q}$ and replace the $\alpha$ on the right-hand side of both expressions with $\beta^\ast$. 
Two sentences after (\ref{goal2}), replace  ``$P(T_\infty>cv(r^\ast,\alpha))\leq \alpha$'' by ``$P(T_\infty>cv(r^\ast,\beta^\ast))\leq \alpha$.''
Finally, in Step 4 parts (A) and (B), replace both occurrences of $\alpha$ with $\widehat\beta_{n_q}$. 
(We modify Step 4 part (C) separately, below.) 

(3) In Step 1, apply Lemma 1 from CS23 to get matrices $H_q$ such that $A_q=H_q B$, $g_q=H_qh_q$, and $\poly(A_q, g_q)=\ppoly(B, h_q; \overline{C}_{n_q})$. 
It follows from Lemma \ref{equivalent_representation} that the definition of $\widehat r_n$ and $\widehat \beta_n$ is the same as using this $A_q$ and $g_q$. 
As in CS23, we can use McMullen's upper bound theorem (see \citeapp{Ziegler1995}) to ensure that the number of rows of $A_q$ is bounded. 
We can take a subsequence so that the number of rows of $A_q$ does not depend on $q$, and denote it by $d_A$. 
We seek to apply Lemma 9 in CS23. 
Without loss of generality, we can take the rows of $A_q$ to either belong to the unit circle or be zero. 
Also note that the $g_q$ are nonnegative (because $\poly(A_q, g_q)$ includes zero). 
Therefore, there exists a sequence of matrices, $G_q$, a sequence of nonnegative vectors, $f_q$, and a further subsequence, $n_q$, such that conditions (a)-(d) in Lemma 9 of CS23 hold. 
In particular, we have elementwise convergence of $[A_q; G_q]$ and $[g_q;f_q]$ to some limit, say $[A^\dagger_\infty; G^\dagger_\infty]$ and $[g^\dagger_\infty; f^\dagger_\infty]$. 
We also have $\poly([A_q; G_q], [g_q;f_q])\overset{K}{\to} \poly([A^\dagger_\infty; G^\dagger_\infty], [g^\dagger_\infty; f^\dagger_\infty])$ and $\poly(A_q, g_q)\subseteq \poly(G_q, f_q)$. 
Again, by Lemma \ref{equivalent_representation}, the definition of $\widehat r_{n_q}$ and $\widehat\beta_{n_q}$ could have been made with respect to this collection of inequalities. 
For simplicity, we can ignore the $G_q$ and $f_q$ and just suppose without loss of generality that the original $A_q$ and $g_q$ satisfy $A_q\rightarrow A^\dagger_\infty$ and $g_q\rightarrow g^\dagger_\infty$ and that $\poly(A_q, g_q)\overset{K}{\to} \poly(A^\dagger_\infty, g^\dagger_\infty)$. 
Note that $\poly(A^\dagger_\infty, g^\dagger_\infty)=\ppoly(B, h_\infty, C_\infty)$. 
To see this, note that $\ppoly(B, h_q; \overline{C}_{n_q})\overset{K}{\to}\ppoly(B, h_\infty, C_\infty)$ by Lemma \ref{projected_poly_convergence}, combined with the fact that $\poly(A_q, g_q)=\ppoly(B, h_q; \overline{C}_{n_q})$ for all $q$ and $\poly(A_q, g_q)\overset{K}{\to} \poly(A^\dagger_\infty, g^\dagger_\infty)$. 
(Among the set of all closed subsets of any set, the Kuratowski limit of a sequence of sets is unique.) 
Therefore, we can define $r^\ast$ and $\beta^\ast$ using $\poly(A^\dagger_\infty, g^\dagger_\infty)$. 
By Lemma 1 in CS23, this definition of $r^\ast$ agrees with the definition following (\ref{gamma_infty^ast_def}). 
By Lemma \ref{equivalent_representation}, this definition of $\beta^\ast$ agrees with the earlier definition given in (1) that used $A_\infty$ and $g_\infty$. 

(4) The final difference is in Step 4, part (C): $r^\ast>0$ and $\widehat s_{n_q}\ge 1$. 
If $r^\ast>1$, then by (\ref{r1conv}), $\widehat{s}_{n_q}>1$ and the proof is unchanged. 
A change is needed if $r^\ast=1$. 
If $\widehat s_{n_q}>1$ along a subsequence, then $\widehat{\beta}_{n_q} = \alpha$ and hence (\ref{goal0}) holds because $cv(\widehat s_{n_q}, \alpha)> cv(1, \alpha)\ge cv(1, \beta^\ast)$, combined with the fact that $T_{n_q}\rightarrow T_\infty\neq cv(1, \beta^\ast)$ by (\ref{samplepath_new}) and (\ref{95}). 
Thus, we focus on the case that $\widehat s_{n_q}=1$. 
Also let $\widehat r_q=\textup{rk}(I_{J_q}A_q)$, where $J_q = \{j\in\{1,\dots,d_A\}: e_j'A_q\widehat{\eta}_q = e_j'g_q\}$. % and $d_{A_q}$ is the number of rows in $A_q$. 
We can take a further subsequence such that $\widehat r_q$ does not depend on $q$. 
By Theorem \ref{lem:rhat}(a), $\widehat r_q\in\{0,1\}$. 
If $\widehat r_q=0$, then $T_q=0$ along this subsequence. 
This implies that $T_\infty=0$ as well, so (\ref{goal0}) holds. 
Thus, we focus on the case that $\widehat r_q=1$ as well. 

In this case, we show that 
\begin{equation}
\limsup_{q\rightarrow\infty}\widehat\beta_{n_q}\le \beta^\ast \label{beta_convergence}
\end{equation}
along a further subsequence. 
Take a subsequence so that $\widehat\beta_{n_q}$ converges. 
If $\widehat\beta_{n_q}\rightarrow\alpha$, then (\ref{beta_convergence}) holds simply because $\beta^\ast\ge\alpha$. 
Suppose $\lim_{q\rightarrow\infty}\widehat\beta_{n_q}>\alpha$. 
For every $q$ large enough, consider a $\bar j_q\in J_q$ such that $e'_{\bar j_q}A_q\neq 0$. 
We can take a further subsequence so that $\bar j_q$ does not depend on $q$, and for simplicity, we suppose it is 1. 
Note that $e'_1A_q\rightarrow e'_1A^\dagger_\infty$, and the rows of $A_q$ all belong to the unit circle, so $e'_1A^\dagger_\infty\neq 0$. 
(The proof of Lemma 9 in CS23 shows that the rows of $G_q$ also can be taken to belong to the unit circle.) 
Also note that $e'_1A_q\widehat\eta_q=e'_1g_q$ for all $q$, so $e'_1A^\dagger_\infty\eta^\ast_\infty=e'_1g^\dagger_\infty$, and $1\in J^\ast$, where $J^\ast=\{j\in\{1,...,d_A\}: e'_jA^\dagger_\infty\eta^\ast_\infty=e'_j g^\dagger_\infty\}$. (By Lemma 1 in CS23, $r^\ast=\textup{rk}(I_{J^\ast}A^\dagger_\infty)$.) 
Therefore, $\widehat\beta_{n_q}$ and $\beta^\ast$ can both be defined using $e_1$ as the reference row. 
Note that $\widehat\beta_{n_q}$ is defined using equations (\ref{92})-(\ref{94}) applied to $A_q, g_q, \widetilde\Sigma_{n_q}$, and $\widehat\eta_q$. 
Let $\tau^q_j$ and $\tau^q$ be the objects defined by (\ref{92}) and (\ref{94}), respectively, evaluated at $\widehat\eta_q$. 
Also note that $\beta^\ast$ is defined using the same equations applied to $A^\dagger_\infty, g^\dagger_\infty, \Sigma_\infty$, and $\eta^\ast_\infty$. 
Let $\tau^\infty_j$ and $\tau^\infty$ be the objects defined by (\ref{92}) and (\ref{94}), respectively, evaluated at $\eta^\ast_\infty$. 

Let $J^==\{j\in\{2,...,d_A\}: \|e'_jA^\dagger_\infty\|\|e'_1A^\dagger_\infty\|=e'_jA^\dagger_\infty (A^\dagger_\infty)'e_1\}$. 
Let $J^{\neq}=\{2,...,d_A\}/J^=$. 
Fix $j\in\{1,...,d_A\}$ and consider two cases. 
(a) If $j\in J^=$, then $\tau^\infty_j=\infty$. 
(b) If $j\in J^{\neq}$, then $\tau^q_j\rightarrow \tau^\infty_j$, using the convergence of $A_q$, $g_q$, $\widetilde\Sigma_{n_q}$, and $\widehat\eta_q$ to $A^\dagger_\infty$, $g^\dagger_\infty$, $\Sigma_\infty$, and $\eta^\ast_\infty$. 
We can then calculate that 
\begin{align}
\lim_{q\rightarrow\infty}\tau^q=\lim_{q\rightarrow\infty}\inf_{j\in\{2,...,d_A\}}\tau^q_j\le \lim_{q\rightarrow\infty}\inf_{j\in J^{\neq}}\tau^q_j=\inf_{j\in J^{\neq}} \tau^\infty_j=\inf_{j\in\{2,...,d_A\}}\tau^\infty_j=\tau^\infty. 
\end{align}
It then follows from the formula for $\widehat\beta_{n_q}$ and $\beta^\ast$ that (\ref{beta_convergence}) holds. 

We finish by showing that 
\begin{equation}
\limsup_{q\rightarrow\infty}1\{T_{n_q}>cv(\widehat s_{n_q}, \widehat\beta_{n_q})\}\le 1\{T_\infty>cv(r^\ast, \beta^\ast)\} \label{goal_modified}
\end{equation}
in this case (which has $\widehat s_{n_q}=1$ and $r^\ast=1$). 
Note that $T_\infty\neq cv(r^\ast, \beta^\ast)$ by (\ref{samplepath_new}). 
If $T_\infty>cv(r^\ast, \beta^\ast)$, then (\ref{goal_modified}) holds automatically. 
If $T_\infty<cv(r^\ast, \beta^\ast)$, then (\ref{beta_convergence}) implies that $cv(r^\ast, \beta_{n_q})\ge cv(r^\ast, \beta^\ast)-\epsilon$ eventually for some $\epsilon<cv(r^\ast, \beta^\ast)-T_\infty$. 
Then (\ref{goal_modified}) follows because $T_{n_q}\rightarrow T_\infty$ and $\widehat s_{n_q}=1=r^\ast$. 
\end{proof}

\subsection{Using KKT Multipliers in the GCC Test}\label{LMGCC}

We can state Theorems \ref{thm:level_CC} and \ref{Relaxed2} with $\widehat t_n$ instead of $\widehat s_n$ as long as the multipliers are chosen carefully. 
The use of the minimum norm multiplies in (\ref{KKT_finite_q1}) in the proof of Theorem \ref{thm:level_CC} shows that, as long as $\widehat t_n$ is defined using the minimum norm multipliers satisfying the KKT conditions, the GCC test defined with $\widehat t_n$ will be valid. 
The key idea is that this sequence of multipliers is guaranteed to converge to some limit, while a different choice of the multipliers would not be. 
We state the validity of the GCC test with $\widehat{t}_n$ in Theorem \ref{Relaxed3} below. 
However, this version of the GCC test is not recommended because it involves additional computational steps: making sure the delta chosen yields the minimal K-activatable set and then minimizing the norm for choosing psi. 
Let 
\[
\psi^\ddagger(x) = \arg\min_{\psi\in\Psi^\dagger(x)} \|\psi\|. 
\]
be a generic definition of the minimum norm multipliers satisfying the KKT conditions for (\ref{QP}), where $\Psi^\dagger(x)$ is defined in \ref{mini_L}. 

\begin{theorem}\label{Relaxed3}
Theorems \ref{thm:level_CC} and \textup{\ref{Relaxed2}} continue to hold for the GCC test defined by $\widehat t_n$ if $L(\overline{\mu}_n,\psi^\ddagger(\overline{\mu}_n))$ is used for $\widehat L$ in \textup{(\ref{rtildeformula})}. 
\end{theorem}

The proof of Theorem \ref{Relaxed3} is omitted because it is the same as the proof of Theorem \ref{thm:level_CC}. 
Note that the proof of Theorem \ref{thm:level_CC} only uses the fact that $\widehat s_n\ge \widehat t_n$ in (\ref{r1conv}). 
This is not affected by the modifications in the proof of Theorem \textup{\ref{Relaxed2}}. 

\section{Lemmas}\label{UsefulLemmas}

Section \ref{E.1_old} states lemmas used in the proof of Theorem \ref{lem:rhat}. 
Section \ref{E.2} states lemmas used in the proofs of Theorems \ref{lem:delta} and \ref{thm:level_CC}. 
Section \ref{E.3} states lemmas used in Section \ref{Section_Generalizations}. 
Section \ref{E.4} proves the lemmas. 

\subsection{Lemmas for the Proof of Theorem \ref{lem:rhat}}\label{E.1_old}

\begin{lemma}\label{lem:MB2} 
Let $B$ be a $d_C\times d_\mu$ matrix and $C$ a $d_C\times d_\delta$ matrix. 
Then
\begin{itemize}
\item[\textup{(a)}] $\textup{rk}(M_CB) = \textup{rk}([B,C])-\textup{rk}(C)$, where $M_C=I_{d_C}-C(C'C)^+C'$.
\item[\textup{(b)}] For any $d_\mu\times d_\delta$ matrix $\Pi$, we have $\textup{rk}([B,C - B\Pi]) = \textup{rk}([B,C])$. 
\item[\textup{(c)}] If $L\subseteq K\subseteq\{1,\dots,d_C\}$, then $\textup{rk}(M_{C_{L}}B_{L})\leq \textup{rk}(M_{C_{K}}B_{K})$, where $C_K=I_KC$. 
\end{itemize}
\end{lemma}

\begin{lemma}\label{lem:KKTB2_new}
Fix $\overline{\mu}_n\in\R^{d_\mu}$ and let $(\widehat\mu_n, \widehat\delta_n, \widehat\psi_n)$ satisfy \textup{(\ref{BCd_KKT_new})-(\ref{BCd_ComplementarySlackness_new})}. 
Fix $K\subseteq\{1,...,d_C\}$ satisfying $\widehat L\subseteq K\subseteq\widehat K$, where $\widehat K=\{j\in\{1,...,d_C\}: e'_j B \widehat \mu_n+e'_j\overline{C}_n\widehat\delta=e'_j d\}$ and $\widehat L=\{j\in\{1,...,d_C\}: e'_j\widehat\psi_n>0\}$. 
Let $\xi_K=\left(M_{C_K}B_{K}\widetilde\Sigma_n^{1/2}\right)^{+} M_{C_K} I_K d$, where $C_K=I_KC_\infty$, $B_K=I_KB$, and $M_{C_K}=I_{|K|}-C_K(C'_KC_K)^+C'_K$. 
It follows that 
\begin{align}
\widetilde\Sigma_n^{-1/2}(\overline{\mu}_n-\widehat \mu_n)&=\left(M_{C_K}B_K\widetilde\Sigma_n^{1/2}\right)^+M_{C_K}B_K\overline{\mu}_n-\xi_K\in \textup{span}(\widetilde\Sigma_n^{1/2}B'_KM_{C_K}), \text{ and }\label{eq:KKTB2_new}\\
\widetilde\Sigma_n^{-1/2}\widehat \mu_n-\xi_K&\in \textup{span}(\widetilde\Sigma_n^{1/2}B'_KM_{C_K})^\perp. \label{eq:KKTB1_new}
\end{align}
\end{lemma}

\subsection{Lemmas for the Proofs of Theorems \ref{lem:delta} and \ref{thm:level_CC}}\label{E.2}

The proofs of Theorems \ref{lem:delta} and \ref{thm:level_CC} rely on the convergence of the argmin of a sequence of CQPPs. 
The following lemma gives a general statement of such convergence. 
It implies that the key condition is Kuratowski convergence of the sequence of constraint sets. 
In the following lemmas, we give sufficient conditions for Kuratowski convergence of the constraint sets. 
This lemma generalizes Lemma 7 in CS23 in that it no longer requires the constraint sets to be polyhedral sets with non-negative constants. 
The non-negativity condition is replaced by the non-emptiness of the limit constraint set. 

\begin{lemma}\label{Projection-convergence}
For $n\in\mathbb{N}\cup\{\infty\}$, let $x_n\in\mathbb{R}^{d_x}$, let $\Sigma_n$ be positive definite and symmetric $d_x\times d_x$ matrices, and let $S_n\subseteq\mathbb{R}^{d_x}$ be closed and convex. 
If, as $n\rightarrow\infty$, $x_n\rightarrow x_\infty$, $\Sigma_n\rightarrow\Sigma_\infty$, and $S_n\rightarrow_K S_\infty$ with $S_\infty\neq\emptyset$, then 
$\argmin_{x\in S_n}(x_n-x)'\Sigma_n^{-1}(x_n-x)\rightarrow \argmin_{x\in S_\infty}(x_\infty-x)'\Sigma_\infty^{-1}(x_\infty-x)$ as $n\rightarrow\infty$. 
\end{lemma} 

The next lemma gives sufficient conditions for Kuratowski convergence of a sequence of polyhedral sets. 
\begin{lemma}\label{K-convergence_nonempty2_new} Consider a sequence of $d_A\times d_\mu$ matrices $\{A_n\}$ and $d_A$-dimensional vectors $h_n$. 
Suppose 
\begin{itemize}
\item[\textup{(i)}] $A_n\to A_0$ and $h_n\to h_0$ (elementwise) for a finite matrix $A_0$ and a $(-\infty,+\infty]^{d_A}$-valued vector $h_0$ as $n\to\infty$, 
\item[\textup{(ii)}] \textup{poly}$(A_n,h_n)\neq \emptyset$ eventually, and 
\item[\textup{(iii)}] $\textup{rk}(I_JA_n) = \textup{rk}(I_JA_0)$ eventually for any $J\subseteq \{1,...,d_A\}$ that is activatable for $\poly(A_0, h_0)$ and such that $\poly(I_JA_0, I_Jh_0)$ is an affine subspace of $\R^{d_\mu}$. 
\end{itemize}
Then $\textup{poly}(A_n,h_n)\overset{K}{\to}\textup{poly}(A_0,h_0)$. 
\end{lemma}

\textbf{Remarks: }(1) \textit{The statement of Lemma \ref{K-convergence_nonempty2_new} is very similar to Lemma 8 in CS23. 
Relative to Lemma 8 in CS23, Lemma \ref{K-convergence_nonempty2_new} does not require $h_n$ to be nonnegative, but instead requires $\poly(A_n,h_n)$ to be nonempty (eventually). 
Lemma \ref{K-convergence_nonempty2_new} also reduces the sets $J\subseteq\{1,...,d_A\}$ for which $\textup{rk}(I_JA_n) = \textup{rk}(I_JA_0)$ eventually from all subsets of $\{1,...,d_A\}$ to only those subsets which are activatable for $\poly(A_0, h_0)$ and for which $\poly(I_JA_0, I_J h_0)$ is an affine subspace of $\R^{d_\mu}$.} 

(2) \textit{The condition that ``$\poly(I_JA_0, I_Jh_0)$ is an affine subspace of $\R^{d_\mu}$'' says that this collection of inequalities is implicitly  a collection of equalities. 
It is equivalent to saying $\{x\in\R^{d_\mu}: I_J A_0x\le I_J h_0\}=\{x\in\R^{d_\mu}: I_J A_0x=I_J h_0\}$. 
That is, all of the inequalities are implicitly forced to hold with equality. 
It is impossible for any one of them to be slack. 
This is a very small subset of the set of all $J\subseteq\{1,...,d_A\}$, even among those that are activatable for $\poly(A_0, h_0)$. 
Lemma \ref{minimal_activatable} implies that this is the set of minimally activatable inequalities for $\poly(A_0,h_0)$. 
Practically, this includes any equalities specified using two inequalities, together with any accidental equalities specified among the inequalities.}\medskip

While Lemma \ref{K-convergence_nonempty2_new} is sufficient for a sequence of polyhedral sets, the next lemma gives Kuratowski convergence of a sequence of polyhedral sets after projecting out nuisance parameters. 
Recall that $\ppoly(B,h;C)=\{x\in\R^{d_\mu}: Bx+C\delta\le h \text{ for some }\delta\in\R^{d_\delta}\}$ for a $d_C\times d_\mu$ matrix $B$, $d_C\times d_\delta$ matrix $C$ and a $d_C\times 1$ vector $h$. 
Here, we allow some of the elements of $h$ be $+\infty$. 
Recall the definition of the minimal activatable set  given in Section \ref{activatability}. 

\begin{lemma}\label{projected_poly_convergence}
Let $C_n$ be a sequence of $d_C\times d_\delta$ matrices converging to $C_\infty$. 
Let $h_n$ be a sequence of nonnegative vectors converging to $h_\infty$. 
Let $B$ be a fixed $d_C\times d_\mu$ matrix. 
Suppose $\textup{rk}(I_J[B, C_n])=\textup{rk}(I_J[B, C_\infty])$ eventually for all $J\subseteq\{1,2,\dots,d_C\}$. 
Also suppose for every $x\in\ppoly(B, h_\infty; C_\infty)$, $\rk(I_JC_n)=\rk(I_JC_\infty)$ eventually as $n\rightarrow\infty$, where $J$ is the minimal activatable set for $\poly(C_\infty, h_\infty-Bx)$. 
Then, $\ppoly(B, h_n; C_n)\overset{K}{\to} \ppoly(B, h_\infty; C_\infty)$. 
\end{lemma}

The next lemma shows that the structure of $C=B\Pi+D$, together with convergence of the minimizers of a CQPP, implies the convergence of the multipliers on the inequalities along a subsequence. 
\begin{lemma}\label{multiplier_convergence}
Let $B$ and $D$ be fixed $d_C\times d_\mu$ and $d_C\times d_\delta$ matrices, respectively. 
Let $\Pi_n$ be a sequence of $d_\mu\times d_\delta$ matrices converging to $\Pi_\infty$. 
Let $h_n$ be a sequence of $d_C$-dimensional vectors converging to $h_\infty$. 
Let $x_n$ be a sequence of $d_\mu$-dimensional vectors converging to $x_\infty$. 
Let $\Sigma_n$ be a sequence of $d_\mu\times d_\mu$ symmetric and positive definite matrices converging to $\Sigma_\infty$, also a symmetric and positive definite matrix. 
For any $n\in\mathbb{N}\cup\{\infty\}$, consider the CQPP
\begin{equation}
\underset{{\{(\mu,\delta): B\mu+C_n\delta\le h_n\}}}{\textup{min}} (x_n-\mu)'\Sigma_n^{-1}(x_n-\mu), \label{multiplier_lemma_CQPP}
\end{equation}
where $C_n=B\Pi_n+D$. 
Let $(\widehat\mu_n, \widehat\delta_n)$ solve \textup{(\ref{multiplier_lemma_CQPP})} with active set $\widehat{K}_n=\{k\in\{1,...,d_C\}: e'_k(h_n-B\widehat\mu_n-C_n\widehat\delta_n)=0\}$. 
For $n\in\mathbb{N}$, let $\widehat\psi_n$ denote the minimum (Euclidean) norm KKT multipliers satisfying the KKT conditions for \textup{(\ref{multiplier_lemma_CQPP})}: 
\begin{align}
\widehat\psi_n&\ge 0, \label{multiplier_lemma_KKT1}\\
2\Sigma_n^{-1}(x_n-\widehat\mu_n)&=B'\widehat\psi_n \label{multiplier_lemma_KKT2}\\
C'_n\widehat\psi_n&=0 \label{multiplier_lemma_KKT3}\\
I_{\widehat K_n^c}\widehat\psi_n&=0 \label{multiplier_lemma_KKT4}
\end{align}
where $\widehat K_n^c=\{1,...,d_C\}/\widehat K_n$. 
 Assume $(\widehat\mu_n,\widehat\delta_n)\rightarrow(\widehat\mu_\infty, \widehat\delta_\infty)$ as $n\rightarrow\infty$. 
Let $n_q$ be a subsequence along which $\widehat K_{n_q}$ does not depend on $q$. 
Then, there exists a $\widehat\psi_\infty$ satisfying \textup{(\ref{multiplier_lemma_KKT1})-(\ref{multiplier_lemma_KKT4})} with $n=\infty$ such that $\widehat \psi_{n_q}\rightarrow \widehat\psi_\infty$ as $q\rightarrow\infty$. 
\end{lemma}

\noindent\textbf{Remark:} {\it The conditions of Lemma \textup{\ref{multiplier_convergence}} can be verified using Lemmas \textup{\ref{K-convergence_nonempty2_new}} and \textup{\ref{projected_poly_convergence}}, which, combined with Lemma \textup{\ref{Projection-convergence}}, lead to the convergence of $(\widehat\mu_n, \widehat\delta_n)$ to $(\widehat\mu_\infty, \widehat\delta_\infty)$. The $\widehat{\psi}_n$ is well defined because the KKT conditions are necessary conditions for the CQPP.} \medskip

The following lemma shows that a convex quadratic form of a Gaussian random vector has a continuous distribution.
\begin{lemma}\label{lem:MixedChi2} 
For any $k\times \ell $ matrix $A$, $k\times 1$ vector $b$, and $k\times k$ positive semi-definite matrix $\Upsilon$, and a random vector $X\sim N(\mathbf{0},\Sigma)$ with a positive semi-definite variance matrix $\Sigma$, if $\Upsilon^{1/2}A\Sigma^{1/2}$ is not a zero matrix, then $\|AX+b\|^2_{\Upsilon}$ has a continuous distribution. 
\end{lemma}

\begin{lemma}\label{s=0_implies_T=0}
Consider 
\begin{equation}
T=\argmin_{\{(\mu,\delta): B\mu+C\delta\le d\}}\|X-\mu\|^2_{\Sigma^{-1}}, \label{generic_CQPP}
\end{equation}
where $B$ and $C$ are matrices, $d$ and $X$ are vectors, and $\Sigma$ is a symmetric and positive definite matrix. 
Let $(\widehat\mu, \widehat\delta, \widehat\psi)$ satisfy the KKT conditions for \textup{(\ref{generic_CQPP})}, where $\widehat\psi$ is the vector of KKT multipliers. 
Let $\widehat L=\{j: e'_j\widehat\psi>0\}$ and $\widehat t=\textup{rk}(I_{\widehat L}[B,C])-\textup{rk}(I_{\widehat L}C)$. 
If $\widehat t=0$, then $\widehat\mu=X$ and $T=0$. 
\end{lemma}
\noindent\textbf{Remark:} {\it One can combine Lemma \textup{\ref{s=0_implies_T=0}} with Theorem \textup{\ref{lem:rhat}(a)} to get that $T=0$ if $\widehat s=0$ or $\widehat r=0$.}\medskip

The next two lemmas are useful linear algebra facts. 
\begin{lemma}\label{LimitRankNonincreasing}
If $A_n$ is a sequence of matrices converging to $A$, then $\textup{rk}(A_n)\ge \textup{rk}(A)$ eventually. 
\end{lemma}

\begin{lemma}\label{ContinuousSingularValues}
Let $A_n$ be a sequence of matrices converging to $A$. 
Let $s_n$ denote a vector of the left singular values of $A_n$ and let $s$ denote a vector of the left singular values of $A$, both in non-increasing order. 
Then, $s_n\rightarrow s$. 
\end{lemma}

\subsection{Lemma for the Proof of Theorem \ref{thm:level_RCC}}\label{E.3}

The next lemma says that two different representations of the same polyhedral set lead to the same values of the projection onto that set, the rank of the active inequalities, and the minimal slackness of the additional (inactive) inequalities. 

\begin{lemma}\label{equivalent_representation}
Let $A$ and $B$ be two matrices and let $c, d$ be two vectors. 
Let $x\in\R^{d_\mu}$ and $\Sigma$ be a positive definite $d_\mu\times d_\mu$ matrix. 
Let $\widehat\mu_1=\argmin_{\mu\in\poly(A,c)} \|x-\mu\|^2_{\Sigma^{-1}}$ and $\widehat\mu_2=\argmin_{\mu\in\poly(B,d)} \|x-\mu\|^2_{\Sigma^{-1}}$. 
Let $\widehat{J}_1=\{j\in\{1,...,d_A\}: e'_jA\widehat\mu_1=e'_jc\}$ and $\widehat{J}_2=\{j\in\{1,...,d_B\}: e'_jB\widehat\mu_1=e'_jd\}$, where $d_A$ and $d_B$ are the number of rows of $A$ and $B$, respectively. 
Let $\widehat r_1=\rk(I_{\widehat{J}_1}A)$ and $\widehat r_2=\rk(I_{\widehat{J}_2}B)$. 
Let $\widehat\beta_1=\beta(\widehat\mu_1)$ with $\beta(x)$ defined in \textup{(\ref{92})-(\ref{94})} with respect to $\poly(A,c)$ and $\widehat\beta_2=\beta(\widehat\mu_2)$ with $\beta(x)$ defined in \textup{(\ref{92})-(\ref{94})} with respect to $\poly(B,d)$. 
If $\textup{poly}(A,c)=\textup{poly}(B,d)$, then $\widehat\mu_1=\widehat\mu_2$, $\widehat r_1=\widehat r_2$, and $\widehat\beta_1=\widehat\beta_2$. 
\end{lemma}

\subsection{Proof of the Lemmas}\label{E.4}
\begin{proof}[Proof of Lemma \textup{\ref{lem:MB2}}]
For part (a), observe that the linear span of the columns of $[B,C]$ is the same as that of $[M_CB, C]$. 
Thus
\begin{align}
\textup{rk}([B,C]) = \textup{rk}([M_CB,C]).\label{rkBC2}
\end{align}
Consider the matrix $A = [M_CB,C]'[M_CB,C]$. 
The eigenvalues of $A$ are the squares of the column singular values of $[M_CB,C]$. 
Thus, $\textup{rk}(A)=\textup{rk}([M_CB,C])$. 
Now observe that
\begin{align}
A=\left(\begin{smallmatrix}B'M_CB&\mathbf{0}\\\mathbf{0}&C'C\end{smallmatrix}\right). 
\end{align}
This block-diagonal feature implies that $\textup{rk}(A) = \textup{rk}(B'M_CB)+\textup{rk}(C'C)$. 
The same argument for $\textup{rk}(A)=\textup{rk}([M_CB,C])$ also shows that $\textup{rk}(B'M_CB)=\textup{rk}(M_CB)$ and $\textup{rk}(C'C)=\textup{rk}(C)$. 
Therefore
\begin{align}
\textup{rk}([M_CB,C])=\textup{rk}(M_CB)+\textup{rk}(C). \label{rkMBC}
\end{align}
This, combined with (\ref{rkBC2}), proves part (a).

Part (b) holds because $[B,C-B\Pi]$ can be obtained from $[B,C]$ by elementary column operations.

For part (c), note that 
\begin{align}
M_{C_K}B_K=B_K-C_K\Delta, 
\end{align}
where $\Delta=(C'_KC_K)^+C'_KB_K$. 
Left-multiplying on both sides by $I_LI_K'$, we have 
\begin{align}
I_LI'_KM_{C_K}B_K=B_L-C_L\Delta,\label{BC2} 
\end{align}
where $I_LI'_KB_K=B_L$ and $I_LI'_KC_K=C_L$ hold because $L\subseteq K$. 
Next, left-multiply (\ref{BC2}) by $M_{C_L}$ to get 
\begin{align}
M_{C_L}I_LI'_KM_{C_K}B_K=M_{C_L}B_L. \label{BC22}
\end{align}
Therefore, $\textup{rk}(M_{C_L}B_L)=\textup{rk}(M_{C_L}I_LI'_KM_{C_K}B_K)\leq \textup{rk}(M_{C_K}B_K)$. 
\end{proof}

\begin{proof}[Proof of \textup{Lemma \ref{lem:KKTB2_new}}]
First note that $\widehat\psi_n=I'_KI_K\widehat\psi_n$ because $K\supseteq \widehat L$. 
It then follows from (\ref{BCd_KKT2_new}) that $C'_KI_K\widehat\psi_n=0$, and therefore $I_K\widehat\psi_n=M_{C_K}I_K\widehat\psi_n$. 
From (\ref{BCd_KKT_new}), it follows that $2n\widetilde\Sigma^{-1}_n(\overline{\mu}_n-\widehat\mu_n)=B'_KI_K\widehat\psi_n=B'_KM_{C_K}I_K\widehat\psi_n$. 
Rearranging, this proves that 
\begin{equation}
\widetilde\Sigma_n^{-1/2}(\overline{\mu}_n-\widehat\mu_n)=\widetilde\Sigma_n^{1/2}B'_KM_{C_K}I_K\widehat\psi_n/(2n)\in\textup{span}(\widetilde\Sigma_n^{1/2}B'_KM_{C_K}). \label{KKTB2_eq_new}
\end{equation}

Next, notice that $0=I_K(B\widehat\mu_n+\overline{C}_n\widehat\delta_n-d)=B_K\widetilde\Sigma_n^{1/2}\widetilde\Sigma_n^{-1/2}\widehat\mu_n+C_K\widehat\delta_n-I_Kd$ because $K\subseteq\widehat K$. 
Plugging in for $\widetilde\Sigma_n^{-1/2}\widehat\mu_n$ using (\ref{KKTB2_eq_new}) gives 
\begin{equation}
B_K\overline{\mu}_n-B_K\widetilde\Sigma_nB'_K M_{C_K} I_K\widehat\psi_n/(2n)+C_K\widehat\delta_n-I_Kd=0. \label{KKTB3_eq_new}
\end{equation}
Left-multiply (\ref{KKTB3_eq_new}) by $(M_{C_K}B_K\widetilde\Sigma_n^{1/2})^+M_{C_K}$ to get 
\begin{equation}
(M_{C_K}B_K\widetilde\Sigma_n^{1/2})^+M_{C_K}B_K\overline{\mu}_n-\widetilde\Sigma_n^{1/2}B'_K M_{C_K} I_K\widehat\psi_n/(2n)-(M_{C_K}B_K\widetilde\Sigma_n^{1/2})^+M_{C_K}I_Kd=0, \label{KKTB4_eq_new}
\end{equation}
which uses the fact that $(M_{C_K}B_K\widetilde\Sigma_n^{1/2})^+M_{C_K}B_K\widetilde\Sigma_nB'_K M_{C_K}=\widetilde\Sigma_n^{1/2}B'_K M_{C_K}$. 
Rearrange, and plug (\ref{KKTB4_eq_new}) into (\ref{KKTB2_eq_new}) to get 
\begin{equation}
\widetilde\Sigma_n^{-1/2}\left(\overline{\mu}_n-\widehat\mu_n\right)=(M_{C_K}B_K\widetilde\Sigma_n^{1/2})^+M_{C_K}B_K\overline{\mu}_n-\xi_K, \label{KKTB5_eq_new}
\end{equation}
which uses the definition of $\xi_K$. 
Combining (\ref{KKTB2_eq_new}) and (\ref{KKTB5_eq_new}) proves (\ref{eq:KKTB2_new}). 

To prove (\ref{eq:KKTB1_new}), rearrange (\ref{KKTB5_eq_new}) to get 
\begin{equation}
\widetilde\Sigma_n^{-1/2}\widehat\mu_n-\xi_K=\left(\widetilde\Sigma_n^{-1/2}-(M_{C_K}B_K\widetilde\Sigma_n^{1/2})^+M_{C_K}B_K\right)\overline{\mu}_n. 
\end{equation}
This, combined with the fact that $M_{C_K}B_K\widetilde\Sigma_n^{1/2}\left(\widetilde\Sigma_n^{-1/2}-(M_{C_K}B_K\widetilde\Sigma_n^{1/2})^+M_{C_K}B_K\right)=0$, implies that $\widetilde\Sigma_n^{-1/2}\widehat\mu_n-\xi_K\in\textup{span}(\widetilde\Sigma_n^{1/2}B'_KM_{C_K})^\perp$, proving (\ref{eq:KKTB1_new}). 
\end{proof}

\begin{proof}[{Proof of Lemma \textup{\ref{Projection-convergence}}}]
Let $\hat x_\infty=\argmin_{x\in S_\infty}(x_\infty-x)'\Sigma_\infty^{-1}(x_\infty-x)$, which exists and is unique because $S_\infty$ is nonempty, closed, convex, and $\Sigma_\infty$ is symmetric and positive definite. 
Let $z_n\in S_n$ such that $z_n\rightarrow \hat x_\infty$, which exists because $S_n\rightarrow_K S_\infty$. 
For $n\in\mathbb{N}$, let $\hat x_n$ denote $\argmin_{x\in S_n}(x_n-x)'\Sigma_n^{-1}(x_n-x)$, which exists and is unique because $S_n$ is nonempty, closed, convex, and $\Sigma_n$ is symmetric and positive definite. 
We want to show that $\hat x_n\rightarrow\hat x_\infty$. 

First note that 
\begin{equation}
(x_n-\hat x_n)'\Sigma_n^{-1}(x_n-\hat x_n)\le (x_n-z_n)'\Sigma_n^{-1}(x_n-z_n)\rightarrow (x_\infty-\hat x_\infty)'\Sigma_\infty^{-1}(x_\infty-\hat x_\infty)
\end{equation}
as $n\rightarrow\infty$. 
Taking the limsup, we get 
\begin{equation}
\limsup_{n\rightarrow\infty} ~ (x_n-\hat x_n)'\Sigma_n^{-1}(x_n-\hat x_n)\le (x_\infty-\hat x_\infty)'\Sigma_\infty^{-1}(x_\infty-\hat x_\infty). \label{lowerbd2}
\end{equation}
It follows from (\ref{lowerbd2}) that $\hat x_n$ is bounded. 
Let $n_m$ be an arbitrary subsequence. 
There exists a further subsequence, say $n_q$, such that $\hat x_{n_q}$ converges to some limit, say $y_\infty$. 
It is sufficient to show that $y_\infty=\hat x_\infty$. 
(If every subsequence has a further subsequence that converges to $\hat x_\infty$, then the original sequence must converge to $\hat x_\infty$.) 

It follows from $S_n\rightarrow_K S_\infty$ that $y_\infty\in S_\infty$. 
Next, note that 
\begin{equation}
\lim_{q\rightarrow\infty} (x_{n_q}-\hat x_{n_q})'\Sigma_{n_q}^{-1}(x_{n_q}-\hat x_{n_q})=(x_\infty-y_\infty)'\Sigma_\infty^{-1}(x_\infty-y_\infty)\ge (x_\infty-\hat x_\infty)'\Sigma_\infty^{-1}(x_\infty-\hat x_\infty), \label{upperbd02}
\end{equation}
where the inequality follows from the fact that $y_\infty\in S_\infty$. 
Combining (\ref{lowerbd2}) and (\ref{upperbd02}), we have that equality holds in (\ref{upperbd02}). 
This implies that $y_\infty=\hat x_\infty$ by the uniqueness of $\hat x_\infty$. 
\end{proof}

The following lemma is used in the proof of Lemma \ref{K-convergence_nonempty2_new}. 
It is convenient to state and prove it separately. 
\begin{lemma}\label{K-convergence_nonempty2_new2} Consider a sequence of $d_A\times d_\mu$ matrices $\{A_n\}$ and $d_A$-vectors $h_n$. 
Suppose 
\begin{itemize}
\item[\textup{(i)}] $A_n\to A_0$ and $h_n\to h_0$ (elementwise) for a finite matrix $A_0$ and a $(-\infty,+\infty]^{d_A}$-valued vector $h_0$ as $n\to\infty$, and 
\item[\textup{(ii)}] \textup{poly}$(A_n,h_n)\neq \emptyset$ eventually and $\poly(A_0, h_0)\neq \emptyset$. 
\end{itemize}
Fix $z_0\in \poly(A_0, h_0)$ and let $J=\{j\in\{1,...,d_A\}: e'_j A_0 z_0=e'_j h_0\}$. 
Further suppose 
\begin{itemize}
\item[\textup{(iii)}] $\textup{rk}(I_JA_n) = \textup{rk}(I_JA_0)$ eventually, and 
\item[\textup{(iv)}] $\poly(I_JA_0, I_Jh_0)$ is a linear subspace of $\R^{d_\mu}$. 
\end{itemize}
Then, there exists a sequence $z_n\in\textup{poly}(A_n,h_n)$ eventually such that $z_n\rightarrow z_0$.\footnote{The ``eventually'' here means that the sequence $z_n$ must eventually belong to the set $\poly(A_n, h_n)$. This accommodates the possibility that $\poly(A_n, h_n)$ is empty for finitely many values of $n$.} 
\end{lemma}

\noindent\textbf{Remark:} \textit{The intuition for the proof of Lemma \textup{\ref{K-convergence_nonempty2_new2}} is as follows. 
In the usual case, where the inequalities define a polyhedron that does not belong to a proper subspace of $\R^{d_\mu}$, no rank condition is needed. 
Any point in the interior can be easily approximated by a sequence of points in the sequence of polyhedrons. 
Any point not on the interior can be approximated by a sequence of points in the interior. 
The rank condition is only needed when the limit polyhedron belongs to a proper linear subspace of $\R^{d_\mu}$. 
Then, we look at the relative interior of the polyhedron (relative to the subspace). 
Every point on the relative boundary can be approximated by a sequence of points in the relative interior. 
The only problem is approximating points in the relative interior. 
For this we use the rank condition to ensure that the finite-n inequalities also define a subspace of the same dimension. 
The convergence of the subspaces then allows us to approximate any point in the limit subspace by a sequence of points in the finite-n subspaces.}

\begin{proof}[Proof of Lemma \textup{\ref{K-convergence_nonempty2_new2}}] 
First, let $S=\poly(I_JA_0, I_Jh_0)$ be a linear subspace of $\R^{d_\mu}$. 
We show that it is without loss of generality to assume $z_0=0$ and $h_0\geq 0$. 
If not, we can let $g_n=h_n-A_nz_0$ and $g_0=h_0-A_0z_0$. 
Notice that $g_n\rightarrow g_0$. Also notice that $g_0\ge 0$ because $z_0\in \poly(A_0,h_0)$. 
If there exists a sequence $x_n\in\poly(A_n, g_n)$ eventually converging to $0\in\poly(A_0,g_0)$, then $x_n+z_0\in\poly(A_n,h_n)$ eventually and converges to $z_0$. 

When $z_0=0$, $J=\{j\in\{1,...,d_A\}: e'_jh_0=0\}$. 
Let $r=\textup{rk}(I_J A_0)$. 
Let $K$ be a subset of $J$ with $r$ elements such that $\textup{rk}(I_KA_0)=r$. 
We first note that $\textup{rk}(I_K A_n)\ge r$ eventually by Lemma \ref{LimitRankNonincreasing}. 
Second, we note that $\textup{rk}(I_K A_n)\le \textup{rk}(I_J A_n)= r$ eventually by condition (iii). 
Therefore, $\textup{rk}(I_K A_n)=r$ eventually. 

For every $n$, let $x_n\in\poly(A_n, h_n)$, which exists eventually by condition (ii). 
Let 
\[
z_n=A'_nI'_K(I_KA_nA'_nI'_K)^{-1}I_K A_n x_n, 
\]
which is well defined whenever $I_KA_n$ has rank $r$. 
We show that $z_n\in\poly(I_JA_n, I_Jh_n)$ eventually. 
Note that the rows of $I_JA_n$ belong to the linear span of the rows of $I_KA_n$ whenever $\textup{rk}(I_KA_n)=\textup{rk}(I_JA_n)$, which happens eventually. In that case, there is a $d_A\times r$ matrix $U_n$ such that $I_JA_n = U_nI_KA_n$. Therefore
\begin{align}
I_JA_nz_n = U_n(I_KA_nA_n'I_K')(I_KA_nA_n'I_K')^{-1}I_KA_nx_n = U_nI_KA_nx_n = I_JA_nx_n.
\end{align}
Combined with $I_JA_nx_n\leq I_Jh_n$ (since $x_n\in \poly(A_n, h_n)$), we have $z_n\in\poly(I_JA_n, I_Jh_n)$ eventually. 

We show by contradiction that $z_n\rightarrow 0$. 
If not, then there exists an $\epsilon>0$ and a subsequence $n_m$ such that $\|z_{n_m}\|\ge\epsilon$ for all $m$. 
Let $u_m=\|z_{n_m}\|^{-1}z_{n_m}$, which belongs to the unit circle. 
Take a further subsequence, $m_q$, such that $u_{m_q}$ converges to $u_0$, also in the unit circle, as $q\rightarrow\infty$. 
To simplify notation, let $u_q=u_{m_q}$, $z_q=z_{n_{m_q}}$, $x_q=x_{n_{m_q}}$, $A_q=A_{n_{m_q}}$, and $h_q=h_{n_{m_q}}$. 
First, we note that $u_0\in S=\poly(I_J A_0, 0)$ because, for any $j\in J$, 
\[
e'_jA_0 u_0=\lim_{q\rightarrow\infty}\|z_q\|^{-1}e'_j A_qz_q\le \lim_{q\rightarrow\infty}\|z_q\|^{-1}e'_j h_q=0, 
\]
where the inequality follows because $z_q\in \poly(I_J A_q, I_J h_q)$ eventually and the final equality follows because $e'_j h_q\rightarrow e'_jh_0=0$. 
Second, we show that $u_0\perp S$. 
Note that for any $v\in S$, and for any $j\in J$, $e'_j A_0 v=0$ (because, otherwise, either $x=v$ or $x=-v$ would satisfy $e'_j A_0 x>0$, so we would not be in the case that $S=\poly(I_JA_0,0)$ is a linear subspace of $\R^{d_\mu}$). 
Therefore, for any $v\in S$, 
\[
v'u_0=\lim_{q\rightarrow\infty}\|z_q\|^{-1}v'z_q=\lim_{q\rightarrow\infty}v'A'_qI'_K(I_K A_qA'_q I'_K)^{-1}\|z_q\|^{-1}I_KA_qx_q=0, 
\]
where the final equality follows because $v'A'_qI'_K\rightarrow v'A'_0 I'_K=0$, $(I_K A_qA'_q I'_K)^{-1}=O(1)$, and $\|z_q\|^{-1}I_KA_qx_q=I_K A_qu_q=O(1)$ (because $u_q=O(1)$ and $I_K A_q=O(1)$). 
Therefore, $u_0\in S\cap S^\perp=\{0\}$. 
This is a contradiction because $u_0=0$ does not belong to the unit circle. 
Therefore, $z_n\rightarrow 0$. 

To finish the proof, we show that $z_n\in\poly(A_n, h_n)$ eventually. 
We have already shown that $z_n\in\poly(I_JA_n, I_Jh_n)$ eventually. 
For any $j\notin J$, we have $a_{j,n}z_n\rightarrow 0<e'_j h_0$, so $a_{j,n}z_n<e'_j h_n$ eventually. 
Therefore, $z_n\in\poly(A_n, h_n)$ eventually. 
\end{proof}

\begin{proof}[Proof of Lemma \textup{\ref{K-convergence_nonempty2_new}}] 
Let $a_{j,0}'$ denote the $j$th row of $A_{0}$ and let $a_{j,n}'$ denote the $j$th row of $A_n$. 
For the first half of the definition of Kuratowski convergence, let $n_q$ be a subsequence and $z_q$ be a sequence such that $z_q\in\textup{poly}(A_{n_q},h_{n_q})$ for all $q$ and $z_q\rightarrow z_0\in\R^{d_\mu}$ as $q\rightarrow\infty$. 
Then, 
\begin{equation}
a'_{j,0}z_0 = \lim_{q\rightarrow\infty}a'_{j,n_q} z_q\le \limsup_{q\rightarrow\infty}h_{j,n_q}=h_{j,0}, 
\end{equation}
showing that $z_0\in\textup{poly}(A_0,h_0)$. 

For the second half of the definition of Kuratowski convergence, let $z_0\in\textup{poly}(A_0,h_0)$.\footnote{If $\poly(A_0, h_0)$ is empty, then there is nothing to prove for the second half of the definition of Kuratowski convergence. The first half of the convergence implies that $\poly(A_n, h_n)$ Kuratowski converges to the empty set.} 
We seek a sequence, $z_n\in \poly(A_n, h_n)$ eventually, such that $z_n\rightarrow z_0$.\footnote{As in the proof of Lemma \ref{K-convergence_nonempty2_new2}, the ``eventually'' here means that the sequence $z_n$ must eventually belong to the set $\poly(A_n, h_n)$. This accommodates the possibility that $\poly(A_n, h_n)$ is empty for finitely many values of $n$.} 

We reduce to the case that $z_0=0$. 
Let $g_n=h_n-A_nz_0$ and $g_0=h_0-A_0z_0$. 
Notice that $g_n\rightarrow g_0$. 
(Also notice that $g_0\ge 0$ because $z_0\in \poly(A_0,h_0)$.) 
If there exists a sequence $x_n\in\poly(A_n, g_n)$ eventually converging to $0\in\poly(A_0,g_0)$, then $x_n+z_0\in\poly(A_n,h_n)$ eventually and converges to $z_0$. 
Thus, it is sufficient to prove the second half of Kuratowski convergence assuming $z_0=0$ and $h_0\ge 0$. 

Let $J=\{j\in\{1,...,d_A\}: e'_jh_0=0\}$. 
This is the active set for $\poly(A_0, h_0)$ at $z_0=0$. 
Let $S$ denote the smallest linear subspace of $\R^{d_\mu}$ that contains $\poly(I_JA_0, I_Jh_0)$. 
Let $J^S=\{j\in J: a_{j,0}\perp S\}$. 
By Lemma 11 in CS23, there exists a $\widetilde x\in S$ such that $e'_jA_0\widetilde x<e'_j h_0$ for all $j\in \{1,...,d_A\}/J^S$. 
Notice that, for any $j\in J^S$, $e'_j A_0 \widetilde x =0=e'_j h_0$ (because $\widetilde x\in S$), so $\widetilde x\in\poly(A_0,h_0)$. 

Let $\lambda_m$ be a sequence in (0,1] converging to zero as $m\rightarrow\infty$. 
For every $m$, notice that $\lambda_m \widetilde x\in \poly(A_0,h_0)$ with active inequalities given by $J^S$. 
Also notice that by Lemma 13 in CS23, $S=\poly(I_{J^S}A_0, I_{J^S}h_0)$, which is a linear subspace of $\R^{d_\mu}$. 
(Note that $I_{J^S}h_0=0$ because $J^S\subseteq J$.) 
Therefore, the existence of a sequence $x_{m,n}$ such that $x_{m,n}\in \poly(A_n,h_n)$ eventually and $x_{m,n}\rightarrow \lambda_m\widetilde x$ as $n\rightarrow\infty$ follows from Lemma \ref{K-convergence_nonempty2_new2} (using Condition (iii) in Lemma \ref{K-convergence_nonempty2_new} to satisfy Condition (iii) in Lemma \ref{K-convergence_nonempty2_new2}). 
Let $m=m_n$ grow sufficiently slowly so that $x_{{m_n}, n}\in \poly(A_n, h_n)$ eventually and $\|x_{{m_n}, n}-\lambda_{m_n}\widetilde x\|\le m_n^{-1}$ eventually. 
Then, by the triangle inequality, $x_{{m_n}, n}\rightarrow 0=z_0$ as $n\rightarrow\infty$. 
This verifies the second half of Kuratowski convergence and concludes the proof of Lemma \ref{K-convergence_nonempty2_new}. 
\end{proof}

The following lemma is used in the proof of Lemma \ref{projected_poly_convergence}. 
In general, a sequence of K-converging sets, intersected with a fixed set, does not K-converge. 
(See Remark 3 after Theorem 1 in \citeapp{CoxArgmax} for a simple counterexample.) 
This lemma shows that when the sets are convex, intersection preserves K-convergence. 
\begin{lemma}\label{BoxLemma}
Let $S_n$ be a sequence of convex sets K-converging to $S$, a closed and convex set. 
Let $T$ be a closed convex set such that $S\cap \text{int}(T)\neq \emptyset$, where $\text{int}(\cdot)$ denotes the interior of a set. 
Then, $S_n\cap T\rightarrow_K S\cap T$.  
\end{lemma}
\noindent\textbf{Remark:} {\it Lemma \textup{\ref{BoxLemma}} can be applied with $S_n=\textup{poly}(A_n, b_n)$, $S=\textup{poly}(A,b)$ with $b\ge 0$, and $T=\{x: \|x\|_\infty\le\rho\}$ for some $\rho>0$ because $0\in S\cap \text{int}(T)$.}
\begin{proof}[Proof of Lemma \textup{\ref{BoxLemma}}]
Let $y\in S\cap \text{int}(T)$. 
Let $\epsilon>0$ such that $\{x: \|x-y\|<\epsilon\}\subseteq T$. 
It follows from $S_n\rightarrow_K S$ that there exist $y_n\in S_n$ such that $y_n\rightarrow y$. 
We show K-convergence in two steps. 

(1) Let $x\in S\cap T$. 
It follows from $S_n\rightarrow_K S$ that there exist $x_n\in S_n$ such that $x_n\rightarrow x$. 
Let $\lambda_n=\textup{min}(2\epsilon^{-1}\|x_n-x\|,1)$. 
Let $z_n=(1-\lambda_n)x_n+\lambda_ny_n$. 
Then, $z_n\in S_n$ because $S_n$ is convex. 
Also, $z_n\rightarrow x$ because $\lambda_n\rightarrow 0$. 
If $\lambda_n=0$, then $\|x_n-x\|=0$ and $z_n=x\in T$. 
To show that $z_n\in T$ eventually when $\lambda_n>0$, let $z^\ast_n=\lambda_n^{-1}z_n+(1-\lambda_n^{-1})x$ so that $z_n=\lambda_n z^\ast_n+(1-\lambda_n)x$. 
We then have $\|z_n^\ast-y\|=\|\lambda_n^{-1}(x_n-x)+x-x_n+y_n-y\|\le\epsilon/2+\|x_n-x\|+\|y_n-y\|$, which is less than $\epsilon$ eventually. 
When $\|z_n^\ast-y\|<\epsilon$, then $z^\ast_n\in T$, and, by the convexity of $T$, $z_n\in T$. 

(2) Let $x_n\in S_n\cap T$ and let $n_q$ be an arbitrary subsequence along which $x_{n_q}\rightarrow x$. 
It follows from $S_n\rightarrow_K S$ that $x\in S$. 
It also follows from the closedness of $T$ that $x\in T$. 
Therefore, $x\in S\cap T$, completing the second half of K-convergence. 
\end{proof}

\begin{proof}[Proof of Lemma \textup{\ref{projected_poly_convergence}}]
Given the assumption that $\textup{rk}(I_J[B, C_n])=\textup{rk}(I_J[B, C_\infty])$ eventually for all $J$, we can invoke Lemma \ref{K-convergence_nonempty2_new} to get that 
\begin{equation}
\poly([B, C_n], h_n)\to_K\poly([B, C_\infty], h_\infty). \label{projected_poly_convergence1}
\end{equation}

For the first half of Kuratowski convergence, consider an arbitrary $x\in\ppoly(B, h_\infty; C_\infty)$. 
Then, there exists a $\delta$ such that $(x, \delta)\in\poly([B, C_\infty], h_\infty)$. 
By (\ref{projected_poly_convergence1}), there exists a sequence $(x_n, \delta_n)\in\poly([B, C_n], h_n)$ such that $x_n\rightarrow x$ and $\delta_n\rightarrow\delta$. 
It follows that $x_n\in\ppoly(B, h_n; C_n)$. 

For the second half of the Kuratowski convergence, consider an arbitrary subsequence $n_q$ and an arbitrary sequence of values $x_q\in\ppoly(B, h_{n_q}; C_{n_q})$ such that $x_q\rightarrow x$. 
We want to show that $x\in\ppoly(B, h_\infty; C_\infty)$. 
Note that $\ppoly(B, h_\infty; C_\infty)\neq\emptyset$ because $h_\infty\ge 0$. 
Let $y$ be the projection of $x$ onto $\ppoly(B, h_\infty; C_\infty)$. 
Since $y\in\ppoly(B, h_\infty; C_\infty)$, there exists a $\delta\in\poly(C_\infty, h_\infty-By)$. 
Let $\rho>\|(y, \delta)\|_{\infty}$. 
Let $\text{Box}(\rho)=\{(x,\delta)\in\R^{d_\mu}\times \R^{d_\delta}: \|(x,\delta)\|_\infty\le\rho\}$. 
Let $\ppoly_\rho(B, h_{n_q}; C_{n_q})=\{x\in\R^{d_\mu}: (x,\delta)\in\poly([B, C_{n_q}], h_{n_q})\cap\text{Box}(\rho) \text{ for some }\delta\}$. 
Let $z_q$ denote the projection of $x_q$ onto $\ppoly_\rho(B, h_{n_q}; C_{n_q})$. 
Let $\gamma_q$ be such that $(z_q, \gamma_q)\in\poly([B, C_{n_q}], h_{n_q})\cap\text{Box}(\rho)$. 
Take a further subsequence so that $(z_q, \gamma_q)$ converges to some $(z, \gamma)\in\text{Box}(\rho)$. 
It follows from Lemma \ref{BoxLemma} that 
\begin{equation}
\poly([B, C_{n_q}], h_{n_q})\cap \text{Box}(\rho)\overset{K}{\to} \poly([B, C_\infty], h_\infty)\cap \text{Box}(\rho). \label{projected_poly_convergence2}
\end{equation}
Therefore, $(z, \gamma)\in\poly([B, C_\infty], h_\infty)$. 
Notice that 
\begin{equation}
\|x_q-z_q\|^2\rightarrow \|x-z\|^2\ge \|x-y\|^2, \label{projected_poly_convergence3}
\end{equation}
where the inequality follows because $z\in\ppoly(B, h_\infty; C_\infty)$. 
Note that the inequality in (\ref{projected_poly_convergence3}) holds with equality if and only if $z=y$ (by uniqueness of projection onto a convex set). 
Also, by (\ref{projected_poly_convergence2}), there exists a $(y_q, \delta_q)\in \poly([B, C_{n_q}], h_{n_q})\cap \text{Box}(\rho)$ such that $(y_q, \delta_q)\rightarrow (y, \delta)$. 
Notice that 
\begin{equation}
\|x_q-z_q\|^2\le \|x_q-y_q\|^2\rightarrow \|x-y\|^2. \label{projected_poly_convergence4}
\end{equation}
where the inequality follows because $y_q\in\ppoly(B, h_{n_q}; C_{n_q})$. 
It follows from (\ref{projected_poly_convergence3}) and (\ref{projected_poly_convergence4}) that $\lim_{q\rightarrow\infty}\|x_q-z_q\|^2=\|x-z\|^2=\|x-y\|^2$. 
Therefore, $z=y$. 

Next, let $J$ be the minimal activatable set for $\poly(C_\infty, h_\infty-By)$. 
By assumption, $\textup{rk}(I_J C_{n_q})=\textup{rk}(I_JC_\infty)$ eventually as $q\rightarrow\infty$. 
Therefore, by Lemma \ref{K-convergence_nonempty2_new}, $\poly(C_{n_q}, h_{n_q}-Bz_q)\overset{K}{\to} \poly(C_\infty, h_\infty-By)$. 
Since $\delta\in\poly(C_\infty, h_\infty-By)$, there exists $\widetilde\delta_q\in \poly(C_{n_q}, h_{n_q}-Bz_q)$ such that $\widetilde\delta_q\rightarrow \delta$. 
The fact that $\rho>\|(y, \delta)\|_\infty$ implies that $\|(z_q, \widetilde\delta_q)\|<\rho$ eventually. 
Recall that $z_q$ is the projection of $x_q$ onto $\ppoly_\rho(B, h_{n_q}; C_{n_q})$. 
Since the restriction that $(z_q, \widetilde\delta_q)\in \text{Box}(\rho)$ is not binding, it follows that $z_q=x_q$.\footnote{To clarify this argument, let $\lambda_q$ be such that $(x_q, \lambda_q)\in\poly([B, C_{n_q}], h_{n_q})$. If $z_q\neq x_q$, then note that $\epsilon (x_q, \lambda_q)+(1-\epsilon)(z_q, \widetilde\delta_q)\in \poly([B, C_{n_q}], h_{n_q})\cap \text{Box}(\rho)$ for $\epsilon>0$ small enough. This shows that $z_q$ cannot be the projection of $x_q$ onto $\ppoly_\rho(B, h_{n_q}; C_{n_q})$. Therefore, $z_q=x_q$.} 
Therefore, $x_q\rightarrow y$, which implies that $y=x$ and $x\in\ppoly(B, h_\infty; C_\infty)$. 
\end{proof}

\begin{proof}[Proof of Lemma \textup{\ref{multiplier_convergence}}]
Fix the given subsequence and denote $\widehat K_{n_q}$ by $\widehat K$. 
For $q\in\mathbb{N}$, let 
\[
\widetilde U_q=\left[\begin{array}{c}-I_{d_C}\\B'\\-B'\\C'_{n_q}\\-C'_{n_q}\\I_{\widehat K^c}\\-I_{\widehat K^c}\end{array}\right] \text{ and }\widetilde b_q=\left[\begin{array}{c}0_{d_C}\\2\Sigma_{n_q}^{-1}(x_{n_q}-\widehat\mu_{n_q})\\-2\Sigma_{n_q}^{-1}(x_{n_q}-\widehat\mu_{n_q})\\0_{d_\delta}\\0_{d_\delta}\\0_{|\widehat{K}^c|}\\0_{|\widehat{K}^c|}\end{array}\right]. 
\]
Similarly, let 
\[
\widetilde U_\infty=\left[\begin{array}{c}-I_{d_C}\\B'\\-B'\\C'_\infty\\-C'_\infty\\I_{\widehat K^c}\\-I_{\widehat K^c}\end{array}\right] \text{ and }\widetilde b_\infty=\left[\begin{array}{c}0_{d_C}\\2\Sigma_\infty^{-1}(x_\infty-\widehat\mu_\infty)\\-2\Sigma_\infty^{-1}(x_\infty-\widehat\mu_\infty)\\0_{d_\delta}\\0_{d_\delta}\\0_{|\widehat{K}^c|}\\0_{|\widehat{K}^c|}\end{array}\right]. 
\]
Then, $\widehat\psi_n=\underset{\psi\in{\footnotesize\poly}(\widetilde U_q, \widetilde b_q)}{\argmin}\|\psi\|$. 
Denote the number of inequalities by $d_U=2d_\mu+2d_\delta+2|\widehat{K}^c|+d_C$. 

We first show that 
\begin{equation}
\rk(I_J\widetilde{U}_q)=\rk(I_J\widetilde{U}_\infty)\label{multiplier_convergence_rank_condition}
\end{equation}
for every $J\subseteq\{1,...,d_U\}$ that is activatable for $\poly(\widetilde{U}_q, \widetilde{b}_q)$ or $\poly(\widetilde{U}_\infty, \widetilde{b}_\infty)$. 
Since there are three collections of equalities, any such $J$ must contain the final $2d_\mu+2d_\delta+2|\widehat{K}^c|$ elements of $\{1,...,d_U\}$. 
We can therefore write $J$ as $J_0\cup\{d_C+1,...d_U\}$ where $J_0\subseteq\{1,...,d_C\}$. 
Evaluate: 
\begin{align}
\rk(I_J\widetilde{U}_q) = \rk(\widetilde{U}_q'I_J')
&=\rk([-I'_{J_0}, ~B, ~-B,~C_{n_q},~-C_{n_q},~I'_{\widehat{K}^c},~-I'_{\widehat{K}^c}])\nonumber\\
&=\rk([-I'_{J_0}, ~B, ~C_{n_q},~I'_{\widehat{K}^c}])\nonumber\\
&=\rk([-I'_{J_0}, ~B, ~B\Pi_{n_q}+D,~I'_{\widehat{K}^c}])\nonumber\\
&=\rk([-I'_{J_0}, ~B, ~D,~I'_{\widehat{K}^c}])\nonumber\\
&=\rk([-I'_{J_0}, ~B, ~B\Pi_{\infty}+D,~I'_{\widehat{K}^c}])\nonumber\\
&=\rk([-I'_{J_0}, ~B, ~C_{\infty},~I'_{\widehat{K}^c}])\nonumber\\
&=\rk([-I'_{J_0}, ~B, ~-B,~C_{\infty},~-C_{\infty},~I'_{\widehat{K}^c},~-I'_{\widehat{K}^c}])\nonumber\\
&=\rk(\widetilde U_\infty I'_J) = \rk(I_J\widetilde{U}_\infty),\nonumber
\end{align}
where the second and second-to-last equalities hold by the structure of $J$, and the middle six equalities hold by the rank-preservation property of basic matrix column operations. 
This shows (\ref{multiplier_convergence_rank_condition}). 

Note that $\poly(\widetilde{U}_q,\widetilde{b}_q)$ is nonempty because $\widehat\psi_{n_q}\in\poly(\widetilde{U}_q,\widetilde{b}_q)$. 
This, combined with (\ref{multiplier_convergence_rank_condition}) and Lemma \ref{K-convergence_nonempty}, implies that $\poly(\widetilde{U}_\infty,\widetilde{b}_\infty)$ is nonempty. 
Let $\widehat\psi_\infty=\underset{\psi\in{\footnotesize\poly}(\widetilde U_\infty, \widetilde b_\infty)}{\argmin}\|\psi\|$. 
The fact that $\widehat\psi_\infty\in\poly(\widetilde U_\infty, \widetilde b_\infty)$ implies that $\widehat\psi_\infty$ satisfies conditions (\ref{multiplier_lemma_KKT1})-(\ref{multiplier_lemma_KKT3}) with $n=\infty$. 
For condition (\ref{multiplier_lemma_KKT4}), note that $I_{\widehat{K}^c}\widehat \psi_\infty=0$, which implies that $I_{\widehat{K}^c_\infty}\widehat \psi_\infty=0$ because $\widehat{K}^c_\infty\subseteq\widehat{K}^c$ (equivalently, $\widehat{K}\subseteq\widehat{K}_\infty$). 
To see this last point, note that if, for a given $k\in\{1,...,d_C\}$, we have $e'_k(h_{n_q}-B\widehat{\mu}_{n_q}-C_{n_q}\widehat\delta_{n_q})=0$ for all $q$, then $e'_k(h_\infty-B\widehat{\mu}_\infty-C_\infty\widehat{\delta}_\infty)=0$. 
(Recall $\widehat{K}=\widehat{K}_{n_q}$ for all $q$.) 

Finally, it follows from Lemma \ref{K-convergence_nonempty2_new}, using (\ref{multiplier_convergence_rank_condition}) and the non-emptiness of $\poly(\widetilde{U}_q,\widetilde{b}_q)$, that $\poly(\widetilde{U}_q,\widetilde{b}_q)\rightarrow_K\poly(\widetilde{U}_\infty,\widetilde{b}_\infty)$. 
Then, it follows from Lemma \ref{Projection-convergence} and the non-emptiness of $\poly(\widetilde{U}_\infty,\widetilde{b}_\infty)$, that $\widehat\psi_{n_q}\rightarrow \widehat\psi_\infty$. 
\end{proof}

\begin{proof}[Proof of Lemma \textup{\ref{lem:MixedChi2}}] First note that $X$ has the same distribution as $\widetilde{X} = \Sigma^{1/2} Z$, where $\Sigma^{1/2}$ is the symmetric matrix square root of $\Sigma$ and $Z\sim N(\mathbf{0},I)$. Thus, it suffices to show that $\|A\widetilde{X}+b\|^2_\Upsilon$ has a continuous distribution. 

Next, consider the derivation:
\begin{align}
\|A\widetilde X+b\|^2_{\Upsilon} = Z'\Sigma^{1/2}A'\Upsilon A\Sigma^{1/2} Z + 2b'\Upsilon A\Sigma^{1/2}Z + b'\Upsilon b.
\end{align}
Let $SVD$ be the singular value decomposition of $\Upsilon^{1/2}A\Sigma^{1/2}$, where $S$ and $D$ are unitary matrices (i.e. square matrices such that $SS' = I$ and $D D' = I$) and $V$ is a rectangular diagonal matrix whose diagonal elements are singular values of $\Upsilon^{1/2}A\Sigma^{1/2}$. Then
\begin{align}
\|A\widetilde X+b\|^2_{\Upsilon} = Z'D'VV'DZ + 2b'\Upsilon^{1/2} SVDZ + b'\Upsilon b.
\end{align}
Let $\widetilde{Z} = DZ$. Then $\widetilde{Z}\sim N(\mathbf{0},I)$. Let $\widetilde{Z}_j$ be the $j$th element of $\widetilde{Z}$ for $j\in\{1,\dots,\ell\}$. Then
\begin{align}
\|A\widetilde X+b\|^2_{\Upsilon}  = \sum_{j=1}^\ell (v_j^2\widetilde{Z}_j^2 + w_jv_j\widetilde{Z}_j) + b'\Upsilon b,
\end{align}
where $v_j$ is the $j$th diagonal element of $V$ (we let $v_j = 0$ for $j>k$ when $\ell>k$) and $w_j$ is the $j$th element of the vector $2S'\Upsilon^{1/2}b$. The lemma is proved by observing that $\{v_j^2\widetilde{Z}_j^2+w_jv_j\widetilde{Z}_j\}_{j=1}^k$ are mutually independent, and for each $j$, $v_j^2\widetilde{Z}_j^2+w_jv_j\widetilde{Z}_j$ is a continuous random variable unless $v_j=0$. 
\end{proof}

\begin{proof}[Proof of Lemma \textup{\ref{s=0_implies_T=0}}]
The first two KKT conditions for (\ref{generic_CQPP}) are 
\begin{align*}
2\Sigma^{-1}(X-\widehat\mu)&=B'\widehat\psi\\
C'\widehat\psi&=\mathbf{0}. 
\end{align*}
It follows from the definition of $\widehat L$ that $\widehat \psi=I'_{\widehat L}I_{\widehat L}\widehat\psi$. 
Plugging this into the first two KKT conditions, we get that 
\begin{align}
2\Sigma^{-1}(X-\widehat\mu)&=B'_{\widehat{L}}I_{\widehat{L}}\widehat\psi \label{generic1}\\
C'_{\widehat L}I_{\widehat{L}}\widehat\psi&=\mathbf{0}. \label{generic2}
\end{align}
It then follows from (\ref{generic2}) that $I_{\widehat{L}}\widehat\psi=M_{C_{\widehat{L}}}I_{\widehat{L}}\widehat\psi$. 
Plugging this into (\ref{generic1}), we get that 
\[
2\Sigma^{-1}(X-\widehat\mu)=B'_{\widehat{L}}M_{C_{\widehat{L}}}I_{\widehat{L}}\widehat\psi. 
\]
Then note that, by Lemma \ref{lem:MB2}(a), $0=\textup{rk}(M_{C_{\widehat{L}}}B_{\widehat{L}})$. 
Therefore, $\Sigma^{-1}(X-\widehat\mu)=\mathbf{0}$, which implies that $\widehat\mu=X$ and $T=0$. 
\end{proof}

\begin{proof}[Proof of Lemma \textup{\ref{LimitRankNonincreasing}}]
Note that the rank of any matrix is equal to the number of nonzero singular values. The matrix $A$ has $\textup{rk}(A)$ nonzero singular values. 
By Lemma \ref{ContinuousSingularValues}, $A_n$ has at least $\textup{rk}(A)$ nonzero singular values eventually. 
Therefore, $\textup{rk}(A_n)\ge \textup{rk}(A)$ eventually. 
\end{proof}

\begin{proof}[Proof of Lemma \textup{\ref{ContinuousSingularValues}}]
The left singular values of $A_n$ and $A$ are the non-negative square roots of the eigenvalues of $A'_nA_n$ and $A'A$, respectively. 
The result then follows from Theorem 2.4.9.2 of \citeapp{HornJohnson2012}, which implies that the eigenvalues of a Hermitian matrix are continuous in the entries of the matrix. 
\end{proof}

\begin{proof}[Proof of Lemma \textup{\ref{equivalent_representation}}]
The fact that $\widehat\mu_1=\widehat\mu_2$ follows because the feasible sets are the same. 

For the other two results, it is sufficient to prove them when the rows of $A$ are a subset of the rows of $B$. 
This is because we can define $F=[A;B]$ and $g=[c;d]$, where the semicolon denotes vertical concatenation, so that $\textup{poly}(F,g)=\textup{poly}(A,c)=\textup{poly}(B,d)$. 
Then the result with the rows of $A$ being a subset of the rows of $F$ combines with the result applied to $B$ and $F$ to yield $\widehat r_1$ and $\widehat r_2$ are both equal to the rank of the active inequalities for projection onto $\textup{poly}(F,g)$ and $\widehat\beta_1$ and $\widehat\beta_2$ are both equal to the value calculated when projecting onto $\textup{poly}(F,g)$. 

When the rows of $A$ are a subset of the rows of $B$, then $\widehat r_1=\widehat r_2$ follows from the proof of part (c) of Lemma 9 in CS23 (the bottom of page 90 in the supplemental materials). The fact that $\widehat\beta_1=\widehat\beta_2$ follows from Lemma 10 in CS23. 
\end{proof}

\section{Feasibility of the Limit}\label{A1iiiDiscussion}

Assumption \ref{Assumption1}(vi) requires poly$(C_\infty,b_\infty)$ to be nonempty. 
The remark on Assumption \ref{Assumption1} points out a sufficient condition based on a compact set $\Delta$. 
In fact, that condition can be weakened to $\{\delta\in \mathbb{R}^{d_\delta}: C_{F}\delta\leq b_{F}\}\cap\Delta\neq\emptyset$ for all $F\in{\cal F}_{n0}$. 
The argument for sufficiency is the same. 
Namely, for any sequence $F_{n_q}\in{\cal F}_{n_q0}$, there exists a $\delta_{F_{n_q}}\in\Delta$ such that  $B\mu_{F_{n_q}}+(B\Pi_{F_{n_q}}+D)\delta_{F_{n_q}}\leq d_{n_q}$. This sequence has a subsequence that converges to some limit $\delta_\infty\in\Delta$ that satisfies $B\mu_{\infty}+(B\Pi_\infty+D)\delta_\infty\leq d_\infty$. 
The remark on Assumption \ref{Assumption1} also claims that Assumption \ref{Assumption1}(vi) follows from Assumption \ref{Assumption1}(i) and a strengthened version of Assumption \ref{assu:rank:simple}.  
This section states Lemma \ref{K-convergence_nonempty}, which formalizes this claim. 

\begin{lemma}\label{K-convergence_nonempty} Consider a sequence of $d_C\times d_\delta$ matrices  $\{C_n\}$ and $d_C$-dimensional vectors $\{b_n\}$. 
Suppose 
\begin{itemize}
\item[\textup{(i)}] $C_n\to C_\infty$ and $b_n\to b_\infty$ as $n\to\infty$ for some 
$C_\infty\in\R^{d_C\times d_\delta}$ and $b_\infty\in\R^{d_C}$, 
\item[\textup{(ii)}] \textup{poly}$(C_n,b_n)\neq \emptyset$ eventually, and
\item[\textup{(iii)}] for any $K\subseteq\{1,\dots,d_C\}$ that is activatable for \textup{poly}$(C_n,b_n)$ infinitely often, we have $\textup{rk}(I_KC_n) = \textup{rk}(I_KC_\infty)$ eventually. 
\end{itemize}
Then, \textup{poly}$(C_\infty,b_\infty)\neq \emptyset$. 
\end{lemma}

\noindent\textbf{Remark:} {\it Lemma \textup{\ref{K-convergence_nonempty}} verifies Assumption \textup{1(vi)} under Assumption \textup{\ref{Assumption1}(i)} and a strengthened version of Assumption \textup{\ref{assu:rank:simple}} where we require equation \textup{(\ref{rankst})} to hold also for every $K$ such that $K^=\subseteq K$ and $K\in \cup_{q=1}^{\infty}{\cal A}(C_{F_{n_q}}, b_{F_{n_q}})$ eventually.
This is  because, by Assumption \textup{\ref{Assumption1}(i)}, $C_{F_{n_q}}=D+B\Pi_{F_{n_q}}\rightarrow D+B\Pi_\infty=C_\infty$ and $b_{F_{n_q}}=d_{n_q}-B\mu_{F_{n_q}}\rightarrow d_\infty-B\mu_\infty=b_\infty$. 
Also, the fact that $F_{n}\in\mathcal{F}_{n0}$ implies that $\poly(C_{F_n},b_{F_n})$ is nonempty. 
Finally, any $K\subseteq\{1,...,d_C\}$ that is activatable for $\poly(C_{F_{n_q}},b_{F_{n_q}})$ infinitely often must include $K^=$ and must belong to $\cup_{q=1}^\infty\mathcal{A}(C_{F_{n_q}},b_{F_{n_q}})$. 
Therefore, all conditions of the lemma are satisfied, and Assumption \textup{\ref{Assumption1}(vi)} is verified.}

\begin{proof}[Proof of Lemma \textup{\ref{K-convergence_nonempty}}] 
This proof uses Lemma \ref{MP_convergence}, stated below. 
By condition (ii), $\poly(C_n,b_n)$ is a non-empty closed set eventually. 
Thus, the following argmin is well-defined for large enough $n$:
\begin{equation}
\hat x_n=\underset{x\in\mathbb{R}^{d_\delta}}{\argmin}\|x\|^2 \text{ s.t. } C_nx\le b_n. \label{KKTmin1}
\end{equation}
Let $K_n = \{j\in\{1,\dots,d_C\}: e_j'(C_n\hat x_n-b_n)=0\}$. 
Since $K_n$ can take at most $2^{d_C}$ distinct values, there exists a subsequence $\{n_q\}$ such that $K_{n_q}$ does not depend on $q$ and thus can be simply denoted by $K$. 
It then follows from Lemma \ref{lem:KKTB2_new} (with no nuisance parameters) that 
\begin{equation}
\hat x_{n_q}=\left(I_KC_{n_q}\right)^+I_Kb_{n_q} 
\end{equation}
for all $q\in\N$. 

Note that $K$ is an active set of inequalities for $\poly(C_n,b_n)$ infinitely often, and therefore by condition (iii), $\rk(I_KC_{n_q})=\rk(I_K C_\infty)$ eventually as $q\rightarrow\infty$. 
It follows from Lemma \ref{MP_convergence} that $\left(I_KC_{n_q}\right)^+\rightarrow \left(I_KC_\infty\right)^+$ as $q\rightarrow\infty$. 
Let $\hat x_\infty=\left(I_KC_\infty\right)^+I_K b_\infty$ and note that $\hat x_{n_q}\rightarrow \hat x_\infty$. 
To complete the proof, note that 
\[
C_\infty \hat x_\infty-b_\infty = \lim_{q\to\infty}(C_{n_q}\hat x_{n_q}-b_{n_q})\leq 0.
\]
Thus, $\hat x_\infty\in\poly(C_\infty,b_\infty)$, which shows that $\poly(C_\infty,b_\infty)\neq\emptyset$. 
\end{proof}

\begin{lemma}\label{MP_convergence}
Let $A_n$ be a sequence of matrices converging to $A_\infty$ such that $\rk(A_n)=\rk(A_\infty)$ eventually. Then, $A_n^+\rightarrow A^+_\infty$. 
\end{lemma}
\begin{proof}[Proof of Lemma \ref{MP_convergence}]
Fix an arbitrary subsequence, $n_m$. It is sufficient to show that there exists a further subsequence along which the convergence in the conclusion holds. 

Let $A_n=U_n\Sigma_nV'_n$ be a singular value decomposition of $A_n$, where $U_n$ and $V_n$ are orthonormal matrices and $\Sigma_n$ is a (possibly non-square) diagonal matrix with diagonal elements equal to the singular values of $A_n$ in nonincreasing order. 
Let $\Sigma_\infty$ be the (possibly non-square) diagonal matrix with diagonal elements equal to the singular values of $A_\infty$. 
Let $U_\infty$ and $V_\infty$ be orthonormal matrices and let $n_q$ be a further subsequence of $n_m$ such that $U_{n_q}\rightarrow U_\infty$, $V_{n_q}\rightarrow V_\infty$, and $r=\rk(A_{n_q})=\rk(A_\infty)$ does not depend on $q$. 
(These exist because the set of orthonormal matrices is compact.) 
It follows from Lemma \ref{ContinuousSingularValues} that $\Sigma_{n_q}\rightarrow \Sigma_\infty$. 
It follows from convergence that $A_\infty=U_\infty\Sigma_\infty V'_\infty$ is a singular value decomposition of $A_\infty$. 

The fact that $\Sigma_{n_q}$ and $\Sigma_\infty$ both have $r$ nonzero diagonal elements, together with the fact that $\Sigma_{n_q}\rightarrow \Sigma_\infty$, implies that $\Sigma_{n_q}^+\rightarrow \Sigma_\infty^+$. 
(The Moore-Penrose pseudo-inverse of a possibly non-square diagonal matrix is a diagonal matrix of the same dimension with diagonal elements equal to the pseudo-inverse of each diagonal element of the original matrix.) 
Therefore, $A_{n_q}^+=U_{n_q}\Sigma_{n_q}^+V'_{n_q}\rightarrow U_\infty \Sigma_\infty^+ V'_\infty=A_\infty^+$. 
(In general, the Moore-Penrose pseudo-inverse of a matrix $A$ with singular value decomposition $U\Sigma V'$ is $A^+=U\Sigma^+ V'$; see, for example, Chapter 6 in \citeapp{Ben-IsraelGreville2003}.) 
\end{proof}

\section{Inference on Policy Relevant Treatment Effects}\label{app:MTE}

In this section, we demonstrate the GCC test in a simulation of Example \ref{ex:PRTE}. 
The simulations show that the GCC test is fast to compute and has good size and power. 

\subsection{The Data Generating Process} \label{E.1}
We follow Section 5 in \citeapp{MogstadSantosTorgovitsky2018}. 
We generate an i.i.d.\ sample of $\{Y_i,D_i,Z_i\}_{i=1}^n$ according to the following distribution. 
Suppose $Y$ is binary and there are no exogenous covariates $X$. 
Let $Z$ be independent of $(Y_0, Y_1)$ with support $\{0,1,2\}$ and distribution given by $P(Z=0)=0.5$, $P(Z=1) = 0.4$, and $P(Z=2) = 0.1$. 
Also, let $D = 1\{p(Z)\geq U\}$, where $U|Z\sim \textup{Uniform}[0,1]$ and $p(z)$ is the propensity score, defined by $p(0)=0.35$, $p(1) = 0.6$, and $p(2) = 0.7$. Let $Y_d|U=u\sim \textup{Bernoulli}(\kappa_d(u))$, where the marginal treatment response functions are
\begin{align}
	\kappa_0(u) &= 0.6\varphi_0^2(u)+0.4\varphi_1^2(u)+0.3\varphi_2^2(u),\nonumber\\
	\kappa_1(u) &= 0.75\varphi_0^2(u)+0.5\varphi_1^2(u)+0.25\varphi_2^2(u),\label{m0m1}
\end{align}
and $\varphi_0^2(u), \varphi_1^2(u), \varphi_2^2(u)$ are the Bernstein basis polynomials of degree $2$:
$$
\varphi_k^2(u) = {2 \choose k} u^k(1-u)^{2-k}.
$$
These marginal treatment response functions are the same as those used in MST18 and are depicted in Figure 1 in that paper. 

\subsection{Local Average Treatment Effects}

For $0\le\alpha_1<\alpha_2\le 1$, let 
\begin{equation}
\textup{LATE}(\alpha_1,\alpha_2)=\mathbb{E}[Y_1-Y_0|\alpha_1<U<\alpha_2]. 
\end{equation}
As shown in \citeapp{ImbensAngrist1994}, the three point-identified LATEs for the DGP specified in Section \ref{E.1} are: LATE(0.35,0.6), LATE(0.35,0.7), and LATE(0.6,0.7). Other LATEs are generally not point identified. 
We take our (policy-relevant) parameter of interest to be $\theta=$LATE(p(0),0.9). 
Recall the true value of $p(0)$ is $0.35$. 
This parameter is closely related to the instrument that identifies the LATE(0.35,0.7) parameter. 
It has the interpretation of a policy change that includes the same compliers as that instrument plus additional compliers with values of $u$ up to $0.9$. 
(Think of the instrument as being a past policy change, while the policy change of interest is the same but with more compliers.) 
The true value of $\theta$ for the DGP specified in Section \ref{E.1} is approximately 0.046. 
This parameter of interest is not point identified. However, we use the bounds in MST18 to calculate a confidence interval for $\theta$. 

\subsection{Marginal Treatment Response Parameterization}
We define a set of MTRs to be a set of functions of $u\in[0,1]$ that are bounded between $0$ and $1$ and constant on the intervals $\mathcal{U}_1=[0,p(0))$, $\mathcal{U}_2=[p(0),p(1))$, $\mathcal{U}_3=[p(1),p(2))$, $\mathcal{U}_4=[p(2),0.9)$, and $\mathcal{U}_5=[0.9,1]$. 
The endpoints of these intervals are taken from the range of the propensity score function, together with the boundary of the LATE that we are interested in. 
By Proposition 4 in MST18, this parameterization of MTRs is sufficient to achieve tight nonparametric bounds on $\theta$. 
For $d\in\{0,1\}$, let $\delta_d=(\delta_{d1},\delta_{d2},\delta_{d3},\delta_{d4},\delta_{d5})'\in[0,1]^5$ be a vector of coefficients such that $\delta_{dj}$ is the value of the function over $\mathcal{U}_j$ for $j\in\{1,...,5\}$. 
Let $\delta=(\delta'_0, \delta'_1)'$. 
We enforce these shape restrictions on $\delta$ by setting 
\begin{equation}
	A=
	\begin{pmatrix}
		-I_{10}\\
		I_{10}
	\end{pmatrix}
	\qquad
	\text{and}
	\qquad
	b=
	\begin{pmatrix}
		\mathbf{0}_{10\times 1}\\                             
		\mathbf{1}_{10\times 1}
	\end{pmatrix}. 
\end{equation}

\subsection{IV-Like Estimands}
\label{sec:ivlike}
Any measurable function of $D$ and $Z$ could be an IV-like estimand. 
We consider a variety of subsets of the following IV-like estimands: 
\begin{align}
	\text{IV}:~& s_{IV}(D,Z) = Z - \E[Z]\nonumber\\
	\text{OLS}:~& s_{OLS}(D,Z) = D-\E[D]\nonumber\\
	\textup{IV}_{Z0}:~&s_{Z0}(D,Z) = 1\{Z=0\}\nonumber\\
	\textup{IV}_{Z1}:~&s_{Z1}(D,Z) = 1\{Z=1\}\nonumber\\
	\textup{IV}_{Z2}:~&s_{Z2}(D,Z) = 1\{Z=2\}\nonumber\\
	\textup{IV}_{DZ10}:~&s_{DZ10}(D,Z) = D\,1\{Z=0\} \nonumber\\
	\textup{IV}_{DZ11}:~&s_{DZ11}(D,Z) = D\,1\{Z=1\} \nonumber\\
	\textup{IV}_{DZ12}:~&s_{DZ12}(D,Z) = D\,1\{Z=2\} \nonumber\\
	\textup{IV}_{DZ00}:~&s_{DZ00}(D,Z) = (1-D)\,1\{Z=0\} \nonumber\\
	\textup{IV}_{DZ01}:~&s_{DZ01}(D,Z) = (1-D)\,1\{Z=1\} \nonumber\\
	\textup{IV}_{DZ02}:~&s_{DZ02}(D,Z) = (1-D)\,1\{Z=2\}. \label{IVlike}
\end{align}
Any $\mathcal{S}\subseteq \overline{\mathcal{S}}:=\{\text{IV}, \text{OLS},\text{Z0},\text{Z1},\text{Z2},\text{DZ10},\text{DZ11},\text{DZ12},\text{DZ00},\text{DZ01},\text{DZ02}\}$ defines a collection of IV-like estimands that defines a vector $m$ with $|\mathcal{S}|$ elements. 

\subsection{Defining and Estimating $m$, $\gamma$, and $\Gamma$}
For any $x\in\overline{\mathcal{S}}$, let $m_x=\mathbb{E}[s_x(D,Z)Y]$. 
For any $\mathcal{S}\subseteq\overline{\mathcal{S}}$, we can take $m$ to be the vector of $m_x$ for $x\in\mathcal{S}$. 
Each $m_x$ can be estimated by $\overline{m}_x=\esum s_x(D_i,Z_i)Y_i$.\footnote{\label{estimate_s_IV_OLS}For $x\in\{\text{IV},\text{OLS}\}$, $s_x(\cdot)$ is estimated by replacing $\mathbb{E}[Z]$ and $\mathbb{E}[D]$ with $\esum Z_i$ and $\esum D_i$, respectively.} 

To define $\gamma$, we use the formula for the weights associated with LATE(p(0),0.9) from Table I in MST18 combined with the piecewise constant basis functions. 
For $d\in\{0,1\}$ and $j\in\{1,...,5\}$, let 
\begin{equation}
\gamma_{dj}=\frac{(-1)^{d+1}}{|0.9-p(0)|}\int \mathds{1}\{u\in\mathcal{U}_j\cap[p(0),0.9]\}du. \label{gamma_def}
\end{equation}
Let $\gamma=(\gamma'_0,\gamma'_1)'$, where $\gamma_d=(\gamma_{d1},\gamma_{d2},\gamma_{d3},\gamma_{d4},\gamma_{d5})'$ for $d\in\{0,1\}$. 
Note that $\gamma$ needs to be estimated because it depends on $p(0)$. 
It also depends on $\mathcal{U}_j$, which need to be estimated.\footnote{In MST18, the basis that parameterizes the MTRs is taken as given and known throughout the paper. However, a specific basis is required for the bounds to be equal to the endpoints of the identified set---one that depends on the unknown propensity scores; see Proposition 4 in MST18. This basis must be estimated in practice.} 
We estimate the propensity scores by $(\widehat{p}(0),\widehat{p}(1),\widehat{p}(2))=\esum(\mathds{1}\{Z_i=0\},\mathds{1}\{Z_i=1\},\mathds{1}\{Z_i=2\})$. 
We then take $\widehat{\mathcal{U}}_j$ to be estimated versions of $\mathcal{U}_j$ with the estimated propensity scores plugged in. 
Similarly, we estimate $\gamma_{dj}$ and $\gamma$ by $\widehat{\gamma}_{dj}$ and $\widehat{\gamma}$ following (\ref{gamma_def}) with $\widehat{p}(0)$ and $\widehat{\mathcal{U}}_j$ plugged in. 

To define $\Gamma$, we use the formula for the weights from Proposition 1 in MST18 combined with the piecewise constant basis functions. 
For $x\in\overline{\mathcal{S}}$, let $\Gamma_x$ denote the row of $\Gamma$ associated with $s_x(D,Z)$. 
We can define $\Gamma_x=(\Gamma'_{x0},\Gamma'_{x1})$, where $\Gamma_{xd}=(\Gamma_{xd1},...,\Gamma_{xd5})'$ and 
\begin{equation}
\Gamma_{xdj}=\begin{cases}
\mathbb{E}\int s_x(0,Z)\mathds{1}\{u\in\mathcal{U}_j\}\mathds{1}\{u>p(Z)\}du&\text{ if }d=0\\
\mathbb{E}\int s_x(1,Z)\mathds{1}\{u\in\mathcal{U}_j\}\mathds{1}\{u\le p(Z)\}du&\text{ if }d=1
\end{cases}, 
\end{equation}
for $d\in\{0,1\}$, $j\in\{1,...,5\}$, and $x\in\overline{\mathcal{S}}$. 
We can estimate $\Gamma$ by 
\begin{equation}
\widehat{\Gamma}_{xdj}=\begin{cases}
\esum\int s_x(0,Z_i)\mathds{1}\{u\in\widehat{\mathcal{U}}_j\}\mathds{1}\{u>\widehat{p}(Z_i)\}du&\text{ if }d=0\\
\esum\int s_x(1,Z_i)\mathds{1}\{u\in\widehat{\mathcal{U}}_j\}\mathds{1}\{u\le \widehat{p}(Z_i)\}du&\text{ if }d=1
\end{cases}, 
\end{equation}
where, for $x\in\{\text{IV},\text{OLS}\}$, we follow footnote \ref{estimate_s_IV_OLS} for estimating $s_x(\cdot,\cdot)$. 

We estimate the variance-covariance matrix of the estimators of $\gamma$, $m$, and $\Gamma$ by the bootstrap with $B=1000$ bootstrap draws. 

\subsection{Results}
We implement the GCC and RGCC tests using (\ref{LP1}) to write $\Gamma$, $m$, $\gamma$, $\theta$, $A$, and $b$ in terms of $B$, $\mu$, $\Pi$, $D$, and $d$.\footnote{Because $\gamma$ is estimated, we follow the strategy in Remark (2) above Example \ref{ex:npiv} in Section \ref{sub:bound_lp}.} 
We can similarly define $\overline{\mu}_n$, $\overline{\Pi}_n$, and $\overline{\Omega}_n$ using the estimators of $\gamma$, $m$, $\Gamma$, and the estimator of their variance-covariance matrix. 
For each of $5000$ simulations, we calculate the confidence interval implied by the GCC and RGCC tests using bisection. 
We then calculate the frequency with which any given value of $\theta$ in $[-1,1]$ lies outside the confidence interval. 

We consider a variety of choices of $\mathcal{S}\subseteq\overline{\mathcal{S}}$. 
Figure \ref{Fig:mte1} reports the results for only the IV estimand: $\mathcal{S}_1=\{\text{IV}\}$. 
Figure \ref{Fig:mte2} adds the OLS estimand: $\mathcal{S}_2=\{\text{IV},\text{OLS}\}$. 
Figure \ref{Fig:mte3} breaks the instrument into three components and keeps the OLS estimand: $\mathcal{S}_3=\{\text{Z0},\text{Z1},\text{Z2},\text{OLS}\}$. 
Figure \ref{Fig:mte4} saturates the support of $(D,Z)$: $\mathcal{S}_4=\{\text{DZ10},\text{DZ11},\text{DZ12},\text{DZ00},\text{DZ01},\text{DZ02}\}$. 
Note that $\mathcal{S}_4$ gives the tightest bounds on $\theta$. 
The identified sets for $\theta$ are also depicted in Figures \ref{Fig:mte1}-\ref{Fig:mte4} by the shaded region.\footnote{For a given $\mathcal{S}\subseteq\overline{\mathcal{S}}$, the identified set for $\theta$ can be calculated by solving two LPPs: $\theta_{\min}=\min_{\delta: A\delta\le b, \Gamma\delta=m}\gamma'\delta$ and $\theta_{\max}=\max_{\delta: A\delta\le b, \Gamma\delta=m}\gamma'\delta$ for the $\Gamma$ and $m$ associated with $\mathcal{S}$.} 
In the legend, the number in square brackets indicates the computational time (in seconds) to compute the conﬁdence interval once, taking the median over simulations. 

\noindent\textbf{Remarks:} 
(1) \textit{In each figure, both GCC and RGCC tests have well controlled null rejection rates and reasonable power outside the identified set. 
We can see that the identified sets shrink as we add more IV-like estimands, as expected.  It is encouraging to see that the power curves get steeper with more IV-like estimands as well, indicating that our tests can effectively capture the identification power of the additional IV-like estimators despite the added noise. 
The RGCC test appears to be size-exact on the boundary of the identified sets. Also, the difference between the GCC and RGCC tests gets smaller as the number of equalities increases.} 

(2) \textit{It is interesting to note that the computational time of the RGCC does not change monotonically as we move from Figure \ref{Fig:mte1} to Figure \ref{Fig:mte4}. It takes the longest in Figure \ref{Fig:mte3}. The nonmonotonicity results from the way we implement the refinement. Specifically, we do not implement the vertex enumeration step unless $\hat{s}_n = 1$. The event $\hat{s}_n = 1$ may occur less frequently as the number of inequalities/equalities increases.} 

\newcommand{\centerparameter}{3.2}

\begin{figure}
	\pgfplotsset{width=0.5*\textwidth}
	\begin{center}
		\begin{tikzpicture}
			\begin{axis}
				[ymin=0,ymax=1,xmin=-1,xmax=1,legend style={at={(axis cs:-0.52,0.985)},anchor=north west}]
				\draw [pattern=north west lines, pattern color=gray!70,draw=none] (-0.4209,0) rectangle (0.5003,0.996);
				\draw[dotted] (-1,0.05) -- (1,0.05);
				\addplot[very thick] table [x=theta, y=rr_GCC, col sep=semicolon]
				{MTE_Rej_n_eqs500_Sim_eqs5000_CI1.txt};
				\addplot[dashed, very thick, gray] table [x=theta, y=rr_RGCC, col sep=semicolon]
				{MTE_Rej_n_eqs500_Sim_eqs5000_CI1.txt};
				\addlegendentry{{\scriptsize GCC test  [0.93s]}};
				\addlegendentry{{\scriptsize RGCC test [1.22s]}};
			\end{axis}
			\node[align=center, below] at (\centerparameter,-.6) {(a) $n=500$};
		\end{tikzpicture}
		\begin{tikzpicture}
			\begin{axis}
				[ymin=0,ymax=1,xmin=-1,xmax=1,legend style={at={(axis cs:-0.52,0.985)},anchor=north west}]
				\draw [pattern=north west lines, pattern color=gray!70,draw=none] (-0.4209,0) rectangle (0.5003,0.996);
				\draw[dotted] (-1,0.05) -- (1,0.05);
				\addplot[very thick] table [x=theta, y=rr_GCC, col sep=semicolon]
				{MTE_Rej_n_eqs1000_Sim_eqs5000_CI1.txt};
				\addplot[dashed, very thick, gray] table [x=theta, y=rr_RGCC, col sep=semicolon]
				{MTE_Rej_n_eqs1000_Sim_eqs5000_CI1.txt};
				\addlegendentry{{\scriptsize GCC test  [0.93s]}};
				\addlegendentry{{\scriptsize RGCC test [1.18s]}};
			\end{axis}
			\node[align=center, below] at (\centerparameter,-.6) {(b) $n=1000$};
		\end{tikzpicture}
	\end{center}
\caption{Power Curve with only the IV Slope Coefficient ($\mathcal{S}_1$)}
\label{Fig:mte1}
\end{figure}
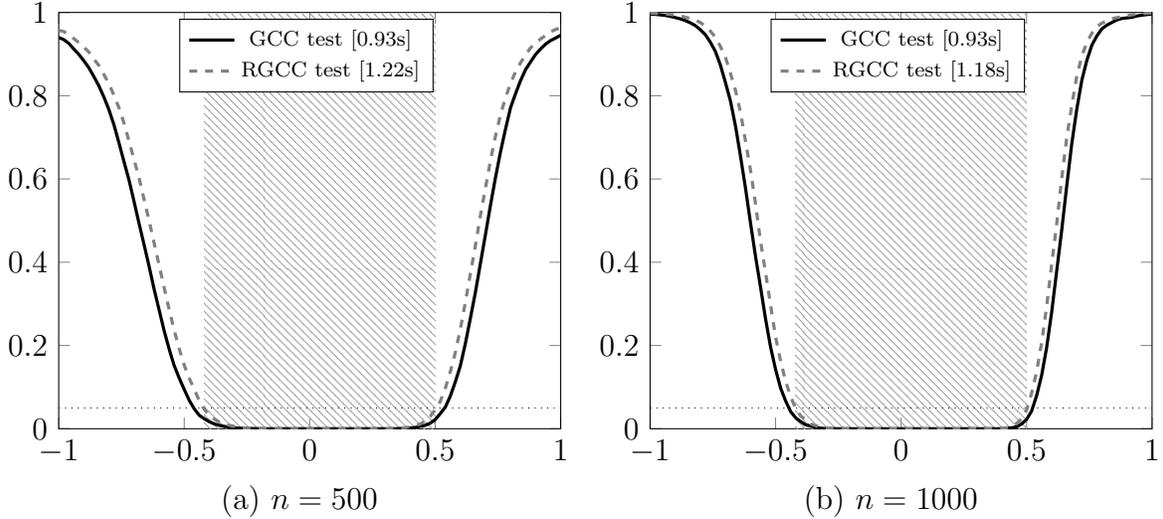

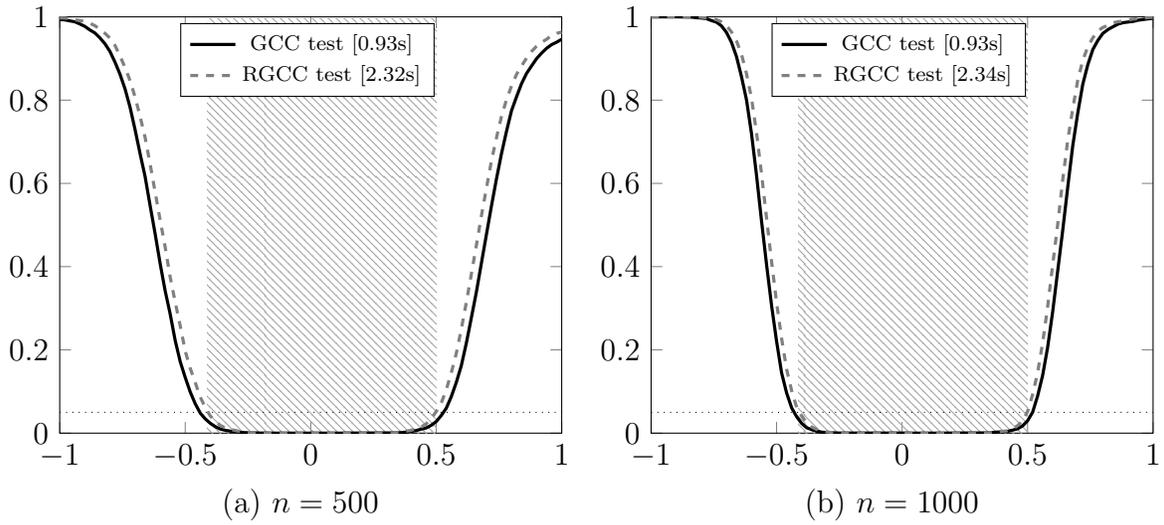
\begin{figure}
	\pgfplotsset{width=0.5*\textwidth}
	\begin{center}
		\begin{tikzpicture}
			\begin{axis}
				[ymin=0,ymax=1,xmin=-1,xmax=1,legend style={at={(axis cs:-0.52,0.985)},anchor=north west}]
				\draw [pattern=north west lines, pattern color=gray!70,draw=none] (-0.4112,0) rectangle (0.5003,0.996);
				\draw[dotted] (-1,0.05) -- (1,0.05);
				\addplot[very thick] table [x=theta, y=rr_GCC, col sep=semicolon]
				{MTE_Rej_n_eqs500_Sim_eqs5000_CI2.txt};
				\addplot[dashed, very thick, gray] table [x=theta, y=rr_RGCC, col sep=semicolon]
				{MTE_Rej_n_eqs500_Sim_eqs5000_CI2.txt};
				\addlegendentry{{\scriptsize GCC test  [0.93s]}};
				\addlegendentry{{\scriptsize RGCC test [2.32s]}};
			\end{axis}
			\node[align=center, below] at (\centerparameter,-.6) {(a) $n=500$};
		\end{tikzpicture}
		\begin{tikzpicture}
			\begin{axis}
				[ymin=0,ymax=1,xmin=-1,xmax=1,legend style={at={(axis cs:-0.52,0.985)},anchor=north west}]
				\draw [pattern=north west lines, pattern color=gray!70,draw=none] (-0.4112,0) rectangle (0.5003,0.996);
				\draw[dotted] (-1,0.05) -- (1,0.05);
				\addplot[very thick] table [x=theta, y=rr_GCC, col sep=semicolon]
				{MTE_Rej_n_eqs1000_Sim_eqs5000_CI2.txt};
				\addplot[dashed, very thick, gray] table [x=theta, y=rr_RGCC, col sep=semicolon]
				{MTE_Rej_n_eqs1000_Sim_eqs5000_CI2.txt};
				\addlegendentry{{\scriptsize GCC test  [0.93s]}};
				\addlegendentry{{\scriptsize RGCC test [2.34s]}};
			\end{axis}
			\node[align=center, below] at (\centerparameter,-.6) {(b) $n=1000$};
		\end{tikzpicture}
	\end{center}
\caption{Power Curve with Both the IV and OLS Slope Coefficients ($\mathcal{S}_2$)}
\label{Fig:mte2}
\end{figure}

\begin{figure}
	\pgfplotsset{width=0.5*\textwidth}
	\begin{center}
		\begin{tikzpicture}
			\begin{axis}
				[ymin=0,ymax=1,xmin=-1,xmax=1,legend style={at={(axis cs:-0.52,0.985)},anchor=north west}]
				\draw [pattern=north west lines, pattern color=gray!70,draw=none] (-0.2988,0) rectangle (0.4075,0.996);
				\draw[dotted] (-1,0.05) -- (1,0.05);
				\addplot[very thick] table [x=theta, y=rr_GCC, col sep=semicolon]
				{MTE_Rej_n_eqs500_Sim_eqs5000_CI3.txt};
				\addplot[dashed, very thick, gray] table [x=theta, y=rr_RGCC, col sep=semicolon]
				{MTE_Rej_n_eqs500_Sim_eqs5000_CI3.txt};
				\addlegendentry{{\scriptsize GCC test  [0.90s]}};
				\addlegendentry{{\scriptsize RGCC test [7.66s]}};
			\end{axis}
			\node[align=center, below] at (\centerparameter,-.6) {(a) $n=500$};
		\end{tikzpicture}
		\begin{tikzpicture}
			\begin{axis}
				[ymin=0,ymax=1,xmin=-1,xmax=1,legend style={at={(axis cs:-0.42,0.985)},anchor=north west}]
				\draw [pattern=north west lines, pattern color=gray!70,draw=none] (-0.2988,0) rectangle (0.4075,0.996);
				\draw[dotted] (-1,0.05) -- (1,0.05);
				\addplot[very thick] table [x=theta, y=rr_GCC, col sep=semicolon]
				{MTE_Rej_n_eqs1000_Sim_eqs5000_CI3.txt};
				\addplot[dashed, very thick, gray] table [x=theta, y=rr_RGCC, col sep=semicolon]
				{MTE_Rej_n_eqs1000_Sim_eqs5000_CI3.txt};
				\addlegendentry{{\scriptsize GCC test  [0.89s]}};
				\addlegendentry{{\scriptsize RGCC test [13.85s]}};
			\end{axis}
			\node[align=center, below] at (\centerparameter,-.6) {(b) $n=1000$};
		\end{tikzpicture}
	\end{center}
\caption{Power Curve with Breaking the IV Slope into Three Components and OLS Slope Coefficients ($\mathcal{S}_3$)}
\label{Fig:mte3}
\end{figure}
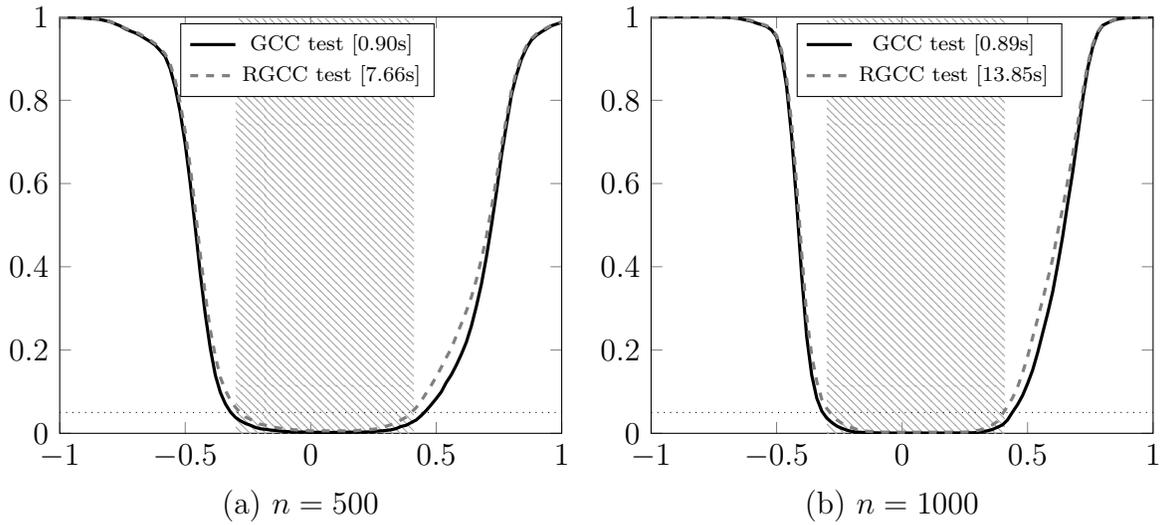

\begin{figure}
	\pgfplotsset{width=0.5*\textwidth}
	\begin{center}
		\begin{tikzpicture}
			\begin{axis}
				[ymin=0,ymax=1,xmin=-1,xmax=1,legend style={at={(axis cs:-0.52,0.985)},anchor=north west}]
				\draw [pattern=north west lines, pattern color=gray!70,draw=none] (-0.1378,0) rectangle (0.4075,0.996);
				\draw[dotted] (-1,0.05) -- (1,0.05);
				\addplot[very thick] table [x=theta, y=rr_GCC, col sep=semicolon]
				{MTE_Rej_n_eqs500_Sim_eqs5000_CI4.txt};
				\addplot[dashed, very thick, gray] table [x=theta, y=rr_RGCC, col sep=semicolon]
				{MTE_Rej_n_eqs500_Sim_eqs5000_CI4.txt};
				\addlegendentry{{\scriptsize GCC test  [0.93s]}};
				\addlegendentry{{\scriptsize RGCC test [1.03s]}};
			\end{axis}
			\node[align=center, below] at (\centerparameter,-.6) {(a) $n=500$};
		\end{tikzpicture}
		\begin{tikzpicture}
			\begin{axis}
				[ymin=0,ymax=1,xmin=-1,xmax=1,legend style={at={(axis cs:-0.38,0.985)},anchor=north west}]
				\draw [pattern=north west lines, pattern color=gray!70,draw=none] (-0.1378,0) rectangle (0.4075,0.996);
				\draw[dotted] (-1,0.05) -- (1,0.05);
				\addplot[very thick] table [x=theta, y=rr_GCC, col sep=semicolon]
				{MTE_Rej_n_eqs1000_Sim_eqs5000_CI4.txt};
				\addplot[dashed, very thick, gray] table [x=theta, y=rr_RGCC, col sep=semicolon]
				{MTE_Rej_n_eqs1000_Sim_eqs5000_CI4.txt};
				\addlegendentry{{\scriptsize GCC test  [0.91s]}};
				\addlegendentry{{\scriptsize RGCC test [1.05s]}};
			\end{axis}
			\node[align=center, below] at (\centerparameter,-.6) {(b) $n=1000$};
		\end{tikzpicture}
	\end{center}
\caption{Power Curve with All Six IV-like Estimands ($\mathcal{S}_4$)}
\label{Fig:mte4}
\end{figure}
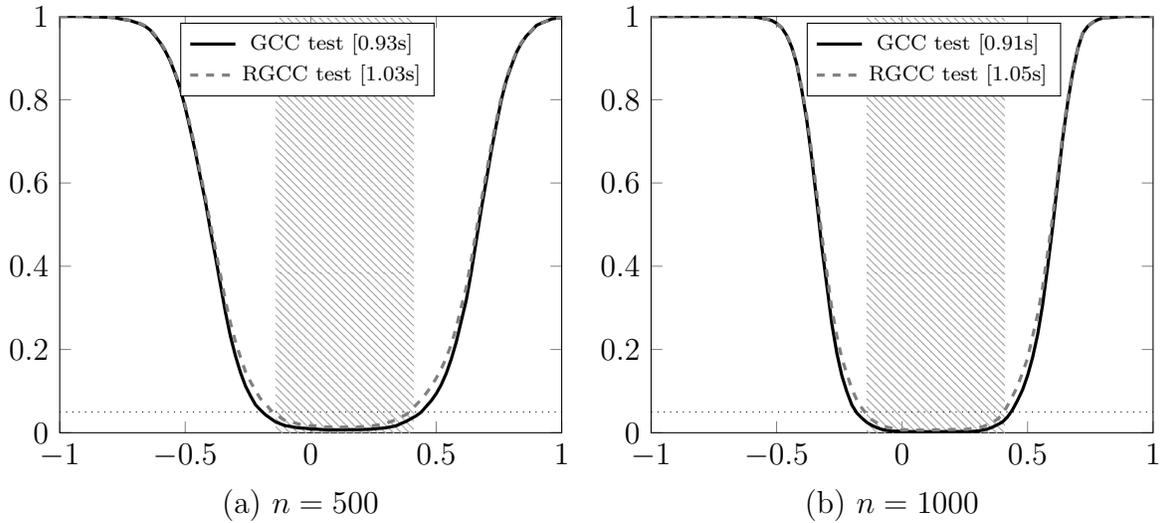

\bibliographystyleapp{apalike}
\bibliographyapp{references}

\end{document}